\documentclass[preprint,11pt]{imsart}
\RequirePackage[T1]{fontenc}
\usepackage{color}
\usepackage[latin1]{inputenc}
\usepackage[T1]{fontenc}
\usepackage{natbib}
\usepackage{apalike}
\usepackage{indic}

\usepackage{dsfont}
\RequirePackage{amsfonts,amsthm,amsmath}
\RequirePackage[colorlinks,citecolor=blue,urlcolor=blue]{hyperref}
\usepackage{a4wide}
\newtheorem{thm}{Theorem}
\newtheorem{prop}[thm]{Proposition}

\newtheorem{lemme}{Lemma}
\newtheorem{rmq}{Remark}
 \usepackage[latin1]{inputenc}
 \usepackage{amsmath}
 \usepackage{amsfonts}
 \usepackage{enumerate}
  \usepackage{multirow}
 \usepackage{graphicx}
\startlocaldefs
\theoremstyle{plain}
\endlocaldefs
\usepackage{sansmath}
\sansmath

\hypersetup{pdfstartview=FitH}
\ifpdf%
    \graphicspath{{PDFFIG/}}
\else
    \graphicspath{{EPSFIG/}}
\fi%

\begin{document}
%\doublespacing
\begin{frontmatter}
\title{Diagnostic checking in FARIMA models with uncorrelated but non-independent error terms\\[0.5cm]
}
%\title{Modified portmanteau tests for ARMA models with uncorrelated but non-independent error
%terms : a self-normalisation approach\\[0.5cm]
%}
\runtitle{Validation of weak FARIMA models}
%\thankstext{T1}{Merci tout le monde!}
%
\begin{aug}
\author{\fnms{\Large{Yacouba}}
\snm{\Large{Boubacar Ma\"{\i}nassara}}
\ead[label=e1]{yacouba.boubacar$\_$mainassara@univ-fcomte.fr}}\\[2mm]
\author{\fnms{\Large{Youssef}}
\snm{\Large{Esstafa}}
\ead[label=e2]{youssef.esstafa@univ-fcomte.fr}}\\[2mm]
\author{\fnms{\Large{Bruno}}
\snm{\Large{Saussereau}}
\ead[label=e3]{bruno.saussereau@univ-fcomte.fr}}
\runauthor{Y. Boubacar Ma\"{\i}nassara, Y. Esstafa and B. Saussereau}
\affiliation{Universit\'e Bourgogne Franche-Comt\'e}

\address{\hspace*{0cm}\\
Universit\'e Bourgogne Franche-Comt\'e, \\
Laboratoire de math\'{e}matiques de Besan\c{c}on, \\ UMR CNRS 6623, \\
16 route de Gray, \\ 25030 Besan\c{c}on, France.\\[0.2cm]
%\hspace*{0cm} \\
\printead{e1}\\
\printead{e2}\\
\printead{e3}}
\end{aug}
\vspace{0.5cm}
%
%
%%%%  For blind version
%\begin{aug}
%\author{....}
%\runauthor{...}
%\affiliation{...}
%
%\address{...}
%\end{aug}
%\vspace{0.5cm}
%
\begin{abstract} This work considers the problem of modified portmanteau tests for testing the adequacy of FARIMA models under the assumption that the errors are uncorrelated but not necessarily independent ({\em i.e.} weak FARIMA). We first study the joint distribution of the least squares estimator and the noise empirical autocovariances. We then derive the asymptotic distribution of residual empirical autocovariances and autocorrelations. We deduce the asymptotic distribution of the Ljung-Box (or Box-Pierce) modified portmanteau statistics for weak FARIMA models. We also propose another method based on a self-normalization approach to test the adequacy of FARIMA models. Finally some simulation studies are presented to corroborate our theoretical work.
An application to the Standard \& Poor's 500 and Nikkei returns also illustrate the practical relevance of our theoretical results.
\end{abstract}
\end{frontmatter}
%************************************************
%************************************************
%************************************************
%************************************************
%
\section{Introduction}
To model the long memory phenomenon, a widely used model is the fractional autoregressive
integrated moving average  (FARIMA, for short) model (see for instance \cite{Granger&Joyeux}, \cite{Fox&Taqqu1986}, \cite{Dahlhaus1989}, \cite{Hosking}, \cite{Beran2013}, \cite{palma}, among others). This model plays an important role in many scientific disciplines and applied fields such as hydrology, climatology, economics, finance, to name a few.

We consider a  centered stationary process $X:=(X_t)_{t\in\mathbb{Z}}$ satisfying a  FARIMA$(p,d_0,q)$ representation of the form %be such that, for all $t\in\mathbb{Z}$,
\begin{equation}\label{FARIMA}
a(L)(1-L)^{d_0}X_t=b(L)\epsilon_t,
\end{equation}
where $d_0$
% \in\left]0,1/2\right[$
is the long memory parameter, $L$ stands for the back-shift operator and
$a(L)=1-\sum_{i=1}^pa_iL^i$, respectively $b(L)=1-\sum_{i=1}^qb_iL^i$, is the autoregressive, respectively the moving average,  operator. These operators represent the short memory part of the model (by convention $a_0=b_0=1$).
%It is assumed that the model defined by \eqref{FARIMA} is stationary, invertible and not redundant.
%Without loss of generality, assume that $a_p^2+b_q^2\neq 0$.
In the standard situation ${\epsilon}:=({\epsilon}_t)_{t\in\mathbb Z}$ is assumed to be a sequence of independent and identically distributed (iid for short)  random variables with zero mean and with a common variance. In this standard framework, ${\epsilon}$ is said to be a \emph{strong
white noise} and the representation (\ref{FARIMA})  is called a strong   FARIMA$(p,d_0,q)$ process. In contrast with this previous definition, the representation  (\ref{FARIMA}) is said to be a weak FARIMA$(p,d_0,q)$ if the noise process $\epsilon$  is a \emph{weak white noise}, that is, if it satisfies
\begin{itemize}
\item[\hspace*{1em} {\bf (A0):}]
\hspace*{1em} $\mathbb{E}({{\epsilon}}_t)=0$,
$\mbox{Var}\left({{\epsilon}}_t\right)=\sigma^2_0$ and
$\mbox{Cov}\left({{\epsilon}}_t,{{\epsilon}}_{t-h}\right)=0$ for all
$t\in\mathbb{Z}$ and all $h\neq 0$.
\end{itemize}
 A strong white noise is obviously  a weak white
noise because independence entails uncorrelatedness. Of course the
converse is not true.
The strong FARIMA model was introduced by \cite{Hosking}. The particular strong FARIMA$(0,d_0,0)$ process
was discussed by \cite{Granger&Joyeux}. To ensure the stationarity and  the invertibility  of the model defined by \eqref{FARIMA},
we assume that $-1/2<d_0<1/2$ and all roots of $a(z)b(z)=0$ are outside the unit disk (see \cite{Granger&Joyeux} and \cite{Hosking} for details).
It is also assumed that $a(z)$ and $b(z)$ have no common factors in order to insure unique identifiability of the parameters.
%Without loss of generality it is assumed that $a_p^2+b_q^2\neq 0$.

The validity of the different steps of the traditional methodology
of Box and Jenkins (identification, estimation and validation)
depends on the noise properties. After estimating the FARIMA process, the next important step in the modeling consists in checking if
the estimated model fits satisfactorily the data. Thus, under
the null hypothesis that the model has been correctly identified, the residuals ($\hat{\epsilon}_t$)
are approximately a white noise.
 This adequacy checking step allows to validate or invalidate the choice of the
orders $p$ and $q$. The choice of $p$ and
$q$ is particularly important because the number of parameters
$(p+q+1)$ quickly increases with $p$ and $q$, which entails
statistical difficulties.
In particular, the selection of too large orders $p$ and $q$ may introduce
terms that are not necessarily relevant in the model.
%In other terms, overidentification generally leads to a loss of precision in parameter estimation.
Conversely, the selection of too small orders $p$ and $q$ causes loss of some information,
that can be detected by the correlation of the residuals.

Thus it is important to check the validity
of a FARIMA$(p,d_0,q)$ model, for given orders $p$ and $q$.
Based on the residual empirical autocorrelation, \cite{bp70} have proposed a goodness-of-fit test,
the so-called portmanteau test, for strong ARMA models.
The intuition behind these portmanteau tests is that if a given
time series model with iid innovation is appropriate for the data at hand,
the autocorrelations of the residuals $\hat{\epsilon}_t$ should be close to zero, which is the theoretical
value of the autocorrelations of  $\epsilon_t$ (see Assumption {\bf (A0)} below).
 A modification of the test of \cite{bp70} has been proposed by \cite{lb}
 which is nowadays  one of the most popular diagnostic checking tools in strong ARMA modeling of time series.
A modified  portmanteau test statistic was proposed by \cite{LiMcL1986} for checking the overall significance of the residual
 autocorrelations of a strong FARIMA$(p,d_0,q)$  model. All these above test statistics have been obtained under the iid assumption
 on the noise
and they may be invalid when the series is uncorrelated but dependent (see \cite{rt1996}, \cite{frz}, \cite{BMS2018}, \cite{zl2015JE}, \cite{lobatoNS2001},
\cite{lobatoNS2002}, \cite{XY2020}, to name a few).
%The distinction between iid and uncorrelated dependence is important in view of
%the popularity of conditional heteroscedastic models (e.g. GARCH
%models) as used in econometric and financial time series modeling.

As mentioned above, the works on the portmanteau statistic are generally
performed under the assumption that the errors $\epsilon_t$ are independent (see for instance \cite{LiMcL1986}).
This independence assumption is often considered too restrictive by
practitioners. It precludes conditional heteroscedasticity and/or other forms of nonlinearity
(see \cite{fz05} for a review on weak univariate ARMA models) which can not be generated by FARIMA models with iid noises.\footnote{ To cite few
examples of nonlinear processes, let us mention: the generalized autoregressive conditional heteroscedastic (GARCH) model (see \cite{FZ2010}), the self-exciting
threshold autoregressive (SETAR), the smooth transition
autoregressive (STAR), the exponential autoregressive (EXPAR), the
bilinear, the random coefficient autoregressive (RCA), the
functional autoregressive (FAR) (see \cite{Tong1990} and \cite{FY2008}, for references on these nonlinear time series models).} Relaxing this independence assumption  allows to cover
linear representations of general nonlinear processes and to extend
the range of application of the FARIMA models.

This paper is devoted to the problem of the validation step of weak FARIMA processes.
For the asymptotic theory of weak FARIMA model validation, recently
\cite{shao2011ET} studied the diagnostic checking for long memory time series models with nonparametric conditionally
heteroscedastic martingale difference errors. This author also generalized the test statistic based on the kernel-based spectral proposed by \cite{hong1996} under weak assumptions
on the innovation process. 
%As mentioned by the author, his work is no longer valid for $d_0\in[1/4,1/2[$ (see Remark 3.3 in \cite{shao2011ET}).
Note also that \cite{ling-li} have studied the  \cite{bp70} type test for FARIMA-GARCH models by assuming a parametric form for the GARCH model.

To our knowledge, it does not exist any diagnostic checking methodology for FARIMA models when the (possibly dependent) error is subject to unknown conditional  heteroscedasticity.
We think that this is due to the difficulty that arises when one has to estimate the asymptotic  covariance matrix of the parameter estimates.
%Thus it is  important to have a soundness validation procedure for the residual of
%FARIMA models when the (possibly dependent) error is subject to unknown conditional  heteroscedasticity.
In our paper, thanks to the asymptotic results obtained by \cite{BMES2019}, we are able to  extend for weak FARIMA models the diagnostic checking
methodology proposed by \cite{frz} as well as the self-normalized approach proposed by \cite{BMS2018}.
%We strength the fact that, contrarily to \cite{shao2011ET}, our results are valid for $d_0\in]0,1/2[$.

The paper is organized as follows. In Section \ref{11}, we recall the results on the least squares estimator asymptotic distribution of weak FARIMA models obtained by \cite{BMES2019}.
In Section \ref{result}, a modified version of the portmanteau test is proposed thanks to the investigation of the asymptotic distribution of the residual autocorrelations. Our first main result is stated in Theorem \ref{loi_res_rho}. The second main result of this section is obtained in Theorem \ref{sn2} by means of a self-normalized approach. Two examples are also proposed in Section \ref{exple} in order to illustrate our results.
Some numerical illustrations are gathered in Section \ref{num-ill}. They corroborate our theoretical work.
An application to the Standard \& Poor's 500 and Nikkei  returns also illustrate the practical relevance of our theoretical results.
All our proofs are given in Section \ref{appendix} and figures and tables are brought together in Section \ref{sec-fig}.

\section{Assumptions and estimation procedure}\label{11}
In this section, we recall the results on the least squares estimator asymptotic distribution of weak FARIMA models obtained by \cite{BMES2019} in  order to have a self-containing paper.

Let $\Theta^{*}$ be the parameter space
\begin{align*}
\Theta^{*}:=&  \Big
\{(\theta_1,\theta_2,\dots,\theta_{p+q})
\in{\Bbb R}^{p+q},\text{ where }  a_{\theta}(z)=1-\sum_{i=1}^{p}\theta_{i}z^{i},\, \text{
and } b_{\theta}(z)=1-\sum_{j=1}^{q}\theta_{p+j}z^{j} \\
& \hspace{2.5cm} \text{ have all
their zeros outside the unit disk}\Big \}\ .
\end{align*}
Denote by $\Theta$ the cartesian product $\Theta^{*}\times\left[ d_1,d_2\right]$, where $\left[ d_1,d_2\right]\subset\left] -1/2,1/2\right[$ with $d_1-d_0>-1/2$.
The unknown parameter of interest $\theta_0=(a_1,a_2,\dots,a_p,b_1,b_2,\dots,b_q,d_0)^{'}$ is supposed to belong to the parameter space $\Theta$.

The fractional difference operator $(1-L)^{d_0}$ is defined, using the generalized binomial series, by
\begin{equation*}%\label{OP-FR}
(1-L)^{d_0}=\sum_{j\geq 0}\alpha_j(d_0)L^j,
\end{equation*}
where for all $j\ge 0$, $\alpha_j(d_0)=\Gamma(j-d_0)/\left\lbrace \Gamma(j+1)\Gamma(-d_0)\right\rbrace $ and $\Gamma(\cdot)$ is the Gamma function.
Using the Stirling formula we obtain that for large $j$,  $\alpha_j(d_0)\sim j^{-d_0-1}/\Gamma(-d_0)$ (one refers to \cite{Beran2013} for further details).

For all $\theta\in\Theta$ we define $( \epsilon_t(\theta)) _{t\in\mathbb{Z}}$ as the second order stationary process which is the solution of
\begin{equation}\label{FARIMA-th}
\epsilon_t(\theta)=\sum_{j\geq 0}\alpha_j(d)X_{t-j}-\sum_{i=1}^p\theta_i\sum_{j\geq 0}\alpha_j(d)X_{t-i-j}+\sum_{j=1}^q\theta_{p+j}\epsilon_{t-j}(\theta).
\end{equation}
Observe that, for all $t\in\mathbb{Z}$, $\epsilon_t(\theta_0)=\epsilon_t$ a.s.  Given a realization  $X_1,\dots,X_n$ of length $n$, $\epsilon_t(\theta)$ can be approximated, for $0<t\leq n$, by $\tilde{\epsilon}_t(\theta)$ defined recursively by
\begin{equation}\label{exp-epsil-tilde}
\tilde{\epsilon}_t(\theta)=\sum_{j=0}^{t-1}\alpha_j(d)X_{t-j}-\sum_{i=1}^p\theta_i\sum_{j=0}^{t-i-1}\alpha_j(d)X_{t-i-j}+\sum_{j=1}^q\theta_{p+j}\tilde{\epsilon}_{t-j}(\theta),
\end{equation}
with $\tilde{\epsilon}_t(\theta)=X_t=0$ if $t\leq 0$.

As shown in of \cite{BMES2019}, these initial values are asymptotically
negligible and in particular it holds that
$\epsilon_t(\theta)-\tilde{\epsilon}_t(\theta)\to 0$ in $\mathbb{L}^2$ as $t\to\infty$.
Thus the choice of the initial values has no influence on the asymptotic properties of the model parameters estimator.
Let $\Theta^{*}_{\delta}$ denotes the compact set
$$\Theta^{*}_{\delta}=\left\lbrace \theta\in\mathbb{R}^{p+q}; \text{ the roots of the polynomials } a_{\theta}(z) \text{ and } b_{\theta}(z) \text{ have modulus } \geq 1+\delta\right\rbrace.$$
We define the set $\Theta_{\delta}$ as the cartesian product of $\Theta^{*}_{\delta}$ by $\left[ d_1,d_2\right] $, {\em i.e.} $\Theta_{\delta}=\Theta^{*}_{\delta}\times\left[ d_1,d_2\right]$, where $\delta$ is a positive constant chosen such that $\theta_0$ belongs to $\Theta_{\delta}$.

The least squares estimator is defined, almost-surely, by
\begin{equation}\label{theta_chap}
\hat{\theta}_n=\underset{\theta\in\Theta_{\delta}}{\mathrm{argmin}} \ Q_n(\theta),
\text{ where }Q_n(\theta)=\frac{1}{n}\sum_{t=1}^n\tilde{\epsilon}_t^2(\theta).
\end{equation}
The asymptotic properties of this estimator are well known when the innovation process $(\epsilon_t)_{t\in\mathbb{Z}}$ is a strong or a semi-strong white noise (see for instance \cite{hualde2011}, \cite{Nielsen2015} and \cite{CAVALIERE2017} who have considered the problem of conditional sum-of squares estimation with $d_0$ allowed to lie in an arbitrary large compact set).
To ensure the consistency of the least squares estimator in our context, we assume as in \cite{BMES2019} that the parametrization satisfies the following condition.
\begin{itemize}
\item[\hspace*{1em} {\bf (A1):}]
\hspace*{1em} The process $(\epsilon_t)_{t\in\mathbb{Z}}$ is strictly stationary and ergodic.
\end{itemize}
The consistency of the estimator is obtained under the assumptions \textbf{(A0)} and \textbf{(A1)}.
Additional assumptions are required in order to establish the asymptotic normality of the  least squares estimator. We assume
that $\theta_0$ is not on the boundary of the parameter space
${\Theta}_\delta$.
\begin{itemize}
\item[\hspace*{1em} {\bf (A2):}]
\hspace*{1em} We have $\theta_0\in\overset{\circ}{\Theta}_\delta$, where $\overset{\circ}{\Theta}_\delta$ denotes the interior of $\Theta_\delta$.
\end{itemize}
The stationary process $\epsilon$ is not supposed to be an independent sequence. So one needs to control its dependency by means of its  strong mixing coefficients $\left\lbrace \alpha_{\epsilon}(h)\right\rbrace _{h\in\mathbb{N}}$ defined by
$$\alpha_{\epsilon}\left(h\right)=\sup_{A\in\mathcal F_{-\infty}^t,B\in\mathcal F_{t+h}^{\infty}}\left|\mathbb{P}\left(A\cap
B\right)-\mathbb{P}(A)\mathbb{P}(B)\right|,$$
where $\mathcal F_{-\infty}^t=\sigma (\epsilon_u, u\leq t )$ and $\mathcal F_{t+h}^{\infty}=\sigma (\epsilon_u, u\geq t+h)$.

%Let $\epsilon$ be a $\tau$-th order stationary process.
We shall need  an integrability assumption on the moment of the noise $\epsilon$ and a summability condition on the strong mixing coefficients $(\alpha_{\epsilon}(h))_{h\ge 0}$.
\begin{itemize}
\item[\hspace*{1em} {\bf (A3):}]
\hspace*{1em}  There exists an integer $\tau$ such that for some $\nu\in]0,1]$, we have $\mathbb{E}|\epsilon_t|^{\tau+\nu}<\infty$ and $\sum_{h=0}^{\infty}(h+1)^{k-2} \left\lbrace \alpha_{\epsilon}(h)\right\rbrace ^{\frac{\nu}{k+\nu}}<\infty$ for $k=1,\dots,\tau$.
\end{itemize}
%We have used the notation  $|\cdot|$ for the Euclidian norm $|z|=\left(\sum_{i=1}^\tau z_i^2\right)^{1/2}$ of a vector $z=(z_1,\dots,z_\tau)'$.

Note that {\bf (A3)} implies the following weak assumption on the joint cumulants of the innovation process $\epsilon$ (see \cite{doukhan1989}, for more details).
%but the reverse is true under certain conditions .
\begin{itemize}
\item[\hspace*{1em} {\bf (A3'):}]
\hspace*{1em} There exists an integer $\tau\ge 2$ such that $C_\tau:=\sum_{i_1,\dots,i_{\tau-1}\in\mathbb{Z}}|\mathrm{cum}(\epsilon_0,\epsilon_{i_1},\dots,\epsilon_{i_{\tau-1}})|<\infty \ .$
\end{itemize}
In the above expression, $\mathrm{cum}(\epsilon_0,\epsilon_{i_1},\dots,\epsilon_{i_{\tau-1}})$ denotes the $\tau-$th order joint cumulant of the stationary process $\epsilon$. Due to the fact that the $\epsilon_t$'s are centered,  we notice that for fixed $(i,j,k)$
$$\mathrm{cum}(\epsilon_0,\epsilon_i,\epsilon_j,\epsilon_k)=\mathbb{E}\left[\epsilon_0\epsilon_i\epsilon_j\epsilon_k\right]-\mathbb{E}\left[\epsilon_0\epsilon_i\right]\mathbb{E}\left[\epsilon_j\epsilon_k\right]-\mathbb{E}\left[\epsilon_0\epsilon_j\right]\mathbb{E}\left[\epsilon_i\epsilon_k\right]-\mathbb{E}\left[\epsilon_0\epsilon_k\right]\mathbb{E}\left[\epsilon_i\epsilon_j\right].$$
Assumption \textbf{(A3)} is a usual technical hypothesis which is useful when one proves the asymptotic normality (see \cite{fz98} for example). Let us notice however that we impose a stronger convergence speed for the mixing coefficients than in the works on weak ARMA processes. This is due to the fact that the coefficients in the infinite AR or MA representation of $\epsilon_t(\theta)$ have no more exponential decay because of the fractional operator (see Subsection 6.1 in \cite{BMES2019} for details and comments).

As mentioned before, Hypothesis  \textbf{(A3)} implies  \textbf{(A3')} which is also a technical assumption usually used in the fractional ARIMA processes framework (see for instance \cite{s2010JRSSBa,shao2011ET}) or even in an ARMA context (see \cite{FZ2007,zl2015JE}).

For all $t\in\mathbb{Z}$, let
\begin{align*}%\label{Ht_process}
H_t(\theta)&=2\epsilon_t(\theta)\frac{\partial}{\partial\theta}\epsilon_t(\theta)%\nonumber \\&
=\left(2\epsilon_t(\theta)\frac{\partial}{\partial\theta_1}\epsilon_t(\theta),\dots,2\epsilon_t(\theta)\frac{\partial}{\partial\theta_{p+q+1}}\epsilon_t(\theta)\right) ^{'}.
\end{align*}
Remind that the sequence $( \epsilon_t(\theta))_{t\in \mathbb Z}$ is given by \eqref{FARIMA-th}.
Under the assumptions {\bf (A0)}, {\bf (A1)}, {\bf (A2)} and {\bf (A3)} with $\tau=4$, \cite{BMES2019} showed that $\hat{\theta}_n\to\theta_0$ in probability as $n\to\infty$ %(the consistency)
and $\sqrt{n}(\hat{\theta}_n-\theta_0)$ is asymptotically
normal with mean $0$ and covariance matrix
$\Sigma_{\hat\theta}:=J^{-1}IJ^{-1},$ where $J=J(\theta_0)$ and
$I=I(\theta_0)$, with
$$I(\theta)=\sum_{h=-\infty}^{+\infty}\mathrm{Cov}\left( H_t(\theta),H_{t-h}(\theta)\right) \text{  and  } J(\theta)=2\mathbb{E}\left( \frac{\partial}{\partial\theta}\epsilon_t(\theta)\frac{\partial}{\partial\theta'}\epsilon_t(\theta)\right) \ \mathrm{a.s.}$$
%Remind that the sequence $( \epsilon_t(\theta))_{t\in \mathbb Z}$ is given by \eqref{FARIMA-th}.

%In the strong FARIMA case, {\em i.e.} when {\bf (A1)} is replaced by the assumption that $(\epsilon_t)_{t\in \mathbb Z}$ is an iid sequence, we have $I=2\sigma^2_{\epsilon}J$, so that
%$\Sigma_{\hat{\theta}}=2\sigma^2_{\epsilon}J^{-1}$ (see Remark 2 in \cite{BMES2019}).
\section{Diagnostic checking in weak FARIMA models}\label{result}
After the estimation phase,  the next important step consists in checking if the estimated model fits satisfactorily the data.
In this section we derive the limiting distribution of the residual autocorrelations and that of the portmanteau
statistics (based on the standard and the self-normalized approaches) in the framework of weak FARIMA models.

For $t\geq 1$, let $\hat{e}_t=\tilde{\epsilon}_t(\hat{\theta}_n)$ be the least squares residuals. By \eqref{exp-epsil-tilde} we notice that $\hat{e}_t=0$ for $t\leq 0$ and $t>n$. By \eqref{FARIMA} it holds that
$$\hat{e}_t=\sum_{j=0}^{t-1}\alpha_j(\hat{d})\hat{X}_{t-j}-\sum_{i=1}^p\hat{\theta}_i\sum_{j=0}^{t-i-1}\alpha_j(\hat{d})\hat{X}_{t-i-j}+\sum_{j=1}^q\hat{\theta}_{p+j}\hat{e}_{t-j},$$
for $t=1,\dots,n$, with $\hat{X}_t=0$ for $t\leq 0$ and $\hat{X}_t=X_t$ for $t\geq 1$.

For a fixed integer $m\geq1$ consider the vector of residual autocovariances
%\begin{eqnarray*}
$${\hat\gamma}_m
=\left(\hat\gamma(1),\dots,\hat\gamma(m)\right)' \  \text{where}\ \; \hat\gamma(h)=\frac{1}{n}\sum_{t=h+1}^{n}\hat e_t\, \hat e_{t-h}
 \  \text{for}\ \;0\leq h<n.$$
%\end{eqnarray*}
%with $\gamma(0):=\sigma^2$.
In the sequel we will also need the vector of the first $m$  sample autocorrelations
$$\hat\rho_m =\left(\hat \rho
(1),\dots,\hat \rho(m)\right)' \  \text{where}\ \; \hat \rho(h)=\hat \gamma(h)/\hat \gamma(0).$$
Since the papers by \cite{bp70} and \cite{lb}, portmanteau tests have been
popular diagnostic checking tools in the ARMA modeling of time series.
Based on the residual empirical autocorrelations, their test statistics are defined respectively by
\begin{equation}\label{stat-test}
{Q}_m^{\textsc{bp}}=n\sum_{h=1}^m\hat{\rho}^2(h) \ \text{ and } \ {Q}_m^{\textsc{lb}}=n(n+2)\sum_{h=1}^m\frac{\hat{\rho}^2(h)}{n-h}.
\end{equation}
These statistics are usually used to test the null hypothesis
\begin{itemize}
\item[{\bf (H0)}:]  $({X}_t)_{t\in \mathbb Z}$ satisfies a FARIMA$(p,d_0,q)$ representation;
\end{itemize}
against the alternative
\begin{itemize}
\item[{\bf ({H1})}:]  $({X}_t)_{t\in \mathbb Z}$ does not admit a FARIMA representation or admits a FARIMA$(p^{'},d_0,q^{'})$ representation with $p^{'}>p$ or $q^{'}>q$.
\end{itemize}
These tests are very useful tools to check the global significance of the residual autocorrelations.
\subsection{Asymptotic distribution of the residual autocorrelations}
First of all, the mixing assumption {\bf (A3)} will entail the asymptotic normality of the "empirical" autocovariances
\begin{equation}\label{gamma_m}
\gamma_m =\left(\gamma(1),\dots,\gamma(m)\right)'
\text{ where } \gamma(h)=\frac{1}{n}\sum_{t=h+1}^{n}\epsilon_t\, \epsilon_{t-h} \  \text{for}\ \;0\leq h<n.
\end{equation}
%\end{eqnarray*}
It should be noted that $\gamma(h)$ is not a computable statistic because it depends on
the unobserved innovations $\epsilon_t=\epsilon_t(\theta_0)$.
%since they depend on the unknown parameter $\theta_0$.
They are introduced as a device to facilitate future derivations.
Let $\Psi_m$ be the $m\times (p+q+1)$ matrix defined by
\begin{equation}\label{Psi_m}
\Psi_m=\mathbb{E}
\left\{\left(\begin{array}{c}
\epsilon_{t-1}\\\vdots\\\epsilon_{t-m}\end{array}\right)\frac{\partial
\epsilon_{t}}{\partial\theta'}\right\}.
\end{equation}
By a Taylor expansion of $\sqrt{n}\hat{\gamma}_m$, one should prove that (see Section \ref{proof_loi_res})
\begin{align}\label{hat_gamma}
\sqrt{n}\hat{\gamma}_m =%\left( \sqrt{n}\hat{\gamma}(1),\dots,\sqrt{n}\hat{\gamma}(m)\right)^{'}=
\sqrt{n}\gamma_m+\Psi_m\sqrt{n}\left(\hat{\theta}_n-\theta_0\right)+\mathrm{o}_{\mathbb{P}}(1),
\end{align}
where $\Psi_m$ is given in $\eqref{Psi_m}$. We shall also prove (see Section \ref{proof_loi_res} again) that
\begin{align}\label{hat_rho}
\sqrt{n}\hat{\rho}_m=\sqrt{n}\frac{\hat{\gamma}_m}{\sigma_{\epsilon}^2}+\mathrm{o}_{\mathbb{P}}(1).
\end{align}
Thus from \eqref{hat_rho} the asymptotic distribution of the residual autocorrelations $\sqrt{n}\hat{\rho}_m$ depends on the
 distribution of $\hat{\gamma}_m$.
In view of \eqref{hat_gamma} the asymptotic distribution of the residual autocovariances $\sqrt{n}\hat{\gamma}_m$
will be obtained from the joint asymptotic behavior of $\sqrt{n}( \hat{\theta}_n'-\theta_0^{'},\gamma_m^{'})^{'}$.
%\subsection{Joint distribution of $\hat{\theta}_n$ and the noise empirical autocovariances}

In view of Theorem 1 in \cite{BMES2019} and {\bf (A2)}, we have
$\hat{\theta}_n\to \theta_0\in \stackrel{\circ}{\Theta}$ in probability. Thus
$\partial Q_n(\hat{\theta}_n)/\partial \theta=0$ for sufficiently
large $n$ and a Taylor expansion gives
\begin{align}\label{score}
\sqrt{n}\frac{\partial}{\partial\theta}O_n(\theta_0)+J(\theta_0)\sqrt{n}(\hat{\theta}_n-\theta_0)=\mathrm{o}_{\mathbb{P}}(1),
\end{align}
where $O_n(\theta)={n}^{-1}\sum_{t=1}^n\epsilon_t^2(\theta)$ and 
the sequence $( \epsilon_t(\theta))_{t\in \mathbb Z}$ is given by \eqref{FARIMA-th}.
The equation \eqref{score} is proved in \cite{BMES2019} (see the proof of Theorem 2).
Consequently from \eqref{score} we have
\begin{align}\label{ecart-theta}
\sqrt{n}(\hat{\theta}_n-\theta_0)=-\frac{2}{\sqrt{n}}\sum_{t=1}^nJ^{-1}(\theta_0)\epsilon_t(\theta_0)\frac{\partial\epsilon_t(\theta_0)}{\partial\theta} +\mathrm{o}_{\mathbb{P}}\left(1\right).
\end{align}

For integers $m,m'\ge 1$, one needs the matrix $\Gamma_{m,m'}=[\Gamma(\ell,\ell^{'})]_{1\leq\ell\leq m,1\leq\ell'\leq m'}$ where
$$\Gamma(\ell,\ell^{'})=\sum_{h=-\infty}^{\infty}\mathbb{E}\left[ \epsilon_t\epsilon_{t-\ell}\epsilon_{t-h}\epsilon_{t-h-\ell^{'}}\right].$$
%for $(\ell,\ell^{'})\neq (0,0)$.
The existence of $\Gamma(\ell,\ell^{'})$ will be justified in Lemma~\ref{lemGamma} of the
appendix.
\begin{prop}\label{loijointe} Under the assumptions {\bf (A0)}, {\bf (A1)}, {\bf (A2)} and {\bf (A3)} with $\tau=4$, the random vector
$$\sqrt{n}\left( \left(\hat{\theta}_n-\theta_0\right)^{'},\gamma_m^{'}\right)^{'}$$
has a limiting centered normal distribution with covariance matrix %$\Xi$, where
\begin{align}\label{eq:specU}
\Xi=\begin{pmatrix}\Sigma_{\hat{\theta}} \ \ \ \Sigma_{\hat{\theta},\gamma_m}\vspace{0.2cm}\\
\Sigma_{\hat{\theta},\gamma_m}^{'} \ \ \ \Gamma_{m,m}
\end{pmatrix}=\sum_{h=-\infty}^{\infty}\mathbb{E}\left[ U_tU_{t-h}^{'}\right],
\end{align}
where from \eqref{gamma_m} and \eqref{ecart-theta} we have
\begin{equation}\label{U_t}
U_t=\begin{pmatrix}U_{1t}\vspace{0.2cm}\\
U_{2t}
\end{pmatrix}=\begin{pmatrix}-2J^{-1}(\theta_0)\epsilon_t(\theta_0)\frac{\partial}{\partial\theta}\epsilon_t(\theta_0)\vspace{0.2cm}\\
(\epsilon_{t-1},\dots,\epsilon_{t-m})^{'}\epsilon_t
\end{pmatrix}.
\end{equation}
\end{prop}
The proof of the proposition is given in Subsection~\ref{proof_loijointe} of the appendix.

The following theorem which is an extension of the result given in \cite{frz} provides the limit distribution of the residual autocovariances and autocorrelations of weak FARIMA models.
\begin{thm}\label{loi_res_rho} Under the assumptions of Proposition~\ref{loijointe}, we have
\begin{equation}\label{loi_hat_gamma}
\sqrt{n}\hat{\gamma}_m\overset{\mathcal{D}}{\underset{n\rightarrow \infty}{\mathbf{\longrightarrow}}}\mathcal{N}\left(0,\Sigma_{\hat{\gamma}_m}\right) \ \ \text{ where } \ \
\Sigma_{\hat{\gamma}_m}=\Gamma_{m,m}+\Psi_m\Sigma_{\hat{\theta}}\Psi_m^{'}+\Psi_m\Sigma_{\hat{\theta},\gamma_m}+\Sigma_{\hat{\theta},\gamma_m}^{'}\Psi_m^{'}
\end{equation}
and
\begin{equation}\label{loi_hat_rho}
 \sqrt{n}\hat{\rho}_m\overset{\mathcal{D}}{\underset{n\rightarrow \infty}{\mathbf{\longrightarrow}}}\mathcal{N}\left(0,\Sigma_{\hat{\rho}_m}\right) \ \ \text{ where } \ \
\Sigma_{\hat{\rho}_m}=\frac{1}{\sigma_{\epsilon}^4}\Sigma_{\hat{\gamma}_m}.
\end{equation}
\end{thm}
The detailed proof of this result is postponed to the Subsection~\ref{proof_loi_res} of Appendix.
\begin{rmq}
It is clear from Theorem~\ref{loi_res_rho} that for a given FARIMA$(p,d_0,q)$ model, the asymptotic distribution of
the residual autocorrelations depends only on the noise distribution through the quantities $\Gamma(\ell,\ell^{'})$ (which depends
on the fourth-order structure of the noise).
It is also worth noting that this asymptotic distribution depends on the asymptotic normality of the least squares estimator of the FARIMA$(p,d_0,q)$ only through the matrix
$\Sigma_{\hat{\theta}}$.
\end{rmq}
\begin{rmq}\label{casfort} In the standard strong FARIMA case, {\em i.e.}
when {\bf (A1)} is replaced by the assumption that $(\epsilon_t)_{t\in\mathbb{Z}}$ is
iid, \cite{BMES2019} have showed in Remark 2 that $I(\theta_0)=2\sigma_{\epsilon}^2J(\theta_0)$. Thus the asymptotic
covariance matrix is then reduced as $\Sigma_{\hat\theta}=2\sigma_{\epsilon}^2J^{-1}(\theta_0)$.
In the strong case, we also have: $\Gamma(\ell,\ell')=0$ when $\ell\neq\ell'$ and $\Gamma(\ell,\ell)=\sigma_\epsilon^4$. % for  $h=0$.
Thus $\Gamma_{m,m}$ is reduced as $\Gamma_{m,m}=\sigma_\epsilon^4I_m$, where $I_m$ denotes the $m\times m$ identity matrix.
Because $\Sigma_{\hat\theta}=2\sigma_{\epsilon}^2J^{-1}(\theta_0)$ we obtain that
\begin{eqnarray*}
\Sigma_{\hat\theta,\gamma_m}&=&-2\sum_{h=-\infty}^{\infty}\mathbb{E}\left\{\epsilon_tJ^{-1}(\theta_0)\frac{\partial
\epsilon_t(\theta_0)}{\partial\theta}
\right\}\left\{\left(\begin{array}{c}
{\epsilon}_{t-1-h}\\\vdots\\
{\epsilon}_{t-m-h}\end{array}\right)
{\epsilon}_{t-h}\right\}'\\&=&-\left(2\sigma_{\epsilon}^2J^{-1}(\theta_0)\right)\left\{\mathbb{E}\left[\left(\begin{array}{c}
{\epsilon}_{t-1}\\\vdots\\
{\epsilon}_{t-m}\end{array}\right)\frac{\partial
\epsilon_t(\theta_0)}{\partial\theta'}
\right]\right\}'=-\Sigma_{\hat\theta}\Psi_m'.
\end{eqnarray*}
We denote by $\Sigma_{\hat\gamma_m}^{\textsc{s}}$ and $\Sigma_{\hat\rho_m}^{\textsc{s}}$
the asymptotic variances obtained respectively in \eqref{loi_hat_gamma} and \eqref{loi_hat_rho} for the strong FARIMA case.
Thus we obtain, in the strong case, the following simpler expressions $$\Sigma_{\hat\gamma_m}^{\textsc{s}}=\sigma_\epsilon^4
I_m-2\sigma_\epsilon^2\Psi_mJ^{-1}(\theta_0)\Psi'_m\quad\text{ and }\quad\Sigma_{\hat\rho_m}^{\textsc{s}}=
I_m-\frac{2}{\sigma_\epsilon^2}\Psi_mJ^{-1}(\theta_0)\Psi'_m,$$ which are the matrices obtained by \cite{LiMcL1986}.
\end{rmq}
To validate a FARIMA$(p,d_0,q)$ model, the most basic technique is to examine the autocorrelation function
of the residuals. Theorem~\ref{loi_res_rho} can be used to obtain asymptotic significance limits for the residual
autocorrelations. However, the asymptotic variance matrices $\Sigma_{\hat{\gamma}_m}$ and $\Sigma_{\hat{\rho}_m}$ depend on the unknown matrices $\Xi$, $\Psi_m$ and the positive scalar $\sigma_{\epsilon}^2$ which need to be estimated. This is the purpose of the following discussion.
\subsection{Modified version of the portmanteau test}\label{diagnostic1}
From Theorem~\ref{loi_res_rho} we can deduce the following result, which gives the limiting
distribution of the standard portmanteau statistics \eqref{stat-test} under general assumptions on
the innovation process of the fitted FARIMA$(p,d_0,q)$ model.
\begin{thm}\label{loi_stat_test} Under the assumptions of Theorem~\ref{loi_res_rho} and {\bf (H0)},
the statistics ${Q}_m^{\textsc{bp}}$ and ${Q}_m^{\textsc{lb}}$ defined by \eqref{stat-test}  converge in distribution, as $n\rightarrow\infty$, to
$$Z_m(\xi_m)=\sum_{k=1}^m\xi_{k,m}Z_k^2,$$
where $\xi_m=(\xi_{1,m},\dots,\xi_{m,m})^{'}$ is the vector of the eigenvalues of the matrix $\Sigma_{\hat{\rho}_m}=\sigma_{\epsilon}^{-4}\Sigma_{\hat{\gamma}_m}$ and $Z_1,\dots,Z_m$ are independent
$\mathcal{N}(0,1)$ variables.
\end{thm}
It is possible to evaluate the distribution of a quadratic form of a Gaussian vector
by means of the Imhof algorithm (see \cite{i1961}).
\begin{rmq}\label{remKhi2} In view of remark~\ref{casfort} when $m$ is large,
$\Sigma_{\hat\rho_m}^{\textsc{s}}\simeq I_m-2{\sigma_\epsilon^{-2}}\Psi_mJ^{-1}(\theta_0)\Psi'_m$ is close
to a projection matrix. Its eigenvalues are therefore equal to 0
and 1. The number of eigenvalues equal to 1 is $\mathrm{Tr}(I_m-2{\sigma_\epsilon^{-2}}\Psi_mJ^{-1}(\theta_0)\Psi'_m
)=\mathrm{Tr}(I_{m-(p+q+1)})=m-(p+q+1)$ and $p+q+1$ eigenvalues equal to 0, $\mathrm{Tr}(\cdot)$ denotes the trace of a matrix.
Therefore we retrieve the well-known  result obtained by \cite{LiMcL1986}. More precisely,  under {\bf (H0)} and in the strong
FARIMA case, the asymptotic distributions of the statistics ${Q}_m^{\textsc{bp}}$ and ${Q}_m^{\textsc{lb}}$  are approximated by a $\mathcal{X}_{m-(p+q+1)}^2$, where $m>p+q+1$ and $\mathcal{X}^2_k$ denotes the chi-squared distribution with $k$ degrees of freedom. Theorem~\ref{loi_stat_test} shows that this approximation is no longer valid in the framework of weak FARIMA$(p,d,q)$ models and that the asymptotic null distributions of the statistics ${Q}_m^{\textsc{bp}}$ and ${Q}_m^{\textsc{lb}}$  are more complicated.
\end{rmq}
The limit distribution $Z_m(\xi_m)$ depends on the nuisance parameter $\sigma_{\epsilon}^2$, the matrix $\Psi_m$ and the elements of $\Xi$. Therefore, the asymptotic distribution of the portmanteau statistics \eqref{stat-test}, under weak assumptions on the noise, requires a computation of a consistent estimator of the asymptotic covariance matrix $\Sigma_{\hat{\rho}_m}$.
 The $m\times (p+q+1)$ matrix $\Psi_m$ and the noise variance $\sigma_{\epsilon}^2$ can be estimated by its empirical counterpart. Thus we may use
$$\hat{\Psi}_m=\frac{1}{n}\sum_{t=1}^n\left\lbrace \left(\hat{e}_{t-1},\dots,\hat{e}_{t-m} \right)^{'}\frac{\partial\hat{e}_{t}}{\partial\theta^{'}} \right\rbrace  \ \ \text{ and } \ \ \hat{\sigma}_{\epsilon}^2=\hat\gamma(0)=\frac{1}{n}\sum_{t=1}^n\hat{e}_t^2.$$
A consistent estimator of $\Xi$ is obtained by means of an autoregressive spectral estimator, as in \cite{BMES2019} (see also \cite{berk}, \cite{BMCF2012} and \cite{haan}, to name a few for a more comprehensive exposition of this method). The stationary process $(U_t)_{t\in\mathbb{Z}}$ admits the Wold decomposition
\begin{align*}
U_t=v_t+\sum_{i=1}^\infty\varpi_i v_{t-i},
\end{align*}
where $(v_t)_{t\in\mathbb{Z}}$ is a $(p+q+1+m)$-variate weak white noise. Assume that the covariance matrix $\Sigma_v :=\mathrm{Var}(v_t)$ is non-singular, that $\sum_{i=1}^\infty\|\varpi_i\|<\infty$, where $\|\cdot\|$ denotes any norm on the space of the real $(p+q+1+m)\times (p+q+1+m)$ matrices, and that $\det\left\lbrace I_{p+q+1+m}+\sum_{i=1}^\infty\varpi_iz^i\right\rbrace \neq 0$ if $\left|z\right|\leq 1$. Then $(U_t)_{t\in\mathbb{Z}}$ admits an $\mathrm{AR}(\infty)$ representation (see \cite{Akutowicz1957}) of the form
\begin{equation}\label{U_t_inf}
\Delta(L)U_t:=U_t-\sum_{i=1}^{\infty}\Delta_iU_{t-i}=v_t,
\end{equation}
such that $\sum_{i=1}^\infty\|\Delta_i\|<\infty$ and $\det\left\lbrace \Delta(z)\right\rbrace \neq 0$ if $\left|z\right|\leq 1$. In view of \eqref{eq:specU}, the matrix $\Xi$ can be interpreted as $2\pi$ times the spectral density of the stationary process $(U_t)_{t\in\mathbb{Z}}=( (U_{1t}^{'},U_{2t}^{'})^{'})_{t\in\mathbb{Z}}$ evaluated at frequency 0 (see p. 459 of \cite{broc-d}). We then obtain that
$$\Xi=\Delta^{-1}(1)\Sigma_{v}\Delta^{' -1}(1)$$
%when $(U_t)_{t\in\mathbb{Z}}$ satisfies an $\mathrm{AR}(\infty)$ representation of the form
%\begin{equation}\label{U_t_inf}
%\Delta(L)U_t:=U_t-\sum_{k=1}^{\infty}\Delta_kU_{t-k}=v_t,
%\end{equation}
%such that $\sum_{k=1}^{\infty}\left\|\Delta_k\right\|<\infty$ and $\det\left\lbrace \Delta(z)\right\rbrace \neq 0$ for all $\left|z\right|\leq 1$ and where $(v_t)_{t\in\mathbb{Z}}$ is a $(p+q+1+m)$-variate weak white noise with variance matrix $\Sigma_v$. It is proved in \cite{yac-these,luk} that one may find a constant $K$ and $0<\rho<1$ such that
%\begin{align}\label{speed}
%\left\|\Delta_k\right\| & \le K\, \rho^k  .
%\end{align}
Since $U_t$ is unobservable,  we introduce  $\hat{U}_t\in\mathbb{R}^{p+q+1+m}$  obtained by replacing $\epsilon_t(\theta_0)$ by $\tilde{\epsilon}_t(\hat{\theta}_n)$ and $J(\theta_0)$ by its empirical or observable counterpart $\hat{J}_n$ in \eqref{U_t}. Let $\hat{\Delta}_r(z)=I_{p+q+1+m}-\sum_{k=1}^r\hat{\Delta}_{r,k}z^k$, where $\hat{\Delta}_{r,1},\dots,\hat{\Delta}_{r,r}$ denote the coefficients of the least squares regression of $\hat{U}_t$ on $\hat{U}_{t-1},\dots,\hat{U}_{t-r}$. Let $\hat{v}_{r,t}$ be the residuals of this regression, and let $\hat{\Sigma}_{\hat{v}_r}$ be the empirical variance of $\hat{v}_{r,1},\dots,\hat{v}_{r,n}$.
We are now able to state Theorem~\ref{estXi} which is an extension of a result given in \cite{BMCF2012}.
\begin{thm} \label{estXi}
We assume \textbf{(A0)}, \textbf{(A1)}, \textbf{(A2)} and Assumption {\bf (A3')} with  $\tau=8$. In addition, we assume that the innovation process $(\epsilon_t)_{t\in\mathbb{Z}}$ of the FARIMA$(p,d_0,q)$ model \eqref{FARIMA} is such that the process $(U_t)_{t\in\mathbb{Z}}$ defined in \eqref{U_t} admits a multivariate AR$(\infty)$ representation \eqref{U_t_inf}, where $\|\Delta_i\|=\mathrm{o}(i^{-2})$ as $i\to\infty$, the roots of $\det(\Delta(z))=0$ are outside the unit disk, and $\Sigma_v=\mathrm{Var}(v_t)$ is non-singular.
Then the spectral estimator of $\Xi$ satisfies
$$\hat{\Xi}^{\mathrm{SP}}_n:=\hat{\Delta}_r^{-1}(1)\hat{\Sigma}_{\hat{v}_r}\hat{\Delta}_r^{'-1}(1)\xrightarrow[]{} \Xi=\Delta^{-1}(1)\Sigma_v\Delta^{-1}(1)$$
in probability when $r=r(n)\to\infty$ and $r^5(n)/n^{1-2(d_0-d_1)}\to0$ as $n\to\infty$ (remind that  $d_0\in [ d_1{,}d_2]\subset] -1/2{,}1/2[$).
\end{thm}
The proof of this theorem is similar to the proof of Theorem 3 in \cite{BMES2019} and it is omitted.

We are now in a position to define the modified versions of the Box-Pierce (BP) and Ljung-Box (LB) goodness-of-fit portmanteau tests. The standard versions of the portmanteau tests are
useful tools to detect if the orders $p$ and $q$ of a FARIMA$(p,d_0,q)$ model are well chosen, provided
the error terms $(\epsilon_t)_{t\in\mathbb{Z}}$ of the FARIMA$(p,d_0,q)$ equation be a strong white noise and provided the
number $m$ of residual autocorrelations is not too small (see Remark~\ref{remKhi2}). Now we define the modified
versions which are aimed to detect if the orders $p$ and $q$ of a weak FARIMA$(p,d_0,q)$ model are well
chosen. These tests are also asymptotically valid for strong FARIMA$(p,d_0,q)$ even for small
$m$. The modified versions of the portmanteau tests will be denoted by $\text{BP}_{\textsc{w}}$ and
$\text{LB}_{\textsc{w}}$, the subscript $\textsc{w}$ referring to the term weak.

Let $\hat{\Sigma}_{\hat{\rho}_m}$ be the matrix obtained by replacing $\Xi$ by $\hat{\Xi}$ and $\sigma_{\epsilon}^2$ by $\hat{\sigma}_{\epsilon}^2$ in $\Sigma_{\hat{\rho}_m}$. Denote by $\hat{\xi}_m=(\hat{\xi}_{1,m},\dots,\hat{\xi}_{m,m})^{'}$ the vector of the eigenvalues of $\hat{\Sigma}_{\hat{\rho}_m}$. At the asymptotic level $\alpha$, it holds under the assumptions of Theorem~\ref{loi_res_rho} and {\bf (H0)} that
%the $\text{BP}_{\textsc{w}}$ test (resp. the $\text{LB}_{\textsc{w}}$ test) consists in rejecting the null hypothesis of the weak FARIMA$(p,d_0,q)$ model (the adequacy of the weak FARIMA$(p,d_0,q)$ model) when
$$\lim_{n\to\infty}\mathbb{P}\left({Q}_m^{\textsc{bp}}>S_m(1-\alpha)\right)=\lim_{n\to\infty}\mathbb{P}\left({Q}_m^{\textsc{lb}}>S_m(1-\alpha)\right)=\alpha,$$
where $S_m(1-\alpha)$ is such that $\mathbb{P}(Z_m(\hat{\xi}_m)>S_m(1-\alpha))=\alpha$. We emphasize the fact that the proposed modified versions of the Box-Pierce and Ljung-Box statistics are more difficult to implement because their critical values have to be computed from the data while the  critical values of the standard method are simply
deduced from a $\chi^2$-table.  We shall evaluate the $p$-values $$\mathbb{P}\left\{Z_m(\hat\xi_m)>{Q}_m^{\textsc{bp}}\right\}\;\mbox{ and }\;
\mathbb{P}\left\{Z_m(\hat\xi_m)>{Q}_m^{\textsc{lb}}\right\},\text{with}\  Z_m(\hat\xi_m)=
\sum_{i=1}^{m}\hat\xi_{i,m}Z_i^2,$$ by  means of the Imhof algorithm (see \cite{i1961}).

A second method avoiding the estimation of the asymptotic matrix is proposed in the next Subsection.
\subsection{Self-normalized asymptotic distribution of the residual autocorrelations}\label{diagnostic2}
In view of Theorem~\ref{loi_stat_test},  the asymptotic distributions of the statistics defined in \eqref{stat-test}  are a mixture of chi-squared distributions, weighted by eigenvalues of  the asymptotic covariance matrix $\Sigma_{\hat{\rho}_m}$ of the vector of autocorrelations obtained in Theorem~\ref{loi_res_rho}. However, this asymptotic variance matrix depends on the unknown matrices $\Xi$, $\Psi_m$ and the noise variance $\sigma_{\epsilon}^2$.
Consequently, in order to obtain a consistent estimator of the asymptotic covariance matrix  $\Sigma_{\hat{\rho}_m}$ of the residual autocorrelations vector we have used an autoregressive spectral estimator  of the spectral density of the stationary process $(U_t)_{t\in\mathbb{Z}}$ to get a consistency estimator of the matrix $\Xi$ (see Theorem~\ref{estXi}).
However, this  approach  presents the problem of choosing the truncation  parameter. Indeed this method is based on an infinite autoregressive representation of the stationary process $(U_t)_{t\in\mathbb{Z}}$ (see \eqref{U_t_inf}). So the choice of the order of truncation is crucial and difficult.

In this section, we propose an alternative method where we do not estimate an asymptotic covariance matrix which is an extension to the results obtained by  \cite{BMS2018}. It is based on a self-normalization approach to construct a test-statistic which is asymptotically distribution-free under the null hypothesis. This approach  has been studied by \cite{BMS2018}  in the weak ARMA case, by proposing new portmanteau statistics.  In this case the critical values are not computed from the data since they are tabulated by  \cite{lobato}. In some sense this method is finally closer to the standard method in which the critical values are simply deduced from a $\mathcal{X}^2$-table.
The idea comes from \cite{lobato} and has been already extended by \cite{BMS2018}, \cite{kl2006}, \cite{s2010JRSSBa}, \cite{s2010JRSSBb} and \cite{shaox} to name a few in more general frameworks.
See also \cite{s2016} for a review on some recent developments on the inference of time series data using the self-normalized approach.

Other alternative methods that avoid the estimation of the covariance of the
parameter estimates by directly eliminating the estimation effect of the test statistics can be found
in \cite{DV2011} or \cite{VW2015}.
\cite{DV2011} developed  an asymptotically distribution-free transform of the sample
autocorrelations of residuals in general parametric linear time-series models and 
shown that the proposed Box-Pierce-type test statistic based on the transformed autocorrelation is not affected by the estimation effect. \cite{VW2015} proposed an asymptotic simultaneous distribution-free transform of the sample autocorrelations of standardized residuals and their squares, which extended the approach developed by \cite{DV2011} to the conditional mean and variance models diagnosis.

We denote by $\Lambda$ the block matrix of $\mathbb{R}^{m\times (p+q+1+m)}$ defined by $\Lambda=(\Psi_m|I_m)$. In view of \eqref{hat_gamma} and \eqref{ecart-theta} we deduce that
$$\sqrt{n}\hat{\gamma}_m=\frac{1}{\sqrt{n}}\sum_{t=1}^n\Lambda U_t+\mathrm{o}_{\mathbb{P}}(1).$$
At this stage, we do not rely on the classical method that would consist in estimating
the asymptotic covariance matrix $\Xi$. We rather try to apply Lemma 1 in \cite{lobato}. So
we need to check that a functional central limit theorem holds for the process $U:=(U_t)_{t\geq 1}$.
For that sake, we define the normalization matrix  $C_m$ of $\mathbb{R}^{m\times m}$ by
$$C_m=\frac{1}{n^2}\sum_{t=1}^nS_tS_t^{'} \text{ where } S_t=\sum_{j=1}^{t}\left(\Lambda U_j-{\gamma}_m\right).$$
To ensure the invertibility of the normalization matrix $C_m$ (it is the result stated in the next proposition), we need the following technical assumption on the distribution of $\epsilon_t$.\begin{itemize}
\item[\hspace*{1em} {\bf (A4):}]
\hspace*{1em} The process $(\epsilon_t)_{t\in\mathbb{Z}}$ has a positive density on some neighbourhood of zero.
\end{itemize}
%The following Proposition give the invertibility of the normalization matrix $C_m$.
\begin{prop}\label{inversibleCm}
Under the assumptions of Theorem~\ref{loi_res_rho} and {\bf (A4)}, the matrix $C_{m}$ is almost surely non singular.
 \end{prop}
The proof of this proposition is given in Subsection~\ref{proof_invCm} of the appendix.

Let $(B_K(r))_{r\geq 0}$ be a $K$-dimensional Brownian motion starting from 0. For $K\geq 1$, we denote by $\mathcal{U}_K$ the random variable defined by:
\begin{equation}\label{Uk}
\mathcal{U}_K=B_K^{'}(1)V_K^{-1}B_K(1),
\end{equation}
where
\begin{equation}\label{Vk}
V_K=\int_0^1\left(B_K(r)-rB_K(1)\right)\left(B_K(r)-rB_K(1)\right)^{'}\mathrm{dr}.
\end{equation}
The critical values of $\mathcal{U}_K$ have been tabulated by \cite{lobato}.

The following theorem states the asymptotic distributions of the sample autocovariances and autocorrelations.
\begin{thm}\label{sn1} Under the assumptions of Theorem~\ref{loi_res_rho}, {\bf (A4)} and under the null hypothesis {\bf (H0)} we have
$$n\hat{\gamma}_m^{'}C_m^{-1}\hat{\gamma}_m\xrightarrow[n\to\infty]{\text{in law}}\mathcal{U}_m \text{ and } n\sigma_{\epsilon}^4\hat{\rho}_m^{'}C_m^{-1}\hat{\rho}_m\xrightarrow[n\to\infty]{\text{in law}}\mathcal{U}_m.$$
\end{thm}
%\end{thm}
The proof of this theorem is given in Subsection~\ref{proof_sn1} of Appendix.

Of course, the above theorem is useless for practical purpose  because the normalization matrix $C_m$ and the nuisance parameter $\sigma_{\epsilon}^2$ are not observable.
This gap will be fixed below (see Theorem~\ref{sn2}) when one replaces the matrix $C_m$ and the scalar $\sigma_{\epsilon}^2$ by their empirical or observable counterparts. Then we denote
$$\hat{C}_m=\frac{1}{n^2}\sum_{t=1}^n\hat{S}_t\hat{S}_t^{'} \text{ where } \hat{S}_t=\sum_{j=1}^{t}\left(\hat{\Lambda} \hat{U}_j-\hat{\gamma}_m\right),$$
with $\hat{\Lambda}=(\hat{\Psi}_m|I_m)$ and where $\hat{U}_t$ and $\hat\sigma_{\epsilon}^2$ are defined in Subsection~\ref{diagnostic1}.

The above quantities are all observable and the following result is the applicable counterpart of Theorem~\ref{sn1}.
\begin{thm}\label{sn2} Under the assumptions of Theorem~\ref{sn1}, we have
$$n\hat{\gamma}_m^{'}\hat{C}_m^{-1}\hat{\gamma}_m\xrightarrow[n\to\infty]{\text{in law}}\mathcal{U}_m \text{ and } {Q}_m^\textsc{sn}=n\hat{\sigma}_{\epsilon}^4\hat{\rho}_m^{'}\hat{C}_m^{-1}\hat{\rho}_m\xrightarrow[n\to\infty]{\text{in law}}\mathcal{U}_m.$$
\end{thm}
The proof of this result is postponed in Subsection~\ref{proof_sn2} of Appendix.

Based on the above result, we propose a modified version of the Ljung-Box statistic when one uses the statistic
$$\tilde{Q}_m^\textsc{sn}=n\hat{\sigma}_{\epsilon}^4\hat{\rho}_m^{'}D^{1/2}_{n,m}\hat{C}_m^{-1}D^{1/2}_{n,m}\hat{\rho}_m,$$
where $D_{n,m}\in\mathbb{R}^{m\times m}$ is diagonal with $(n+2)/(n-1),\dots,(n+2)/(n-m)$ as diagonal terms.
These modified versions of the portmanteau tests will be denoted by $\text{BP}_{\textsc{sn}}$ and
$\text{LB}_{\textsc{sn}}$, the subscript $\textsc{sn}$ referring to the term self-normalized.

\section{Numerical illustrations}\label{num-ill}
In this section, by means of Monte Carlo experiments, we investigate the finite sample
properties of the asymptotic results that we introduced in this work.
The numerical illustrations of this section are made  with the open source
statistical software R (see http://cran.r-project.org/).
\subsection{Simulation studies and empirical sizes}\label{simul}
We study numerically the behavior of the least squares estimator for FARIMA models of the form
\begin{equation}\label{process-sim}
(1-L)^{d_0}\left(X_t-aX_{t-1}\right)=\epsilon_t-b\epsilon_{t-1},
\end{equation}
where the  unknown parameter is  $\theta_0=(a,b,d_0)$.  First we assume that in \eqref{process-sim} the innovation process $(\epsilon_t)_{t\in\mathbb{Z}}$ is an iid centered Gaussian process with common variance 1
which corresponds to the strong FARIMA case. For the weak FARIMA case, we consider that in \eqref{process-sim} the innovation process $(\epsilon_t)_{t\in\mathbb{Z}}$ follows firstly a GARCH$(1,1)$ process given by the model 
% We assume that in \eqref{ex_FARIMA} the innovation process $(\epsilon_t)_{t\in\mathbb{Z}}$
%is a GARCH$(1,1)$ process given by the model
\begin{equation} \label{garch}
\left\{\begin{array}{l}\epsilon_{t}=\sigma_t\eta_{t}\\
\sigma_t^2=\omega+\alpha_1\epsilon_{t-1}^2+\beta_1\sigma_{t-1}^2,
\end{array}\right.
\end{equation}
with $\omega>0$, $\alpha_1\ge0$ and where $(\eta_t)_{t\in\mathbb{Z}}$ is a sequence of iid centered Gaussian random variables with variance 1.
%We also assume that $\alpha_1^2\kappa+\beta_1^2+2\alpha_1\beta_1<1$,\footnote{This  is a necessary and  sufficient condition for the existence of a nonanticipative stationary solution
% process $(\epsilon_t)_{t\in\mathbb{Z}}$ with fourth-order moments (see \cite[Example 2.3]{FZ2010}).} where $\kappa:=\mathbb{E}\eta_1^4$ and we assume that $\kappa>1$.
Secondly we consider that in \eqref{process-sim} a noise defined by
\begin{equation} \label{noise-sim}
\epsilon_{t}=\eta_t^2\eta_{t-1}.
\end{equation}
%where  $(\eta_t)_{t\in\mathbb{Z}}$ is a sequence of iid centered Gaussian random variables with variance 1.
The example \eqref{noise-sim} is an extension of a noise process in \cite{rt1996}.
Contrary to the GARCH$(1,1)$ process, the noise defined in Equation \eqref{noise-sim}
is not a martingale difference sequence for which the limit theory
is more classical.

We simulate  $N=1,000$ independent trajectories of size $n=10,000$ of models \eqref{process-sim}. The same series is  partitioned as three series of sizes $n=1,000$, $n=5,000$ and $n=10,000$.
For each of these $N$ replications,  we use the least squares estimation method to estimate the coefficient $\theta_0$ and we apply portmanteau tests to the residuals for different values of $m\in\{1,2,3,6,12,15\}$, where $m$ is the number of autocorrelations used in the portmanteau test statistic.
For the nominal level $\alpha=5\%$, the empirical  size over the $N$ independent replications should
vary between the significant limits 3.6\% and 6.4\% with probability
95\%.
When the relative rejection frequencies  are outside the 95\%
significant limits, they are displayed in bold type in Tables %\ref{tab1FARIMA}, \ref{tab2FARIMA}, \ref{tab3FARIMA}, 
\ref{tabFARIMA00f}, \ref{tabFARIMA00garch} and \ref{tabFARIMA00pt}.

For the standard  Box-Pierce test, the model is therefore rejected when the statistic $Q_m^{\textsc{bp}}$ or $Q_m^{\textsc{lb}}$ is larger than $\chi_{(m-p-q-1)}^2(0.95)$  in a
FARIMA$(p,d_0,q)$ case (see \cite{LiMcL1986}). Consequently the empirical size is not available (n.a.) for the statistic $Q_m^{\textsc{bp}}$ or $Q_m^{\textsc{lb}}$
because they are not applicable for $m\leq p+q+1$.
For the proposed self-normalized test $\mathrm{BP}_\textsc{sn}$ or $\mathrm{LB}_\textsc{sn}$, the model is rejected when the statistic $
Q_m^\textsc{sn}$ or $\tilde{Q}_m^\textsc{sn}$ is larger than $\mathcal{U}_m(0.95)$, where the critical values $\mathcal{U}_K(0.95)$ (for $K=1,\dots,20$) are tabulated  in Lobato (see Table 1 in \cite{lobato}).

Table \ref{tabFARIMA00f} displays the relative rejection frequencies of the null hypothesis {\bf (H0)} that the data generating process (DGP for short) follows a
strong FARIMA$(0,d_0,0)$ model \eqref{process-sim}, over the $N$ independent replications.
For all tests, the percentages of rejection belong globally to the confident interval with probabilities 95\%, except for $\mathrm{{LB}}_\textsc{s}$ and ${\mathrm{BP}}_\textsc{s}$ (see Table  \ref{tab1FARIMA}).
%%when $m\le 6$ (resp. $m\le 2$).
% Consequently all these tests well control the error of first kind.
%
%We draw the conclusion that in these strong FARIMA cases the proposed modified version may be clearly preferable to the standard ones.

Now, we repeat the same experiments on two weak FARIMA models.
As expected Tables \ref{tabFARIMA00garch} and \ref{tabFARIMA00pt}  show that the standard  $\mathrm{{LB}}_\textsc{s}$ or $\mathrm{{BP}}_\textsc{s}$ test poorly performs in assessing the adequacy of these particular weak FARIMA models.
Indeed, we observe that the observed relative rejection frequencies of $\mathrm{{LB}}_\textsc{s}$ and $\mathrm{{BP}}_\textsc{s}$ are definitely outside the significant limits.
Thus we draw the conclusion that the error of the first kind is globally well controlled by all the tests in the strong case, but only by  the proposed  tests in the weak cases.

%We also investigate the case where the GARCH model \eqref{garch} have infinite fourth moments. As showing in Figure~\ref{fig6} the results are qualitatively similar to what we observe here
%in Tables \ref{tabFARIMA00garch} and \ref{tabFARIMA00pt}.
%
%Figure~\ref{fig6} displays the residual autocorrelations  of a realization of size $n=2,000$ for weak FARIMA models  \eqref{process-sim}--\eqref{garch} with $\omega=0.04$, $\alpha_1=0.13$,  $\beta_1=0.88$ and  $d_0=0.49$, and their 5\% significance limits under the strong  FARIMA
%and weak  FARIMA assumptions. This figure confirms clearly the above conclusions.
%% in Subsection \ref{simul}.
%The horizontal dotted lines (blue color) correspond to the 5\% significant limits obtained under the strong FARIMA assumption.
%The solid lines (red color) and  dashed lines (green color) correspond also  to the 5\% significant limits under the weak FARIMA assumption.
%The full lines  correspond to the asymptotic significance limits for the residual autocorrelations
%obtained in  Theorem~\ref{loi_res_rho}. The dashed lines (green color) correspond to the self-normalized asymptotic significance limits for the residual autocorrelations
%obtained in Theorem~\ref{sn2}.
\subsection{Empirical power}\label{emp-power}
In this section, we repeat the same experiments as in Section \ref{simul} to examine the power of the tests for the null hypothesis of Model \eqref{process-sim} with $a=b=0$ ({\em i.e.} a FARIMA$(0,d_0,0)$) against the FARIMA$(0,d_0,1)$ alternative defined by  Model \eqref{process-sim} with
 $\theta_0=(0,b,d_0)$ and where the innovation process  $(\epsilon_t)_{t\in\mathbb{Z}}$ follows the two weak white noises introduced in Section \ref{simul}.

For each of these $N$ replications we fit a  FARIMA$(0,d_0,0)$ model \eqref{process-sim}
and perform standard and modified tests based on $m=1,2,3$, $6$, $12$ and $15$ residual autocorrelations.

Tables \ref{pFARIMA00garch} and \ref{pFARIMA00pt} compare the empirical powers of Model \eqref{process-sim} with  $\theta_0=(0,0.2,d_0)$ over the $N$ independent replications.
For these particular weak FARIMA models,  we notice that the standard $\mathrm{{BP}}_{\textsc{s}}$ and
$\mathrm{{{LB}}_{\textsc{s}}}$ and our proposed  tests have very similar powers except for $\mathrm{{BP}}_{\textsc{sn}}$ and $\mathrm{{{LB}}_{\textsc{sn}}}$ when  $n=5,000$. % in the weak case.

In these Monte Carlo experiments, we illustrate that the proposed test statistics  have reasonable finite sample performance. Under nonindependent errors, it appears that the standard test statistics are generally non reliable, overrejecting  severely, while the proposed tests statistics offer satisfactory levels.
Even for independent errors, they seem  preferable to the standard ones when the number $m$ of autocorrelations is small. Moreover, the error of first kind is well controlled.
Contrarily to the standard tests  based on $\mathrm{{BP}}_{\textsc{s}}$ or $\mathrm{{{LB}}_{\textsc{s}}}$, the proposed tests can be used safely for $m$ small.% (see for instance Figure~\ref{fig6}).
For all these above reasons, we think that the modified
versions that we propose in this paper are preferable to the standard ones for diagnosing FARIMA models under nonindependent errors.

\subsection{Illustrative example}\label{ill}
We now consider an application to the daily $\log$ returns (also simply called the returns) of the Nikkei and
Standard \& Poor's 500 indices (S\&P 500, for short).
The returns are defined by $r_t=100\log(p_t/p_{t-1})$ where $p_t$ denotes the price index  of the S\&P 500 index  at time $t$.
The observations of the S\&P 500 (resp. the Nikkei) index cover the period from January 3, 1950 to to February 14, 2019 (resp.
from January 5, 1965 to February 14, 2019).
The length of the series is $n=17,391$ (resp. $n=13,319$) for the S\&P 500 (resp. the Nikkei) index. The data can be downloaded
from the website Yahoo Finance: http://fr.finance.yahoo.com/.
%Figure~\ref{graph} (resp. Figure~\ref{nikkei}) plots the returns and the sample autocorrelations of squared returns of the S\&P 500 (resp. of the Nikkei).
%Figure~\ref{graph} also shows that the squared returns $(r_t^2)_{t\ge 1}$ are generally strongly autocorrelated.

In Financial Econometrics the returns are often assumed to be  a white noise.
In view of the so-called volatility clustering, it is well known that the
strong white noise model is not adequate for these series
(see for instance \cite{FZ2010,lobatoNS2001,BMCF2012,BMS2018}).
%To summarize, the common belief that these series are not strong white noises, but could be weak white noises.

A long-range memory property of the stock market returns series was largely investigated by \cite{DGE1993} which shown that there are more correlation beetwen power transformation of the absolute return $|r_t|^v$ ($v>0$) than returns themselves (see also \cite{Beran2013}, \cite{palma}, \cite{baillie1996} and \cite{ling-li}). We choose here the case where $v=2$ which corresponds to the squared returns
$(r_t^2)_{t\ge 1}$ process.
The mean and the standard deviation of $(r_t^2)_{t\ge 1}$ are $0,9347$ and $5,0036$  (resp. $1,6167$ and $5,4759$) for the S\&P 500 (resp. the Nikkei) index. Following a similar way as in \cite{Ling2003} we denote by $(X_t)_{t\ge 1}$ the centered series of the squared returns, that is, $X_t=r_t^2-0,9347$
(resp. $X_t=r_t^2-1,6167$) for the S\&P 500 (resp. the Nikkei) index.
Figure~\ref{graph} (resp. Figure~\ref{nikkei}) plots the returns and the sample autocorrelations of 
$(X_t)_{t\ge 1}$ of the S\&P 500 (resp. of the Nikkei).
The centered squared returns $(X)_{t\ge 1}$  have significant positive autocorrelations at least up to lag 80 (see Figure~\ref{graph} and Figure~\ref{nikkei})  which confirm  the claim that stock market returns have long-term memory (see for instance \cite{DGE1993}, for more details). 
%In particular the returns $(r_t)_{t\ge 1}$ process is characterized by substantially more correlation between absolute or squared returns than between the returns themselves.

We  first fit a FARIMA$(1,d_0,1)$ model defined in \eqref{process-sim} to the process  $(X)_{t\ge 1}$  of the S\&P 500 and the Nikkei returns. 
%
%
%Denoting by $(X_t)_{t\ge 1}$ the mean corrected
%series of the squared returns, we adjust  the following model
%\begin{align}\label{FARIMA-sp500}
%& (1-L)^d(X_t-aX_{t-1})=\epsilon_t-b\epsilon_{t-1} .
%\end{align}
Let $\hat{\theta}_n^\mathrm{SP500}$ and $\hat{\theta}_n^\mathrm{Nikkei}$ be respectively the least squares estimators of the parameter $\theta_0=(a,b,d_0)$ for the model \eqref{process-sim}
in the case of the S\&P 500 and the Nikkei.  The least squares estimators  were obtained as
\begin{eqnarray*}
\hat{\theta}_n^\mathrm{SP500}=\left(\begin{array}{ccc} -0.3371&[0.1105]&(0.0023)\\ -0.1795 &[0.0788]&(0.0227)\\ \quad 0.2338&[0.0367]& (0.0000)\end{array}\right)\text{ and }
\hat\sigma_{\epsilon}^2=22.9076\times10^{-8}
\end{eqnarray*}
and
\begin{eqnarray}\label{LSE-nikkei}
\hat{\theta}_n^\mathrm{Nikkei}=\left(\begin{array}{ccc} -0.0217 &[0.1105]&(0.9528)\\ \quad 0.1579&[0.0788]&(0.6050)\\ \quad 0.3217&[0.0367]& (0.0000)\end{array}\right)\text{ and }
\hat\sigma_{\epsilon}^2=25.6844\times10^{-8},
\end{eqnarray}
where the estimated asymptotic standard errors obtained from $\Sigma_{\hat\theta}:=J^{-1}IJ^{-1}$ (respectively the
$p$-values), of the estimated parameters (first column), are given into brackets
(respectively in parentheses).
Note that for these series, the estimated coefficients $|\hat{a}_n|$ and $|\hat{b}_n|$ are smaller than one. This is in accordance with
the assumptions that the power series $a_\theta^{-1}$ and $b_\theta^{-1}$ are well defined (remind that the moving average polynomial is denoted $b_\theta$ and the autoregressive polynomials $a_\theta$).
We also observe that the estimated long-range dependence coefficients $\hat{d}_n$ is significant for any reasonable asymptotic level and is inside $]-0.5,0.5[$.
So we think that the assumption {\bf (A2)} is  satisfied and thus our asymptotic normality theorem on the residual autocorrelations can be applied.

Concerning the S\&P 500, the estimators of the parameters $a$ and $b$ are significant whereas it is not the case for the Nikkei (see \eqref{LSE-nikkei}). In the Nikkei case,  the coefficients  could reasonably
be set to zero. So we adjust a FARIMA$(0,d_0,0)$ for the squares of Nikkei returns and \eqref{LSE-nikkei} is reduced as
\begin{eqnarray*}
\hat{\theta}_n^\mathrm{Nikkei}=\left(\begin{array}{ccc} 0.2132 &[0.0259]&(0.0000)\end{array}\right)\text{ and }
\hat\sigma_{\epsilon}^2=25.9793\times10^{-8}.
\end{eqnarray*}
We thus apply portmanteau tests to the residuals of FARIMA$(1,d_0,1)$ (resp. FARIMA$(0,d_0,0)$) model for the process $(X)_{t\ge 1}$  of S\&P 500 (resp. of Nikkei).
Table~\ref{sp500FARIMA11} (resp. Table~\ref{nikkeiFARIMA11}) displays the statistics and the $p$-values of the standard and modified versions of BP and LB tests of model \eqref{process-sim}.
%(resp. of FARIMA$(0,d_0,0)$).
From Tables~\ref{sp500FARIMA11} and~\ref{nikkeiFARIMA11}, we draw the conclusion that the strong FARIMA$(1,0.2338,1)$ and FARIMA$(0,0.2132,0)$ models are
rejected  but the weak FARIMA$(1,0.2338,1)$  and FARIMA$(0,0.2132,0)$ models are not rejected.

Figure~\ref{acf} (resp. Figure~\ref{acfnikkei}) displays the residual autocorrelations and their 5\% significance limits under the strong FARIMA and weak FARIMA assumptions.
In view of Figures~\ref{acf} and \ref{acfnikkei}, the diagnostic
checking of residuals does not indicate any inadequacy for the proposed tests.
All of the sample autocorrelations should lie between the bands (at 95\%)
shown as dashed lines (green color) and solid lines (red color) for the modified tests, while the horizontal dotted (blue color) for standard test indicate that
strong FARIMA is not adequate.
Figure~\ref{acf} (resp. Figure~\ref{acfnikkei}) confirms the conclusions drawn from Table~\ref{sp500FARIMA11} (resp. Table~\ref{nikkeiFARIMA11}).

\newpage

\section{Figures and tables}\label{sec-fig}
%\clearpage

\begin{table}[h]
 \caption{\small{Empirical size (in \%) of the modified and standard versions
 of the LB and BP tests in the case of a strong FARIMA$(0,d_0,0)$  defined by \eqref{process-sim}
 with $\theta_0=(0,0,d_0)$.
The nominal asymptotic level of the tests is $\alpha=5\%$.
The number of replications is $N=1,000$. }}
\begin{center}
%\begin{tabular}{lll rrr rrr}
{\scriptsize
\begin{tabular}{c c ccc ccc c}
\hline\hline \\
$d_0$& Length $n$ & Lag $m$ & $\mathrm{{LB}}_{\textsc{sn}}$&$\mathrm{BP}_{\textsc{sn}}$&$\mathrm{{LB}}_{\textsc{w}}$&$\mathrm{BP}_{\textsc{w}}$&$\mathrm{{LB}}_{\textsc{s}}$&$\mathrm{BP}_{\textsc{s}}$
\vspace*{0.2cm}\\\hline
&& $1$&\textbf{3.3}   &\textbf{3.3}   &4.6   &4.5&n.a. &n.a.\\
&& $2$&4.5& 4.5 &4.9 &4.9& 5.8 &5.8\\
0.05 &$n=1,000$& $3$&5.2 &5.1 &4.7 &4.4 &4.9& 4.8\\
 &&$6$& 5.8& 5.8 &4.6 &4.5 &5.1 &5.0\\
 &&$12$&6.0& 5.6 &5.3 &4.6 &5.2 &5.0\\
 &&$15$&5.6 &5.2 &4.7 &4.3 &5.3 &4.7\\
  \cline{2-9}
&& $1$&\textbf{6.8}   &\textbf{6.8}  & \textbf{6.6}  & \textbf{6.6}&n.a. &n.a.\\
&& $2$&\textbf{6.8}& \textbf{6.8}& 6.4& 6.4 &\textbf{7.9}&\textbf{7.9}\\
0.05 &$n=5,000$& $3$&\textbf{6.6}& \textbf{6.6}& 5.7& 5.7& 5.8& 5.8\\
 &&$6$&\textbf{6.5}& 6.4 &5.6 &5.6 &5.7 &5.6\\
 &&$12$&6.4& 6.4& 5.3& 5.3& 6.0& 5.9\\
 &&$15$&6.1& 6.0& 4.7 &4.6 &5.3 &5.2\\
  \cline{2-9}
&& $1$&4.9  & 4.9  & 5.3   &5.3&n.a. &n.a.\\
&& $2$&5.4& 5.4& 6.6& 6.6&\textbf{7.8} &\textbf{7.8}\\
0.05 &$n=10,000$& $3$&5.7& 5.7& 5.9& 5.9& 6.2 &6.2\\
 &&$6$&5.9 &5.8 &4.5 &4.5 &4.6 &4.6\\
 &&$12$&5.3& 5.3& 5.4& 5.4& 5.6& 5.6\\
 &&$15$&4.4 &4.3 &4.8 &4.8 &4.9 &4.9\\

\hline
&& $1$&3.6   &\textbf{3.5}   &4.3  & 4.3 &n.a. &n.a.\\
&& $2$&4.7 &4.7& 4.7& 4.7& 5.8& 5.7\\
0.20 &$n=1,000$& $3$&5.2 &5.0 &4.3& 4.3 &4.9 &4.7\\
 &&$6$&6.0 &5.9 &4.7 &4.5 &5.0 &4.9\\
 &&$12$&5.7 &5.4 &5.3 &4.7 &5.2 &4.9\\
 &&$15$&5.9 &5.6 &4.8& 4.2 &5.2 &4.8\\

  \cline{2-9}
&& $1$&\textbf{6.6}   &\textbf{6.6}   &\textbf{6.5}  & \textbf{6.5}&n.a. &n.a.\\
&& $2$&\textbf{6.6} &\textbf{6.6}& 6.4 &6.4 &\textbf{7.9} &\textbf{7.9}\\
0.20 &$n=5,000$& $3$&\textbf{6.7}& \textbf{6.7}& 5.7 &5.7 &5.8 &5.8\\
 &&$6$&6.3 &6.3 &5.6 &5.6 &5.7 &5.5\\
 &&$12$&6.3 &6.2 &5.5 &5.3 &6.0 &5.9\\
 &&$15$&6.1& 5.9 &4.7 &4.6 &5.3 &5.2\\
  \cline{2-9}
&& $1$& 4.8   &4.8   &5.3  & 5.3&n.a. &n.a.\\
&& $2$&5.4 &5.4 &\textbf{6.6} &\textbf{6.6}& \textbf{7.8}& \textbf{7.8}\\
0.20 &$n=10,000$& $3$&5.5& 5.5 &5.9 &5.9& 6.3 &6.3\\
 &&$6$&5.8 &5.8 &4.5 &4.5 &4.6& 4.6\\
 &&$12$&5.4 &5.3 &5.5 &5.5 &5.6 &5.6\\
 &&$15$&4.4 &4.3 &4.7 &4.7 &4.9 &4.9\\

\hline
&& $1$&3.9   &3.8  & 4.9  & 4.9 &n.a. &n.a.\\
&& $2$&5.1& 5.0& 4.8& 4.6& 5.9& 5.9\\
0.45 &$n=1,000$& $3$&5.2 &5.2 &4.3 &4.3& 4.8 &4.8\\
 &&$6$&6.2& 6.0 &4.7& 4.3 &4.9 &4.9\\
 &&$12$&5.8& 5.4 &4.8 &4.7 &4.9 &4.8\\
 &&$15$&5.6& 5.5& 4.5 &4.2 &5.0 &4.8\\
  \cline{2-9}
&& $1$&\textbf{6.6}   &\textbf{6.6}   &\textbf{6.6}  & \textbf{6.6}&n.a. &n.a.\\
&& $2$&\textbf{6.7} &\textbf{6.7} &\textbf{6.5} &\textbf{6.5} &\textbf{8.0} &\textbf{8.0}\\
0.45 &$n=5,000$& $3$&\textbf{6.6} &\textbf{6.6} &5.7& 5.7 &5.8 &5.8\\
 &&$6$&6.3 &6.3& 5.4 &5.4 &5.6 &5.5\\
 &&$12$&6.2& 6.2 &5.5 &5.5 &6.0 &5.9\\
 &&$15$&6.2& 5.9& 4.6& 4.6 &5.5 &5.3\\
  \cline{2-9}
&& $1$&5.0   &5.0   &5.3   &5.3&n.a. &n.a.\\
&& $2$&5.4 &5.4 &\textbf{6.6} &\textbf{6.6} &\textbf{7.9} &\textbf{7.9}\\
0.45 &$n=10,000$& $3$&5.3 &5.3 &5.9 &5.9 &6.3 &6.3\\
 &&$6$&5.8 &5.8 &4.7 &4.6 &4.7 &4.7\\
 &&$12$&5.4& 5.4 &5.5 &5.5 &5.7 &5.7\\
 &&$15$&4.6 &4.5 &4.9 &4.8 &4.9 &4.9\\
\hline\hline
\\

\end{tabular}
}
\end{center}
\label{tabFARIMA00f}
\end{table}

\begin{table}[h]
 \caption{\small{Empirical size (in \%) of the modified and standard versions
 of the LB and BP tests in the case of a weak FARIMA$(0,d_0,0)$ defined by \eqref{process-sim}
 with $\theta_0=(0,0,d_0)$ and where $\omega=0.4$, $\alpha_1=0.3$ and $\beta_1=0.3$ in \eqref{noise-sim}.
The nominal asymptotic level of the tests is $\alpha=5\%$.
The number of replications is $N=1,000$. }}
\begin{center}
%\begin{tabular}{lll rrr rrr}
{\scriptsize
\begin{tabular}{c c ccc ccc c}
\hline\hline \\
$d_0$& Length $n$ & Lag $m$ & $\mathrm{{LB}}_{\textsc{sn}}$&$\mathrm{BP}_{\textsc{sn}}$&$\mathrm{{LB}}_{\textsc{w}}$&$\mathrm{BP}_{\textsc{w}}$&$\mathrm{{LB}}_{\textsc{s}}$&$\mathrm{BP}_{\textsc{s}}$
\vspace*{0.2cm}\\\hline
&& $1$&4.4   &4.4   &5.4   &5.4&n.a. &n.a.\\
&& $2$&4.3 & 4.2 & 5.7 & 5.7 &\textbf{15.6}& \textbf{15.5}\\
0.05 &$n=1,000$& $3$&5.9  &5.9  &5.3 & 5.0& \textbf{14.2} &\textbf{14.0}\\
 &&$6$& 5.2 & 5.1 & 6.0 & 6.0 &\textbf{14.6} &\textbf{14.4}\\
 &&$12$&4.5 & 4.1  &4.2 & 4.0& \textbf{11.0} &\textbf{10.7}\\
 &&$15$&4.0  &3.9  &4.2  &3.9 &\textbf{11.1} &\textbf{10.6}\\
  \cline{2-9}
&& $1$&4.3  & 4.3   &5.1  & 5.1&n.a. &n.a.\\
&& $2$&4.4 & 4.4  &5.8  &5.8 &\textbf{16.9}& \textbf{16.8}\\
0.05 &$n=5,000$& $3$&5.0 & 5.0  &5.5 & 5.5& \textbf{16.5}& \textbf{16.5}\\
 &&$6$&5.6  &5.6  &4.5 & 4.5 &\textbf{14.8} &\textbf{14.6}\\
 &&$12$&5.1 & 5.1 & 5.0 & 4.9 &\textbf{12.6} &\textbf{12.5}\\
 &&$15$&5.2  &5.1  &4.9 & 4.7& \textbf{11.8}& \textbf{11.6}\\
  \cline{2-9}
&& $1$&5.7   &5.7   &5.3   &5.1&n.a. &n.a.\\
&& $2$&5.0 & 5.0  &4.5 & 4.5& \textbf{17.4} &\textbf{17.4}\\
0.05 &$n=10,000$& $3$&5.5 & 5.5 & 4.7 & 4.6& \textbf{17.2} &\textbf{17.2}\\
 &&$6$&5.3 & 5.3  &5.0  &5.0 &\textbf{14.2} &\textbf{14.1}\\
 &&$12$&4.9 & 4.9 & 4.7 & 4.7& \textbf{11.0}&\textbf{ 11.0}\\
 &&$15$&4.9  &4.8  &4.7 & 4.6& \textbf{10.2}& \textbf{10.2}\\

\hline
&& $1$&4.9   &4.9   &4.3   &4.3 &n.a. &n.a.\\
&& $2$&4.0 & 4.0 & 5.7  &5.6 &\textbf{15.5} &\textbf{15.4}\\
0.20 &$n=1,000$& $3$&6.0 & 6.0 & 5.0 & 4.8& \textbf{14.0}& \textbf{13.8}\\
 &&$6$&5.2  &5.1  &5.7 & 5.6& \textbf{14.3}&\textbf{ 14.2}\\
 &&$12$&4.4  &4.0 & 4.3  &4.0 &\textbf{10.8} &\textbf{10.5}\\
 &&$15$&3.9  &3.8  &4.2 & 3.9 &\textbf{10.8} &\textbf{10.1}\\

  \cline{2-9}
&& $1$&4.3  & 4.3  & 5.0  & 5.0&n.a. &n.a.\\
&& $2$&4.3 & 4.3 & 5.9 & 5.8& \textbf{16.9} &\textbf{16.9}\\
0.20 &$n=5,000$& $3$&5.2 & 5.2 & 5.4 & 5.4 &\textbf{16.7} &\textbf{16.7}\\
 &&$6$&5.6 & 5.5  &4.6  &4.5 &\textbf{14.8}& \textbf{14.7}\\
 &&$12$&5.2 & 5.2  &5.0 & 4.9 &\textbf{12.5} &\textbf{12.4}\\
 &&$15$&5.2 & 5.2  &4.8  &4.6 &11.7 &11.7\\
  \cline{2-9}
&& $1$& 5.7  & 5.7   &5.2  & 5.2&n.a. &n.a.\\
&& $2$&5.1 & 5.1  &4.5 & 4.5& \textbf{17.3} &\textbf{17.3}\\
0.20 &$n=10,000$& $3$&5.7  &5.6  &4.7  &4.7 &\textbf{17.2} &\textbf{17.2}\\
 &&$6$&5.1 & 5.1 & 4.9 & 4.9& \textbf{14.2} &\textbf{14.2}\\
 &&$12$&4.8 & 4.8 & 4.7  &4.7 &\textbf{11.0} &\textbf{11.0}\\
 &&$15$&4.9 & 4.7  &4.6 & 4.6& \textbf{10.2}& \textbf{10.2}\\

\hline
&& $1$&4.5   &4.5   &5.4  & 5.4 &n.a. &n.a.\\
&& $2$&4.1 & 4.1  &6.0  &6.0 &\textbf{16.2} &\textbf{16.1}\\
0.45 &$n=1,000$& $3$&5.9  &5.7  &5.3 & 5.3& \textbf{14.6}& \textbf{14.5}\\
 &&$6$&5.2 & 4.8  &5.5 & 5.4 &\textbf{14.4} &\textbf{14.1}\\
 &&$12$&4.0  &3.7  &4.2 & 4.2 &\textbf{11.2 }&\textbf{10.8}\\
 &&$15$&3.8  &3.7  &4.3 & 3.9 &\textbf{10.6} &\textbf{10.4}\\
  \cline{2-9}
&& $1$&4.6   &4.6   &5.0  & 5.0&n.a. &n.a.\\
&& $2$&4.3  &4.3 & 5.9  &5.9 &\textbf{16.7} &\textbf{16.7}\\
0.45 &$n=5,000$& $3$&4.9  &4.9  &5.4  &5.4 &\textbf{16.8} &\textbf{16.7}\\
 &&$6$&5.7  &5.6  &4.6  &4.6 &\textbf{15.1} &\textbf{14.9}\\
 &&$12$&5.3 & 5.3  &5.1  &5.1 &\textbf{12.7} &\textbf{12.4}\\
 &&$15$&5.1 & 5.0  &4.8 & 4.8 &\textbf{11.7} &\textbf{11.7}\\
  \cline{2-9}
&& $1$&5.7  & 5.7   &5.2  & 5.2&n.a. &n.a.\\
&& $2$&5.0  &5.0  &4.7  &4.7 &\textbf{17.2} &\textbf{17.2}\\
0.45 &$n=10,000$& $3$&5.8 & 5.7  &4.7 & 4.7 &\textbf{17.5} &\textbf{17.4}\\
 &&$6$&5.1 & 5.1  &5.0  &4.9& \textbf{14.3}& \textbf{14.3}\\
 &&$12$&4.8 & 4.8  &4.7  &4.7 &\textbf{10.9} &\textbf{10.9}\\
 &&$15$&4.9  &4.7  &4.6 & 4.6& \textbf{10.2} &\textbf{10.2}\\
\hline\hline
\\

\end{tabular}
}
\end{center}
\label{tabFARIMA00garch}
\end{table}

\begin{table}[h]
 \caption{\small{Empirical size (in \%) of the modified and standard versions
 of the LB and BP tests in the case of weak FARIMA$(0,d_0,0)$  defined by \eqref{process-sim}--\eqref{noise-sim}  with $\theta_0=(0,0,d_0)$.
The nominal asymptotic level of the tests is $\alpha=5\%$.
The number of replications is $N=1,000$. }}
\begin{center}
%\begin{tabular}{lll rrr rrr}
{\scriptsize
\begin{tabular}{c c ccc ccc c}
\hline\hline \\
$d_0$& Length $n$ & Lag $m$ & $\mathrm{{LB}}_{\textsc{sn}}$&$\mathrm{BP}_{\textsc{sn}}$&$\mathrm{{LB}}_{\textsc{w}}$&$\mathrm{BP}_{\textsc{w}}$&$\mathrm{{LB}}_{\textsc{s}}$&$\mathrm{BP}_{\textsc{s}}$
\vspace*{0.2cm}\\\hline
&& $1$&\textbf{3.3}   &\textbf{3.3}   &\textbf{8.7}  & \textbf{8.6}&n.a. &n.a.\\
&& $2$&3.8  &3.7 & 6.1 & 6.1 &\textbf{16.9} &\textbf{16.9}\\
0.05 &$n=1,000$& $3$&\textbf{3.5} & \textbf{3.5} & 4.8  &4.7 &\textbf{14.8} &\textbf{14.8}\\
 &&$6$& \textbf{3.3} & \textbf{3.2}  &4.0 & 4.0 &\textbf{14.1} &\textbf{14.0}\\
 &&$12$&\textbf{1.0} & \textbf{0.9} & \textbf{2.5}  &\textbf{2.4}& \textbf{13.0} &\textbf{12.8}\\
 &&$15$&\textbf{1.0} & \textbf{0.9} & \textbf{2.3} & \textbf{2.1} &\textbf{12.8} &\textbf{12.2}\\
  \cline{2-9}
&& $1$&3.9   &3.9   &5.3   &5.3&n.a. &n.a.\\
&& $2$&4.8  &4.8  &5.2  &5.2& \textbf{18.7} &\textbf{18.7}\\
0.05 &$n=5,000$& $3$&5.6  &5.6  &5.3 & 5.3& \textbf{15.1} &\textbf{15.0}\\
 &&$6$&4.8  &4.8  &4.3 & 4.3 &\textbf{12.4} &\textbf{12.4}\\
 &&$12$&3.9 & 3.9 & \textbf{3.3}  &\textbf{3.3} &\textbf{11.2} &\textbf{11.1}\\
 &&$15$&\textbf{3.5} & \textbf{3.5} & \textbf{2.7} & \textbf{2.7} &\textbf{10.2} &\textbf{10.1}\\
  \cline{2-9}
&& $1$&5.4   &5.4  & 5.2  & 5.2&n.a. &n.a.\\
&& $2$&5.6  &5.6 & 5.3 & 5.3 &\textbf{18.6} &\textbf{18.6}\\
0.05 &$n=10,000$& $3$&4.9&  4.9  &5.3 & 5.2 &\textbf{16.6} &\textbf{16.5}\\
 &&$6$&4.8 & 4.8 & 5.5 & 5.4 &\textbf{13.3} &\textbf{13.3}\\
 &&$12$&4.1 & 4.0 & 4.0 & 4.0 &\textbf{12.2} &\textbf{12.2}\\
 &&$15$&5.0 & 5.0  &\textbf{3.5} & \textbf{3.5} &\textbf{11.2} &\textbf{11.2}\\

\hline
&& $1$&\textbf{3.3}   &\textbf{3.3}   &4.9  & 4.9 &n.a. &n.a.\\
&& $2$&4.2 & 4.1  &4.4 & 4.3 &\textbf{14.7} &\textbf{14.7}\\
0.20 &$n=1,000$& $3$&3.7  &3.7  &\textbf{3.4}&  \textbf{3.2}& \textbf{12.8}& \textbf{12.8}\\
 &&$6$&3.6 & \textbf{3.4} & \textbf{2.7} &\textbf{ 2.7 }&\textbf{12.9} &\textbf{12.8}\\
 &&$12$&\textbf{1.1} & \textbf{1.0} & \textbf{1.9 }& \textbf{1.7} &\textbf{11.8} &\textbf{11.3}\\
 &&$15$&\textbf{0.9}  &\textbf{0.6}  &\textbf{1.8}  &\textbf{1.7} &\textbf{12.0} &\textbf{11.5}\\

  \cline{2-9}
&& $1$&3.8  & 3.8   &5.5   &5.5&n.a. &n.a.\\
&& $2$&4.7 & 4.7 & 5.1 & 5.1& \textbf{18.8}& \textbf{18.8}\\
0.20 &$n=5,000$& $3$&5.8 & 5.8 & 5.2 & 5.2& \textbf{15.0} &\textbf{15.0}\\
 &&$6$&4.9 & 4.9 & 4.3 & 4.3&\textbf{ 12.5}& \textbf{12.4}\\
 &&$12$&3.9 & 3.9 & \textbf{3.4}  &\textbf{3.4} &\textbf{11.1} &\textbf{11.1}\\
 &&$15$&\textbf{3.5} & \textbf{3.3 }& \textbf{2.7} & \textbf{2.7} &\textbf{10.2} &\textbf{10.1}\\
  \cline{2-9}
&& $1$& 5.4  & 5.4  & 5.1  & 5.1&n.a. &n.a.\\
&& $2$&5.6 & 5.6 & 5.3 & 5.3 &\textbf{18.8} &\textbf{18.8}\\
0.20 &$n=10,000$& $3$&5.0 & 5.0 & 5.2 & 5.2 &\textbf{16.6} &\textbf{16.6}\\
 &&$6$&4.8 & 4.8 & 5.4 & 5.4& \textbf{13.3} &\textbf{13.3}\\
 &&$12$&4.0 & 4.0  &4.0 & 4.0& \textbf{12.1}& \textbf{12.1}\\
 &&$15$&5.3 & 5.3 & \textbf{3.4} & \textbf{3.4}& \textbf{11.2} &\textbf{11.2}\\

\hline
&& $1$&\textbf{3.5}  &\textbf{ 3.5}  & \textbf{9.0}  & \textbf{9.0 }&n.a. &n.a.\\
&& $2$&4.1&  4.1 & 5.9 & 5.9& 17.5 &17.5\\
0.45 &$n=1,000$& $3$&3.9 & 3.7 & 5.0 & 4.8& \textbf{15.0} &\textbf{14.6}\\
 &&$6$&\textbf{3.4} & \textbf{3.4} & 3.7 & 3.7 &\textbf{14.1} &\textbf{13.9}\\
 &&$12$&\textbf{0.9} & \textbf{0.9} & \textbf{2.0} & \textbf{2.0} &\textbf{12.9}& \textbf{12.2}\\
 &&$15$&\textbf{1.0} & \textbf{0.5} & \textbf{1.9 }& \textbf{1.7}& \textbf{13.1}& \textbf{12.8}\\
  \cline{2-9}
&& $1$&4.1  & 4.1  & 5.4  & 5.4&n.a. &n.a.\\
&& $2$&4.6 & 4.6 & 5.2 & 5.2& \textbf{18.8}& \textbf{18.7}\\
0.45 &$n=5,000$& $3$&5.6 & 5.6 & 5.2 & 5.2& \textbf{15.2}& \textbf{15.2}\\
 &&$6$&5.1 & 5.0 & 4.4 & 4.4 &\textbf{12.5} &\textbf{12.4}\\
 &&$12$&4.0 & 3.8 & \textbf{3.5} & \textbf{3.5} &\textbf{11.1}& \textbf{11.1}\\
 &&$15$&\textbf{3.5} & \textbf{3.5} & \textbf{2.6} & \textbf{2.6}& \textbf{10.0} & \textbf{9.9}\\
  \cline{2-9}
&& $1$&5.5 &  5.5 &  5.1  & 5.1&n.a. &n.a.\\
&& $2$&5.6 & 5.6 & 5.3&  5.3& \textbf{18.7} &\textbf{18.6}\\
0.45 &$n=10,000$& $3$&4.7 & 4.7 & 5.2 & 5.2 &\textbf{16.6} &\textbf{16.6}\\
 &&$6$&4.8 & 4.8 & 5.3 & 5.3& \textbf{13.3}& \textbf{13.3}\\
 &&$12$&4.0 & 4.0 & 4.0 & 4.0& \textbf{12.1}& \textbf{12.1}\\
 &&$15$&5.2 & 5.2 & \textbf{3.5} & \textbf{3.5}& \textbf{11.1} &\textbf{11.1}\\
\hline\hline
\\

\end{tabular}
}
\end{center}
\label{tabFARIMA00pt}
\end{table}

\begin{table}[h]
 \caption{\small{Empirical power (in \%) of the modified and standard versions of the LB and BP tests in the case of  a weak FARIMA$(0,d_0,1)$  defined by \eqref{process-sim}
 with $\theta_0=(0.,0.2,d_0)$ and where $\omega=0.4$, $\alpha_1=0.3$
and $\beta_1=0.3$ in \eqref{garch}.
The nominal asymptotic level of the tests is $\alpha=5\%$.
The number of replications is $N=1,000$.  }}
\begin{center}
%\begin{tabular}{lll rrr rrr}
{\scriptsize
\begin{tabular}{c c ccc ccc c}
\hline\hline \\
$d_0$& Length $n$ & Lag $m$ & $\mathrm{{LB}}_{\textsc{sn}}$&$\mathrm{BP}_{\textsc{sn}}$&$\mathrm{{LB}}_{\textsc{w}}$&$\mathrm{BP}_{\textsc{w}}$&$\mathrm{{LB}}_{\textsc{s}}$&$\mathrm{BP}_{\textsc{s}}$
\vspace*{0.2cm}\\\hline
%&& $1$&2.9  & 2.9 & 98.8 & 98.8&n.a. &n.a.\\
%&& $2$&13.3& 13.2& 95.0 &95.0 &98.1& 98.1\\
%0.05 &$n=1,000$& $3$&18.9 &18.7& 92.5 &92.4 &97.6 &97.6\\
% &&$6$& 25.6 &25.1 &85.9 &85.9 &95.7 &95.6\\
% &&$12$&20.8 &19.9 &78.8 &78.4& 90.8 &90.6\\
% &&$15$&19.6& 19.1 &74.5& 74.2& 87.7& 87.1\\
%  \cline{2-9}
&& $1$&30.1  & 30.1 &100.0 &100.0&n.a. &n.a.\\
&& $2$&55.7  &55.7 &100.0 &100.0 &100.0 &100.0\\
0.05 &$n=5,000$& $3$&75.7  &75.7 &100.0 &100.0& 100.0 &100.0\\
 &&$6$&87.1 & 87.1 &100.0 &100.0 &100.0 &100.0\\
 &&$12$&87.0  &86.8 &100.0 &100.0 &100.0& 100.0\\
 &&$15$&87.3  &87.2 &100.0 &100.0& 100.0 &100.0\\
  \cline{2-9}
&& $1$& 50.0&   50.0& 100.0& 100.0&n.a. &n.a.\\
&& $2$&79.5 & 79.4 &100.0 &100.0 &100.0 &100.0\\
0.05 &$n=10,000$& $3$&95.2 & 95.2 &100.0 &100.0 &100.0 &100.0\\
 &&$6$&98.0&  98.0 &100.0 &100.0 &100.0 &100.0\\
 &&$12$&98.6 & 98.6 &100.0 &100.0 &100.0 &100.0\\
 &&$15$&99.0&  99.0 &100.0 &100.0 &100.0 &100.0\\

\hline
%&& $1$&58.4  &58.3 & 82.7 & 82.5&n.a. &n.a.\\
%&& $2$&47.7& 47.6& 56.2 &56.1 &78.7 &78.6\\
%0.20 &$n=1,000$& $3$&40.1 &39.9& 53.9 &53.7 &75.6 &75.5\\
% &&$6$&32.2& 31.2& 49.0 &48.7 &70.3 &69.8\\
% &&$12$&23.0& 22.3 &43.8 &43.0 &63.0 &62.1\\
% &&$15$&21.0 &20.0 &40.4 &39.9 &58.0 &57.6\\
%
%  \cline{2-9}
&& $1$&98.2  &98.2&  99.9  &99.9&n.a. &n.a.\\
&& $2$&94.6  &94.6 & 99.5 & 99.5 &100.0 &100.0\\
0.20 &$n=5,000$& $3$&92.3 & 92.3 & 99.6  &99.6& 100.0 &100.0\\
 &&$6$&91.0 & 91.0 & 99.6  &99.6 &100.0 &100.0\\
 &&$12$&88.8 & 88.7 & 99.8 & 99.8& 100.0 &100.0\\
 &&$15$&88.6  &88.6 & 99.8  &99.8 &100.0 &100.0\\
  \cline{2-9}
&& $1$& 99.7 & 99.7& 100.0 &100.0&n.a. &n.a.\\
&& $2$&99.2  &99.2 &100.0 &100.0 &100.0 &100.0\\
0.20 &$n=10,000$& $3$&99.3  &99.2 &100.0 &100.0 &100.0 &100.0\\
 &&$6$&98.8  &98.8 &100.0 &100.0 &100.0 &100.0\\
 &&$12$&99.3  &99.3 &100.0 &100.0 &100.0 &100.0\\
 &&$15$&99.3  &99.3 &100.0 &100.0 &100.0 &100.0\\

\hline
%&& $1$&59.0  &59.0 & 82.4  &82.4 &n.a. &n.a.\\
%&& $2$&47.7& 47.7 &56.7 &56.4& 78.7 &78.6\\
%0.45 &$n=1,000$& $3$&41.0 &41.0& 54.6 &54.2 &76.4 &76.4\\
% &&$6$&33.1 &32.9& 49.9 &49.7& 70.5 &70.4\\
% &&$12$&24.1 &23.1& 44.3 &43.9 &63.5 &63.3
%\\
% &&$15$&21.9 &20.6 &41.0 &40.7 &58.2 &58.0\\
%  \cline{2-9}
&& $1$&98.2 & 98.2 & 99.8 & 99.8&n.a. &n.a.\\
&& $2$&94.4 & 94.3  &99.5  &99.5 &100.0 &100.0\\
0.45 &$n=5,000$& $3$&92.4 & 92.4  &99.6  &99.6 &100.0 &100.0\\
 &&$6$&90.9  &90.8 & 99.6 & 99.6 &100.0& 100.0\\
 &&$12$&88.9 & 88.9  &99.8 & 99.8 &100.0 &100.0\\
 &&$15$&88.8  &88.5 & 99.8  &99.8 &100.0& 100.0\\
  \cline{2-9}
&& $1$&99.7 & 99.7& 100.0 &100.0&n.a. &n.a.\\
&& $2$&99.0  &99.0 &100.0 &100.0 &100.0 &100.0\\
0.45 &$n=10,000$& $3$&99.2  &99.2 &100.0 &100.0 &100.0 &100.0\\
 &&$6$&98.9  &98.9 &100.0 &100.0 &100.0 &100.0\\
 &&$12$&99.3  &99.3 &100.0 &100.0 &100.0 &100.0\\
 &&$15$&99.3  &99.3 &100.0 &100.0 &100.0 &100.0\\
\hline\hline
\\

\end{tabular}
}
\end{center}
\label{pFARIMA00garch}
\end{table}

\begin{table}[h]
 \caption{\small{Empirical power (in \%) of the modified and standard versions of the LB and BP tests in the case of
 a weak FARIMA$(0,d_0,1)$  defined by \eqref{process-sim}--\eqref{noise-sim}
 with $\theta_0=(0.,0.2,d_0)$.
The nominal asymptotic level of the tests is $\alpha=5\%$.
The number of replications is $N=1,000$.}}
\begin{center}
%\begin{tabular}{lll rrr rrr}
{\scriptsize
\begin{tabular}{c c ccc ccc c}
\hline\hline \\
$d_0$& Length $n$ & Lag $m$ & $\mathrm{{LB}}_{\textsc{sn}}$&$\mathrm{BP}_{\textsc{sn}}$&$\mathrm{{LB}}_{\textsc{w}}$&$\mathrm{BP}_{\textsc{w}}$&$\mathrm{{LB}}_{\textsc{s}}$&$\mathrm{BP}_{\textsc{s}}$
\vspace*{0.2cm}\\\hline
%&& $1$&2.7 &  2.7  &91.7 & 91.7&n.a. &n.a.\\
%&& $2$&11.1 &11.0& 85.5 &85.4 &93.4& 93.3\\
%0.05 &$n=1,000$& $3$& 15.7 &15.6 &82.7& 82.7 &91.2& 91.2\\
% &&$6$& 15.8 &15.7 &77.1& 77.0 &87.1 &87.1\\
% &&$12$&6.5 & 5.8 &65.8 &65.7 &80.5 &80.4\\
% &&$15$&3.9 & 3.5 &62.3 &61.8 &78.6 &78.4\\
%  \cline{2-9}
&& $1$&20.3   &20.3  &99.9  &99.9&n.a. &n.a.\\
&& $2$&56.7& 56.6 &99.9& 99.9& 99.9 &99.9\\
0.05 &$n=5,000$& $3$&69.1 &69.1& 99.9& 99.9& 99.9& 99.9\\
 &&$6$&75.9 &75.9 &99.9 &99.9 &99.9 &99.9\\
 &&$12$&71.9& 71.4 &99.9 &99.9 &99.9 &99.9\\
 &&$15$&68.5& 68.0 &99.9 &99.9 &99.9 &99.9\\
  \cline{2-9}
&& $1$& 60.0&   60.0& 100.0& 100.0&n.a. &n.a.\\
&& $2$&81.8  &81.8 &100.0 &100.0 &100.0& 100.0\\
0.05 &$n=10,000$& $3$&90.3  &90.3 &100.0 &100.0 &100.0 &100.0\\
 &&$6$&93.9 & 93.9 &100.0 &100.0& 100.0 &100.0\\
 &&$12$&93.8  &93.8 &100.0 &100.0 &100.0 &100.0\\
 &&$15$&93.7  &93.7 &100.0 &100.0 &100.0 &100.0\\

\hline
%&& $1$&43.3  &43.2 & 64.5 & 64.5  &n.a. &n.a.\\
%&& $2$&31.3 &31.1& 53.2& 53.1 &75.1& 75.1\\
%0.20 &$n=1,000$& $3$&26.6 &26.5 &50.1 &49.8 &71.2 &71.2\\
% &&$6$&18.6 &18.4& 43.1 &42.8 &66.5 &66.4\\
% &&$12$&7.9  &7.7& 32.9 &32.3 &55.1& 54.4\\
% &&$15$&5.2  &4.6 &29.0 &28.4 &51.8 &51.3\\
%
%  \cline{2-9}
&& $1$&92.3  &92.3  &99.9 & 99.9&n.a. &n.a.\\
&& $2$&86.1 &86.0 &98.6& 98.6 &99.8& 99.8\\
0.20 &$n=5,000$& $3$&82.3& 82.3 &99.2& 99.1& 99.8 &99.8\\
 &&$6$&80.0 &80.0 &98.9 &98.9 &99.9 &99.9\\
 &&$12$&73.1& 72.8& 98.7 &98.7& 99.6 &99.6\\
 &&$15$&68.3 &68.0 &98.4& 98.4 &99.5& 99.5\\
  \cline{2-9}
&& $1$& 99.2 & 99.2 &100.0 &100.0&n.a. &n.a.\\
&& $2$&96.4  &96.4& 100.0 &100.0 &100.0 &100.0\\
0.20 &$n=10,000$& $3$&94.6 & 94.6 &100.0 &100.0 &100.0 &100.0\\
 &&$6$&95.1 & 95.1 &100.0& 100.0 &100.0 &100.0\\
 &&$12$&95.2 & 95.2 &100.0 &100.0 &100.0 &100.0\\
 &&$15$&94.0& 94.0 &100.0& 100.0 &100.0 &100.0\\

\hline
%&& $1$&43.6 & 43.2 & 64.6 & 64.6&n.a. &n.a.\\
%&& $2$&31.9& 31.8 &53.4 &53.1& 75.7 &75.6
%\\
%0.45 &$n=1,000$& $3$&27.1& 27.0& 50.2 &50.0 &71.8 &71.6\\
% &&$6$&18.8 &18.2 &43.4 &43.3 &67.1 &67.0\\
% &&$12$&8.7  &8.1& 33.7& 33.2 &55.6 &54.8
%\\
% &&$15$&5.5 & 4.7 &29.1 &28.5 &52.8 &52.0\\
%  \cline{2-9}
&& $1$&92.4 & 92.4  &99.9 & 99.9&n.a. &n.a.\\
&& $2$& 85.6& 85.6& 98.6 &98.6& 99.8 &99.8\\
0.45 &$n=5,000$& $3$&82.1 &82.0 &99.3 &99.3 &99.8& 99.8\\
 &&$6$&80.3 &80.3 &98.9 &98.9 &99.9& 99.9\\
 &&$12$&73.0 &72.7 &98.7 &98.7& 99.6 &99.6\\
 &&$15$&68.2 &68.1 &98.4 &98.4 &99.5 &99.5\\
  \cline{2-9}
&& $1$&99.2  &99.2 &100.0 &100.0&n.a. &n.a.\\
&& $2$&96.4  &96.4& 100.0& 100.0 &100.0& 100.0\\
0.45 &$n=10,000$& $3$&94.8  &94.8 &100.0 &100.0 &100.0 &100.0
\\
 &&$6$&95.2 & 95.2 &100.0 &100.0& 100.0 &100.0\\
 &&$12$&95.0 & 95.0 &100.0 &100.0& 100.0 &100.0\\
 &&$15$&94.0 & 94.0 &100.0 &100.0& 100.0 &100.0\\
\hline\hline
\\

\end{tabular}
}
\end{center}
\label{pFARIMA00pt}
\end{table}
%
%
%\clearpage
%
%\begin{figure}[h]
%\includegraphics[width=12cm,height=10cm]{IGARCH-d0_49n2000}
%\caption{Autocorrelation of  a realization of size $n=2,000$ for a weak  FARIMA$(0,0.49,0)$  model \eqref{process-sim}--\eqref{garch} with $\theta_0=(0,0,0.49)$ and where $\omega=0.04$, $\alpha_1=0.13$ and $\beta_1=0.88$.
%The horizontal dotted lines (blue color) correspond to the 5\% significant limits obtained under the strong FARIMA assumption.
%The solid lines (red color) and  dashed lines (green color) correspond also  to the 5\% significant limits under the weak FARIMA assumption.
%The full lines  correspond to the asymptotic significance limits for the residual autocorrelations
%obtained in Theorem~\ref{loi_res_rho}. The dashed lines (green color) correspond to the self-normalized asymptotic significance limits for the residual autocorrelations
%obtained in Theorem~\ref{sn2}.}
%\label{fig6}
%\end{figure}
%

%%%%%%%%%%%%%%%%%%%%%%%%%%%%%%%%%%%%%%%%%

\begin{table}[h]
 \caption{\small{Modified and standard versions
 of portmanteau tests to check the null hypothesis that the S\&P 500 squared returns follow a FARIMA$(1,0.2338,1)$ model \eqref{process-sim}. }}
\begin{center}
{\small
\begin{tabular}{c ccc cccc c}
\hline\hline
Lag $m$& &1&2&3 &4&5&6& 7\\
\cline{2-9}$\hat\rho(m)$&& 0.0002 &-0.0033 &-0.0350& -0.0393&  0.0893 &-0.0040& -0.0179
\\
$\mathrm{LB}_{\textsc{sn}}$&&0.0653  &18.150  &41.924 &58.057 &186.72&313.78&341.38
\\
$\mathrm{BP}_{\textsc{sn}}$&& 0.0653  &18.146 &41.912  &58.037 &186.64&313.64&341.20
\\
\cline{3-9}
$\mathrm{LB}_{\textsc{w}}$&&0.0008 &  0.1885&  21.445&  48.248& 186.95& 187.23& 192.77
\\
$\mathrm{BP}_{\textsc{w}}$&&0.0008 &  0.1884 & 21.439&  48.232& 186.88& 187.15& 192.67
\\
\cline{3-9}  $\mbox{p}_{\textsc{w}}^{\textsc{lb}}$&&0.8525 &0.6985 &0.0916 &0.3137 &0.0678 &0.0717 &0.0752
\\
$\mbox{p}_{\textsc{w}}^{\textsc{bp}}$&&0.8525 &0.6986 &0.0917 &0.3138& 0.0679& 0.0718& 0.0753
\\
\cline{3-9}  $\mbox{p}_{\textsc{s}}^{\textsc{lb}}$&&n.a.&n.a.&n.a.&0.0000&0.0000&0.0000&0.0000
\\
$\mbox{p}_{\textsc{s}}^{\textsc{bp}}$&&n.a.&n.a.&n.a.&0.0000&0.0000&0.0000&0.0000
\\

\hline
Lag $m$& &8&9&10 &11&12&13& 14\\
\cline{2-9}$\hat\rho(m)$&&0.0047 & 0.0137& -0.0040  &0.0295&0.0093 &-0.0077& -0.0286
\\
$\mathrm{LB}_{\textsc{sn}}$&&397.27& 397.38 &415.22&465.52& 468.76&567.87&573.02
\\
$\mathrm{BP}_{\textsc{sn}}$&&397.04&397.13 &414.93&465.17& 468.33& 567.38 &572.49
\\
\cline{3-9}
$\mathrm{LB}_{\textsc{w}}$&&193.16& 196.42& 196.69&211.82&213.31& 214.34& 228.55
\\
$\mathrm{BP}_{\textsc{w}}$&& 193.09& 196.34 &196.61&211.74 &213.22& 214.25& 228.45
\\
\cline{3-9}  $\mbox{p}_{\textsc{w}}^{\textsc{lb}}$&&0.0758 &0.0786 &0.0986 &0.1053 &0.1148 &0.1226&0.1047
\\
$\mbox{p}_{\textsc{w}}^{\textsc{bp}}$&& 0.0758&0.0787 &0.0987 &0.1054 &0.1150 &0.1228&0.1048
\\
\cline{3-9}  $\mbox{p}_{\textsc{s}}^{\textsc{lb}}$&&0.0000&0.0000&0.0000&0.0000&0.0000&0.0000&0.0000
\\
$\mbox{p}_{\textsc{s}}^{\textsc{bp}}$&&0.0000&0.0000&0.0000&0.0000&0.0000&0.0000&0.0000
\\

\hline
Lag $m$& &15&16&17 &18&19&20&21 \\
\cline{2-9}$\hat\rho(m)$&&0.0021 & 0.0086 & 0.0097 & 0.0137 &-0.0023 & 0.0016&0.0132
\\
$\mathrm{LB}_{\textsc{sn}}$&&588.61&701.16 &738.23 &738.58&749.24&778.88&788.01
\\
$\mathrm{BP}_{\textsc{sn}}$&&588.04&700.44&737.42&737.73&748.33&777.90&786.97
\\
\cline{3-9}
$\mathrm{LB}_{\textsc{w}}$&&228.63& 229.91& 231.54& 234.83& 234.92& 234.97&238.00
\\
$\mathrm{BP}_{\textsc{w}}$&&  228.52& 229.80& 231.44& 234.72& 234.81& 234.86&237.89
\\
\cline{3-9}  $\mbox{p}_{\textsc{w}}^{\textsc{lb}}$&&0.1079 &0.1113& 0.2212 &0.2138 &0.2127 &0.2169&0.2324
\\
$\mbox{p}_{\textsc{w}}^{\textsc{bp}}$&& 0.1080 &0.1114 &0.2214 &0.2140 &0.2130 &0.2171&0.2327
\\
\cline{3-9}  $\mbox{p}_{\textsc{s}}^{\textsc{lb}}$&&0.0000&0.0000&0.0000&0.0000&0.0000&0.0000&0.0000
\\
$\mbox{p}_{\textsc{s}}^{\textsc{bp}}$&&0.0000&0.0000&0.0000&0.0000&0.0000&0.0000&0.0000
\\
\hline\hline
\end{tabular}
}
\end{center}
\label{sp500FARIMA11}
\end{table}
\begin{figure}[h]
\includegraphics[width=12cm,height=10cm]{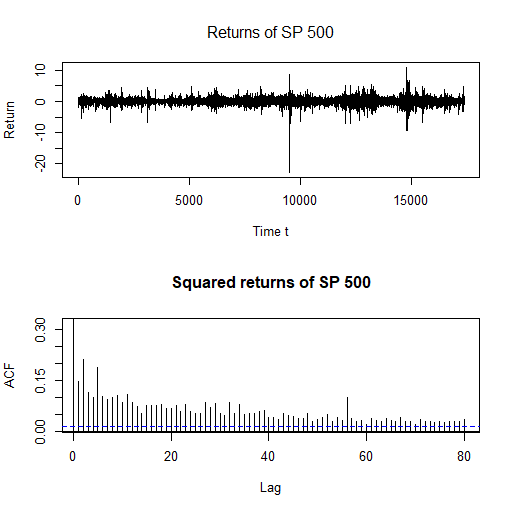}
\caption{Returns and the sample autocorrelations of squared returns of the S\&P 500.  }
\label{graph}
\end{figure}
%\newpage
\begin{figure}[h]
\includegraphics[width=12cm,height=10cm]{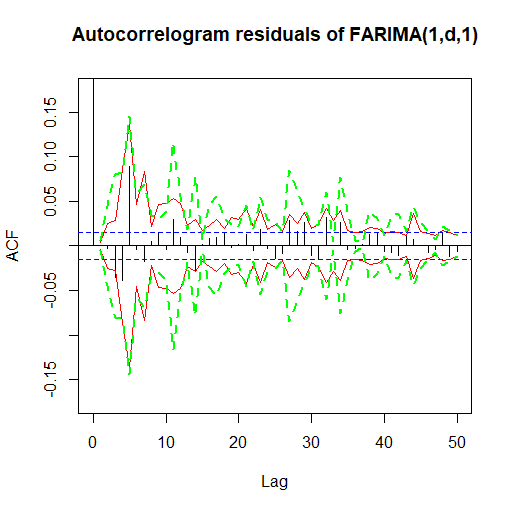}
\caption{Autocorrelation of the FARIMA$(1,0.2338,1)$ residuals for the squares of the S\&P 500 returns.
The horizontal dotted lines (blue color) correspond to the 5\% significant limits obtained under the strong FARIMA assumption.
The solid lines (red color) and  dashed lines (green color) correspond also  to the 5\% significant limits under the weak FARIMA assumption.
The full lines  correspond to the asymptotic significance limits for the residual autocorrelations
obtained in Theorem~\ref{loi_res_rho}. The dashed lines (green color) correspond to the self-normalized asymptotic significance limits for the residual autocorrelations
obtained in Theorem~\ref{sn2}.}
\label{acf}
\end{figure}
\begin{table}[h]
 \caption{\small{Modified and standard versions
 of portmanteau tests to check the null hypothesis that the Nikkei squared returns follow a FARIMA$(0,0.2132,0)$ model as in \eqref{process-sim} with $a=b=0$. }}
\begin{center}
{\small
\begin{tabular}{c ccc cccc c}
\hline\hline
Lag $m$& &1&2&3 &4&5&6& 7\\
\cline{2-9}$\hat\rho(m)$&& -0.0678  &0.0400 & 0.0634& -0.0022 & 0.0165 & 0.0320& -0.0158
\\
$\mathrm{LB}_{\textsc{sn}}$&&5.7332  &29.005  &34.758  &34.779 & 66.692 &288.57 &324.46
\\
$\mathrm{BP}_{\textsc{sn}}$&& 5.7319 & 28.997  &34.745 & 34.764  &66.657& 288.40& 324.24
\\
\cline{3-9}
$\mathrm{LB}_{\textsc{w}}$&&61.211 & 82.507 &136.13 &136.20 &139.84 &153.46& 156.78
\\
$\mathrm{BP}_{\textsc{w}}$&&61.198 & 82.487 &136.09 &136.16& 139.76& 153.41& 156.73
\\
\cline{3-9}  $\mbox{p}_{\textsc{w}}^{\textsc{lb}}$&&0.1086 &0.2186 &0.1830& 0.2551 &0.3002& 0.3519& 0.3609
\\
$\mbox{p}_{\textsc{w}}^{\textsc{bp}}$&&0.1086& 0.2187 &0.1831 &0.2552 &0.3003 &0.3521 &0.3611
\\
\cline{3-9}  $\mbox{p}_{\textsc{s}}^{\textsc{lb}}$&&n.a.&0.0000&0.0000&0.0000&0.0000&0.0000&0.0000
\\
$\mbox{p}_{\textsc{s}}^{\textsc{bp}}$&&n.a.&0.0000&0.0000&0.0000&0.0000&0.0000&0.0000
\\

\hline
Lag $m$& &8&9&10 &11&12&13& 14\\
\cline{2-9}$\hat\rho(m)$&&0.0295 & 0.0384& 0.0121 & 0.0133  &0.0503&  0.0076 & 0.0068
\\
$\mathrm{LB}_{\textsc{sn}}$&&387.88& 512.70& 575.09 &600.81& 791.67 &808.20& 808.27
\\
$\mathrm{BP}_{\textsc{sn}}$&&387.59 &512.28& 574.57& 600.22 &790.83 &807.29 &807.30
\\
\cline{3-9}
$\mathrm{LB}_{\textsc{w}}$&&168.41&188.08 &190.01 &192.36& 226.12 &226.89& 227.50
\\
$\mathrm{BP}_{\textsc{w}}$&& 168.35&187.10& 189.93& 192.29 &225.10 &226.76& 227.39
\\
\cline{3-9}  $\mbox{p}_{\textsc{w}}^{\textsc{lb}}$&&0.3627& 0.3757& 0.3802 &0.3825&0.3320& 0.3447 &0.3526
\\
$\mbox{p}_{\textsc{w}}^{\textsc{bp}}$&& 0.3629 &0.3759 &0.3804 &0.3827&0.3323 &0.3450 &0.3529
\\
\cline{3-9}  $\mbox{p}_{\textsc{s}}^{\textsc{lb}}$&&0.0000&0.0000&0.0000&0.0000&0.0000&0.0000&0.0000
\\
$\mbox{p}_{\textsc{s}}^{\textsc{bp}}$&&0.0000&0.0000&0.0000&0.0000&0.0000&0.0000&0.0000
\\

\hline
Lag $m$& &15&16&17 &18&19&20&21 \\
\cline{2-9}$\hat\rho(m)$&&0.0538  &0.0073  &0.0173 & 0.0067&-0.0027 &-0.0057  &0.0153
\\
$\mathrm{LB}_{\textsc{sn}}$&&839.87 &842.24& 842.31 &845.36 &885.74& 935.70 &946.03
\\
$\mathrm{BP}_{\textsc{sn}}$&&838.80 &841.10 &841.11 &844.10 &884.35 &934.15 &944.40
\\
\cline{3-9}
$\mathrm{LB}_{\textsc{w}}$&&266.16 &266.88&270.85 &271.45 &271.56 &271.99 &275.13
\\
$\mathrm{BP}_{\textsc{w}}$&&  265.99& 266.71&270.68 &271.28 &271.38 &271.82& 274.94
\\
\cline{3-9}  $\mbox{p}_{\textsc{w}}^{\textsc{lb}}$&&0.3105& 0.3163& 0.3161 &0.3264 &0.3289& 0.3329 &0.3366
\\
$\mbox{p}_{\textsc{w}}^{\textsc{bp}}$&& 0.3108 &0.3166& 0.3165& 0.3268& 0.3293& 0.3333 &0.3369
\\
\cline{3-9}  $\mbox{p}_{\textsc{s}}^{\textsc{lb}}$&&0.0000&0.0000&0.0000&0.0000&0.0000&0.0000&0.0000
\\
$\mbox{p}_{\textsc{s}}^{\textsc{bp}}$&&0.0000&0.0000&0.0000&0.0000&0.0000&0.0000&0.0000
\\
\hline\hline
\end{tabular}
}
\end{center}
\label{nikkeiFARIMA11}
\end{table}

\begin{figure}[h]
\includegraphics[width=12cm,height=10cm]{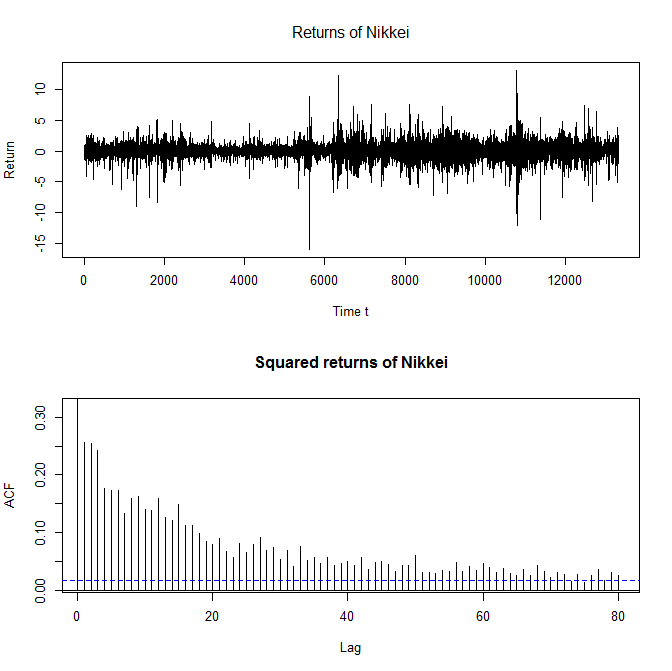}
\caption{Returns and the sample autocorrelations of squared returns of the Nikkei.  }
\label{nikkei}
\end{figure}
%\newpage
\begin{figure}[h]
\includegraphics[width=12cm,height=10cm]{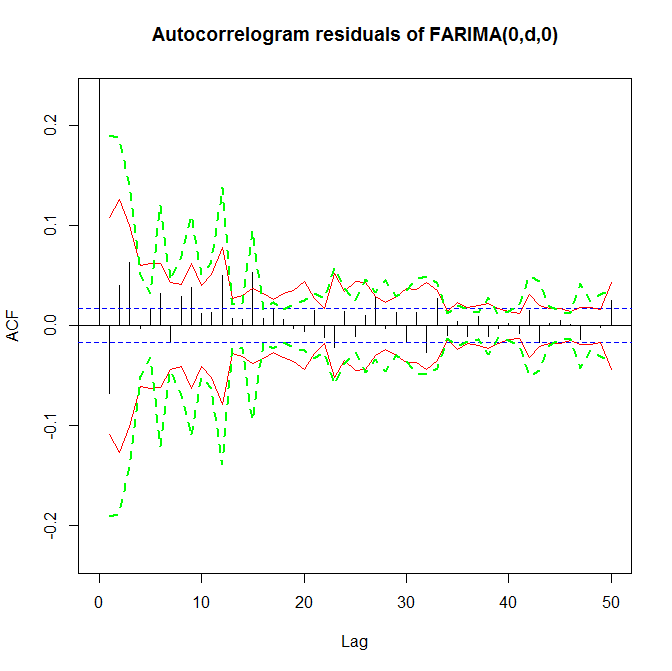}
\caption{Autocorrelation of the FARIMA$(0,0.2132,0)$ residuals for the squares of the Nikkei returns.
The horizontal dotted lines (blue color) correspond to the 5\% significant limits obtained under the strong FARIMA assumption.
The solid lines (red color) and  dashed lines (green color) correspond also  to the 5\% significant limits under the weak FARIMA assumption.
The full lines  correspond to the asymptotic significance limits for the residual autocorrelations
obtained in Theorem~\ref{loi_res_rho}. The dashed lines (green color) correspond to the self-normalized asymptotic significance limits for the residual autocorrelations
obtained in Theorem~\ref{sn2}.}
\label{acfnikkei}
\end{figure}
%
% Trucs pour utiliser .eps
%\vspace*{12cm} \protect \special{psfile=sp500.eps
%voffset=-200 hoffset=-45 hscale=80 vscale=80} \vspace*{2.5cm}

\clearpage
\newpage

%\bibliographystyle{acmtrans-ims}
%\bibliographystyle{plain}
%\addcontentsline{toc}{section}{Bibliography}
\bibliographystyle{apalike}
% or: plain,unsrt,alpha,abbrv,acm,apalike,.
%\bibliographystyle{alpha}
%\bibliographystyle{acm}
%\bibliographystyle{smfalpha2}
\bibliography{biblio-youss}

\newpage
\clearpage
\appendix
\section{Supplemental material: Proofs}\label{appendix}
The following proofs are quite technical and are adaptations of the arguments used in \cite{fz98}, \cite{frz} and \cite{BMS2018}.

The results of \cite{BMES2019} which will be needed for all the proofs  are collected in the following Subsection~\ref{prelim} in order to have a self-containing paper.

In all our proofs, $K$ is a positive constant that may vary from line to line.
\subsection{Preliminary results}\label{prelim}
In this subsection, we shall give some results on estimations of the coefficients of formal power series that will arise in our study.

We begin by recalling the following properties on power series. If for $|z|\le R$, the power series $f(z)=\sum_{i\ge 0}a_iz^i$ and  $g(z)=\sum_{i\ge 0}b_iz^i$ are well defined, then one has $(fg)(z)= \sum_{i\ge 0} c_iz^i$ is also well defined for $|z|\le R$ with the sequence $(c_i)_{i\ge 0}$ which is given by $c=a\ast b$ where $\ast$ denotes the convolution product between $a$ and $b$ defined by $c_i=\sum_{k=0}^i a_kb_{i-k}=\sum_{k=0}^i a_{i-k}b_{k}$. We will make use of the Young inequality that states that if the sequence $a\in \ell^{r_1}$ and $b\in\ell^{r_2}$ and such that $\frac{1}{r_1}+\frac{1}{r_2}=1+\frac{1}r$ with $1\le r_1,r_2,r \le \infty$, then $$
\left \| a\ast b \right \|_{\ell^r} \le  \left \| a \right \|_{\ell^{r_1}} \times \left \|  b \right \|_{\ell^{r_2}} .$$

Now we come back to the power series that arise in our context. Remind that for the true value of the parameter,
\begin{equation}\label{FF}
a_{\theta_0}(L)(1-L)^{d_0}X_t=b_{\theta_0}(L)\epsilon_t.
\end{equation}
Thanks to the assumptions on the moving average polynomials $b_\theta$ and the autoregressive polynomials $a_\theta$, the power series $a_\theta^{-1}$ and $b_\theta^{-1}$ are well defined.

Thus the functions $\epsilon_t(\theta)$ defined in \eqref{FARIMA-th}  can be written as
\begin{align}\label{epsi}
\epsilon_t(\theta) &= b^{-1}_{\theta}(L) a_{\theta}(L)(1-L)^{d}X_t \\
\label{epsi-bis}
& =b^{-1}_{\theta}(L) a_{\theta}(L)(1-L)^{d-d_0}a^{-1}_{\theta_0}(L) b_{\theta_0}(L)\epsilon_t
\end{align}
and if we denote $\gamma(\theta)=(\gamma_i(\theta))_{i\ge 0}$ the sequence of coefficients of the power series $b^{-1}_{\theta}(z) a_{\theta}(z)(1-z)^{d}$,
% \textcolor{red}{(which is absolutely convergent for at least for $|z|\le 1$) (à enlever)}
 we may write for all $t\in\mathbb Z$:
\begin{align}\label{AR-Inf}
\epsilon_t(\theta)&=\sum_{i\geq 0}\gamma_i(\theta)X_{t-i}.
\end{align}
In the same way, by \eqref{epsi} one has
\begin{align*}
X_t & = (1-L)^{-d}a^{-1}_{\theta}(L) b_{\theta}(L)\epsilon_t(\theta)
\end{align*}
and if we denote $\eta(\theta)=(\eta_i(\theta))_{i\ge 0}$ the coefficients of the power series $(1-z)^{-d}a^{-1}_{\theta}(z) b_{\theta}(z)$ one has \begin{align}
\label{MA-Inf}
X_t & = \sum_{i\geq 0}\eta_i(\theta)\epsilon_{t-i}(\theta) \ .
\end{align}
We strength the fact that $\gamma_0(\theta)=\eta_0(\theta)=1$ for all $\theta$. 

For large $j$, \cite{hallin} have shown that uniformly in $\theta$ the sequences  $\gamma(\theta)$ and $\eta(\theta)$
satisfy
\begin{equation}\label{Coef-Gamma}
\frac{\partial^k\gamma_j(\theta)}{\partial\theta_{i_1}\cdots\partial\theta_{i_k}}=\mathrm{O}\left( j^{-1-d}\left\lbrace \log(j)\right\rbrace ^k\right),\text{ for }k=0,1,2,3,
\end{equation}
and
\begin{equation}\label{Coef-Eta}
\frac{\partial^k\eta_j(\theta)}{\partial\theta_{i_1}\cdots\partial\theta_{i_k}}=\mathrm{O}\left( j^{-1+d}\left\lbrace \log(j)\right\rbrace ^k\right), \text{ for }k=0,1,2,3.
\end{equation}
%for $k=0,1,2,3$. % (see \cite{hallin}, for more details).

One difficulty that has to be addressed is that \eqref{AR-Inf} includes the infinite past $(X_{t-i})_{i\ge 0}$ whereas only a finite number of observations $(X_t)_{1\leq t\leq n}$ are available to compute the estimators defined in \eqref{theta_chap}.
The simplest solution is truncation which amounts to setting all unobserved values equal to zero. Thus, %for $t=2,\dots,n$ and
for all $\theta\in\Theta$ and $1\le t\le n$ one defines
%Using \eqref{AR-Inf} and the last system, we deduce that, for all $\theta\in\Theta$,
\begin{equation}\label{epsilon-tilde--}
\tilde{\epsilon}_t(\theta)=\sum_{i=0}^{t-1}\gamma_i(\theta)X_{t-i}= \sum_{i\ge 0} \gamma_i^t(\theta) X_{t-i}
\end{equation}
where the truncated sequence $\gamma^t(\theta)= ( \gamma_i^t(\theta))_{i\ge 0}$ is defined by
$$\gamma_i^t(\theta)=\left\{
\begin{array}{rl}
 \gamma_i(\theta) &\text{ if } \ 0\leq i\leq t-1\ , \\
0& \text{ otherwise.}
\end{array}\right.$$
Since our assumptions are made on the noise in $(\ref{FARIMA})$, it will be useful to express the random variables $\epsilon_t(\theta)$ and its partial derivatives with respect to $\theta$, as a function of $(\epsilon_{t-i})_{i\ge 0}$.

From \eqref{epsi-bis}, there exists a sequence $\lambda(\theta)=(\lambda_i(\theta))_{i\geq0}$ such that
\begin{equation}\label{epsil-th}
\epsilon_t(\theta)=\sum_{i=0}^{\infty}\lambda_i\left( \theta\right) \epsilon_{t-i}
\end{equation}
where the sequence $\lambda(\theta)$ is given by the sequence of the coefficients of the power series $b^{-1}_{\theta}(z) a_{\theta}(z)(1-z)^{d-d_0}a^{-1}_{\theta_0}(z) b_{\theta_0}(z)$. Consequently $\lambda(\theta) = \gamma(\theta)\ast \eta(\theta_0)$ or, equivalently,
\begin{align}\label{Coef-lambda}
\lambda_i( \theta)& =\sum_{j=0}^i\gamma_j(\theta)\eta_{i-j}(\theta_0).
\end{align}
As in \cite{hualde2011}, it can be shown using Stirling's approximation that there exists a positive constant $K$ such that 
\begin{equation}\label{eqasym-lambda}
\sup_{\theta\in\Theta_{\delta}}\left| \lambda_i(\theta)\right| \leq K\sup_{d\in[d_1,d_2]} i^{-1-(d-d_0)}\leq K i^{-1-(d_1-d_0)} \ .
\end{equation}
\noindent Equation \eqref{epsil-th} and Inequality \eqref{eqasym-lambda} imply that for all $\theta\in\Theta$ the random variable $\epsilon_t(\theta)$ belongs to $\mathbb{L}^2$, that %\left(\Omega, \mathcal{F},\mathbb{P} \right)$,%
the sequence $(\epsilon_t(\theta))_t$ is an ergodic sequence and that for all $t\in\mathbb{Z}$ the function $\epsilon_t(\cdot)$ is a continuous function.
We proceed in the same way as regard to the derivatives of $\epsilon_t(\theta)$. More precisely, for any $\theta\in\Theta$, $t\in\mathbb{Z}$ and $1\le k,l \le p+q+1$ there exists sequences $\overset{\textbf{.}}{\lambda}_{k}(\theta)= (\overset{\textbf{.}}{\lambda}_{i,k}(\theta))_{i\geq1}$
and $\overset{\textbf{..}}{\lambda}_{k,l}(\theta)= (\overset{\textbf{..}}{\lambda}_{i,k,l}(\theta))_{i\geq1}$ such that
\begin{align}
\frac{\partial\epsilon_t(\theta)}{\partial\theta_k}&=\sum_{i=1}^{\infty}\overset{\textbf{.}}{\lambda}_{i,k}\left( \theta\right) \epsilon_{t-i} 
 \label{deriveesecepsil}\\
\frac{\partial^2\epsilon_t(\theta)}{\partial\theta_k\partial\theta_{l}}& =\sum_{i=1}^{\infty}\overset{\textbf{..}}{\lambda}_{i,k,l}\left( \theta\right) \epsilon_{t-i} .\label{deriveesecepsil1}
\end{align}
Of course it holds that $\overset{\textbf{.}}{\lambda}_{k}(\theta)=\frac{\partial\gamma(\theta)}{\partial\theta_k}\ast\eta(\theta_0)$ and
$\overset{\textbf{..}}{\lambda}_{k,l}( \theta)=\frac{\partial^2\gamma(\theta)}{\partial\theta_k\partial\theta_{l}}\ast \eta(\theta_0)$.

Similarly we have
\begin{align}\label{epsiltilde-th}
\tilde{\epsilon}_t(\theta)& =\sum_{i=0}^{\infty}\lambda_i^t\left( \theta\right) \epsilon_{t-i}, \\
\frac{\partial\tilde{\epsilon}_t(\theta)}{\partial\theta_k}& =\sum_{i=1}^{\infty}\overset{\textbf{.}}{\lambda}_{i,k}^t\left( \theta\right) \epsilon_{t-i},
\label{deriveesecepsiltilde} \\
\frac{\partial^2\tilde{\epsilon}_t(\theta)}{\partial\theta_k\partial\theta_{l}}& =\sum_{i=1}^{\infty}\overset{\textbf{..}}{\lambda}_{i,k,l}^t\left( \theta\right) \epsilon_{t-i},
\end{align}
where $\lambda^t(\theta) = \gamma^t(\theta)\ast \eta(\theta_0)$, $\overset{\textbf{.}}{\lambda}^t_{k}(\theta)=\frac{\partial\gamma^t(\theta)}{\partial\theta_k}\ast\eta(\theta_0)$ and
$\overset{\textbf{..}}{\lambda}^t_{k,l}( \theta)=\frac{\partial^2\gamma^t(\theta)}{\partial\theta_k\partial\theta_{l}}\ast \eta(\theta_0)$.

In order to handle the truncation error $\epsilon_t(\theta)-\tilde\epsilon_t(\theta)$, one needs some information on the sequence $\lambda(\theta)-\lambda^t(\theta)$. In \cite{BMES2019} the following two lemmas are proved.
\begin{lemme}\label{lemme_sur_les_ecarts_des_coef}
For  $2\leq r\leq \infty$ and $1\le k,l \le p+q+1 $, we have
\begin{equation*}
\parallel\lambda\left(\theta\right)-\lambda^t\left(\theta\right)\parallel_{\ell^r} \ =\mathrm{O}\left(t^{-1+\frac{1}{r}-(d-max(d_0,0))}\right),
\end{equation*}
\begin{equation*}
\parallel\overset{\textbf{.}}{\lambda}_k\left(\theta\right)-\overset{\textbf{.}}{\lambda}_k^t\left(\theta\right)\parallel_{\ell^r} \ =\mathrm{O}\left(t^{-1+\frac{1}{r}-(d-max(d_0,0))}\right)
\end{equation*}
and 
\begin{equation*}
\parallel\overset{\textbf{..}}{\lambda}_{k,l}\left(\theta\right)-\overset{\textbf{..}}{\lambda}_{k,l}^t\left(\theta\right)\parallel_{\ell^r} \ =\mathrm{O}\left(t^{-1+\frac{1}{r}-(d-max(d_0,0))}\right)
\end{equation*}
for any $\theta\in\Theta_{\delta}$ if $d_0\leq 0$ and for $\theta$ with non-negative memory parameter $d$ if $d_0>0$.
%if $d_0\leq 0$ or if $d_0>0$ and $d\geq 0$. 

%Moreover, for any $\theta\in\Theta$, one has
%\begin{equation*}
%\parallel\lambda\left(\theta\right)-\lambda^t\left(\theta\right)\parallel_{\ell^\infty} \ =\mathrm{O}\left(t^{-1-(d-d_0)}\right)
%\end{equation*}
%and
%\begin{equation*}
%\parallel\overset{\textbf{.}}{\lambda}_k\left(\theta\right)-\overset{\textbf{.}}{\lambda}_k^t\left(\theta\right)\parallel_{\ell^\infty} \ =\mathrm{O}\left(t^{-1-(d-d_0)}\right) .
%\end{equation*}
%}
\end{lemme}
%\begin{lemme}\label{lemme_sur_les_ecarts_des_coef}
%For  $2\le r\le \infty$, $1\le k \le p+q+1 $ and $\theta\in\Theta$, we have
%\begin{equation*}
%\parallel\lambda\left(\theta\right)-\lambda^t\left(\theta\right)\parallel_{\ell^r} \ =\mathrm{O}\left(t^{-1+\frac{1}{r}-(d-d_0)}\right)
%\end{equation*}
%and
%\begin{equation*}
%\parallel\overset{\textbf{.}}{\lambda}_k\left(\theta\right)-\overset{\textbf{.}}{\lambda}_k^t\left(\theta\right)\parallel_{\ell^r} \ =\mathrm{O}\left(t^{-1+\frac{1}{r}-(d-d_0)}\right) .
%\end{equation*}
%\end{lemme}
%\begin{rmq}\label{impo}
%Taking $r=\infty$ in the above lemma implies that the sequence $ \overset{\textbf{.}}{\lambda}_k\left(\theta_0\right)-\overset{\textbf{.}}{\lambda^t}_k\left(\theta_0\right)$ is bounded and more precicely there exists $K$ such that
%\begin{align}\label{impoeq}
%\sup_{j\ge 1} \left | \overset{\textbf{.}}{\lambda}_{j,k}\left(\theta_0\right)-\overset{\textbf{.}}{\lambda^t}_{j,k}\left(\theta_0\right) \right | & \le \frac{K}t
%\end{align} for any $t$ and any $1\le k\le p+q+1$.
%\end{rmq}
%One shall also need the following lemmas.
\begin{rmq}\label{impo}
The above lemma implies that the sequence $ \overset{\textbf{.}}{\lambda}_k\left(\theta_0\right)-\overset{\textbf{.}}{\lambda^t}_k\left(\theta_0\right)$ is bounded and more precisely there exists $K$ such that
\begin{align}\label{impoeq}
\sup_{j\ge 1} \left | \overset{\textbf{.}}{\lambda}_{j,k}\left(\theta_0\right)-\overset{\textbf{.}}{\lambda^t}_{j,k}\left(\theta_0\right) \right | & \le \frac{K}{t^{1+\min(d_0,0)}}
\end{align} for any $t\geq 1$ and any $1\le k\le p+q+1$.

\end{rmq}
\begin{rmq}\label{rmq:important}
In order to prove our asymptotic results, it will be convenient to give an upper bound for the norms of the sequences introduced in Lemma \ref{lemme_sur_les_ecarts_des_coef} valid for any $\theta\in\Theta_{\delta}$. Since $d_1-d_0>-1/2$, Estimation \eqref{eqasym-lambda} entails that for any $r\geq 2$,
\begin{align*}
\parallel\lambda\left(\theta\right)-\lambda^t\left(\theta\right)\parallel_{\ell^r} \ =\mathrm{O}\left(t^{-1+\frac{1}{r}-(d_1-d_0)}\right), \ \ \ \forall\theta\in\Theta_{\delta}.
\end{align*}
This can easily be seen since $\parallel\lambda(\theta)-\lambda^t(\theta)\parallel_{\ell^r}\leq K(\sum_{i\geq t}i^{-r-r(d_1-d_0)})^{1/r}\leq Kt^{-1+1/r-(d_1-d_0)}$. As in \cite{hallin}, the coefficients $\overset{\textbf{.}}{\lambda}_{j,k}(\theta)$ and $\overset{\textbf{..}}{\lambda}_{j,k,l}(\theta)$ are $\mathrm{O}(j^{-1-(d-d_0)+\zeta})$ for any small enough $\zeta>0$, so we have 
\begin{equation*}
\parallel\overset{\textbf{.}}{\lambda}_k\left(\theta\right)-\overset{\textbf{.}}{\lambda}_k^t\left(\theta\right)\parallel_{\ell^r} \ =\mathrm{O}\left(t^{-1+\frac{1}{r}-(d_1-d_0)+\zeta}\right)
\end{equation*}
and 
\begin{equation*}
\parallel\overset{\textbf{..}}{\lambda}_{k,l}\left(\theta\right)-\overset{\textbf{..}}{\lambda}_{k,l}^t\left(\theta\right)\parallel_{\ell^r} \ =\mathrm{O}\left(t^{-1+\frac{1}{r}-(d_1-d_0)+\zeta}\right)
\end{equation*}
for any $r\geq 2$, any $1\leq k,l\leq p+q+1$ and all $\theta\in\Theta_{\delta}$.
\end{rmq}
\begin{lemme}\label{miss}
For  any $2\le r\le \infty$, $1\le k \le p+q+1 $ and $\theta\in\Theta$, there exists a constant $K$ such that we have
\begin{align*}
\parallel\lambda_k^t\left(\theta\right)\parallel_{\ell^r} &\le K
%\qquad\text{and}
\\
\text{and}\quad\parallel\overset{\textbf{.}}{\lambda}_k^t\left(\theta\right)\parallel_{\ell^r} &\le K.
\end{align*}
\end{lemme}
%\begin{lemme}\label{miss2}
%There exists a constant $K$ such that we have
%\begin{equation}\label{eq-miss2}
% \left | \overset{\textbf{.}}{\lambda}_{i,k}\left(\theta_0\right) \right | \le  \frac{K}{i}.
%\end{equation}
%\end{lemme}
\subsection{Proof of Proposition~\ref{loijointe}}\label{proof_loijointe}
First we remark that the asymptotic normality of the joint distribution of  $\sqrt{n}( \hat{\theta}_n'-\theta_0^{'},\gamma_m^{'})^{'}$
can be established along the same lines as the proof of Theorem 2 in \cite{BMES2019}. The detailed proof is omitted.
From \eqref{gamma_m} and \eqref{ecart-theta} we have
\begin{align*}
\sqrt{n}\begin{pmatrix}\hat{\theta}_n-\theta_0\vspace{0.2cm}\\
\gamma_m
\end{pmatrix}
%&=\begin{pmatrix}\frac{-2}{\sqrt{n}}J^{-1}(\theta_0)\sum_{t=1}^n\left\lbrace \epsilon_t\frac{\partial}{\partial\theta}\epsilon_t(\theta_0)\right\rbrace +\mathrm{o}_{\mathbb{P}}(1)\vspace{0.2cm}\\
%\frac{1}{\sqrt{n}}\sum_{t=1}^n(\epsilon_{t-1},\dots,\epsilon_{t-m})^{'}\epsilon_t
%\end{pmatrix}\\
&=\frac{1}{\sqrt{n}}\sum_{t=1}^n\begin{pmatrix}-2J^{-1}(\theta_0)\epsilon_t\frac{\partial}{\partial\theta}\epsilon_t(\theta_0)\vspace{0.2cm}\\
(\epsilon_{t-1},\dots,\epsilon_{t-m})^{'}\epsilon_t
\end{pmatrix}+\begin{pmatrix}\mathrm{o}_{\mathbb{P}}(1)\\
\mathbf{0}_{m}
\end{pmatrix}\\
&=\frac{1}{\sqrt{n}}\sum_{t=1}^n U_t+\mathrm{o}_{\mathbb{P}}(1),
\end{align*}
where $\mathbf{0}_m$ is the vector of $\mathbb{R}^{m\times 1}$ with zero components.
It is clear that $U_t$ is a measurable function of $\epsilon_{t},\epsilon_{t-1},\dots$
Thus by using the same arguments as in \cite{BMES2019} (see proof of Theorem 2),
the central limit theorem (CLT) for strongly mixing processes $(U_t)_{t\in\mathbb{Z}}$ of  \cite{herr}  implies that  $({1}/{\sqrt{n}})\sum_{t=1}^{n}U_t$
has a limiting normal distribution with mean 0 and covariance matrix $\Xi$.

%The first block of the matrix $\Xi$ is defined, using the result of Theorem \ref{n.asymptotique}, by:
For $i\ge 1$, we denote $ {\Lambda}_{i}(\theta_0)=(\overset{\textbf{.}}{\lambda}_{i,1}(\theta_0),\dots,
\overset{\textbf{.}}{\lambda}_{i,p+q+1}(\theta_0))'$. From \eqref{deriveesecepsil} we deduce that
\begin{align}\label{scoreVec}
\frac{\partial\epsilon_t(\theta_0)}{\partial\theta}&=\sum_{i=1}^{\infty}{\Lambda}_{i}(\theta_0) \epsilon_{t-i}.
\end{align}
In view of \eqref{ecart-theta} and \eqref{scoreVec}, by applying the CLT for mixing processes we directly obtain
\begin{align*}
\Sigma_{\hat{\theta}}&=\lim_{n\rightarrow\infty}\mathrm{Var}\left(2J^{-1}\frac{1}{\sqrt{n}}\sum_{t=1}^n\epsilon_t\frac{\partial}{\partial\theta}\epsilon_t(\theta_0)\right):=J^{-1}IJ^{-1}
\\&=4J^{-1}\sum_{\ell,\ell'=1}^\infty{\Lambda}_{\ell}\left(\theta_0\right){\Lambda'}_{\ell'}\left(\theta_0\right)\sum_{h=-\infty}^\infty \mathbb{E}\left(\epsilon_t\epsilon_{t-\ell}\epsilon_{t-h}\epsilon_{t-\ell'-h}\right)J^{-1}
\\&=4J^{-1}\sum_{\ell,\ell'=1}^\infty{\Lambda}_{\ell}\left(\theta_0\right){\Lambda'}_{\ell'}\left(\theta_0\right)\Gamma(\ell,\ell')J^{-1},
\end{align*}
which gives the first block of the asymptotic covariance matrix of Proposition~\ref{loijointe}.

By the stationarity of $(\epsilon_t)_{t\in\mathbb{Z}}$ and Lebesgue's dominated convergence theorem, we obtain the $(\ell,\ell^{'})$-th entry of the matrix $\Gamma_{m,m}$:
\begin{align*}
\lim_{n\rightarrow\infty}\mathrm{Cov}(\sqrt{n}\gamma(\ell),\sqrt{n}\gamma(\ell^{'}))&=\lim_{n\rightarrow\infty}\frac{1}{n}\sum_{t=\ell+1}^n\sum_{s=\ell^{'}+1}^n\mathbb{E}\left[ \epsilon_t\epsilon_{t-\ell}\epsilon_s\epsilon_{s-\ell^{'}}\right] \\
&=\sum_{h=-\infty}^{\infty}\mathbb{E}\left[ \epsilon_t\epsilon_{t-\ell}\epsilon_{t-h}\epsilon_{t-h-\ell^{'}}\right]:=\Gamma(\ell,\ell^{'}).
\end{align*}
We thus have $\Gamma_{m,m}=[\Gamma(\ell,\ell^{'})]_{1\leq\ell,\ell'\leq m}$.

Finally, by the stationarity of $(\epsilon_t)_{t\in\mathbb{Z}}$ and $(\epsilon_t\partial\epsilon_t(\theta_0)/\partial\theta)_{t\in\mathbb{Z}}$ we have
\begin{align*}
\mathrm{Cov}\left(-2J^{-1}\frac{1}{\sqrt{n}}\sum_{t=1}^n\epsilon_t\frac{\partial}{\partial\theta}\epsilon_t(\theta_0),\sqrt{n}\gamma(\ell^{'}) \right)
&=-2J^{-1}\frac{1}{{n}}\sum_{t=1}^n\sum_{t'=\ell^{'}+1}^n\mathrm{Cov}\left(\epsilon_t\frac{\partial}{\partial\theta}\epsilon_t(\theta_0),
\epsilon_{t'}\epsilon_{t'-\ell^{'}}\right)\\
&=-2J^{-1}\frac{1}{{n}}\sum_{h=-n+1}^{n-1}(n-|h|)\mathrm{Cov}\left(\epsilon_t\frac{\partial\epsilon_t(\theta_0)}{\partial\theta},
\epsilon_{t-h}\epsilon_{t-\ell^{'}-h}\right).
\end{align*}
By the dominated convergence theorem  and from \eqref{scoreVec}, it follows that
\begin{align*}
\lim_{n\rightarrow\infty}\mathrm{Cov}\left(-2J^{-1}\frac{1}{\sqrt{n}}\sum_{t=1}^n\epsilon_t\frac{\partial}{\partial\theta}\epsilon_t(\theta_0),\sqrt{n}\gamma(\ell^{'}) \right)
&=-2J^{-1}\sum_{h=-\infty}^{\infty}\mathrm{Cov}\left(\epsilon_t\frac{\partial}{\partial\theta}\epsilon_t(\theta_0),
\epsilon_{t-h}\epsilon_{t-\ell^{'}-h}\right)\\
&=-2J^{-1}\sum_{j\geq1}{\Lambda}_{j}\left(\theta_0\right)\sum_{h=-\infty}^{\infty}\mathbb{E}\left(\epsilon_t\epsilon_{t-j}
\epsilon_{t-h}\epsilon_{t-\ell^{'}-h}\right)
\\
&=-2J^{-1}\sum_{j\geq1}{\Lambda}_{j}\left(\theta_0\right)\Gamma(j,\ell'):=\Sigma_{\hat{\theta},\gamma_m}(\cdot,\ell').
\end{align*}
%The existence of the matrix $\Xi$ is guaranteed by the existence of $J^{-1}IJ^{-1}$, $\Gamma(\ell,\ell^{'})$ and $\Sigma_{\hat{\theta},\gamma_m}(\ell,\ell^{'})$.
It is clear that the existence of the above matrices is ensured by the existence of
$\Gamma(\ell,\ell^{'})$ and $\sum_{\ell,\ell'=1}^\infty\|{\Lambda}_{\ell}(\theta_0){\Lambda'}_{\ell'}(\theta_0)\Gamma(\ell,\ell^{'})\|$. %The proof is then complete.
The proof will thus follow from Lemma~\ref{lemGamma} below.
$\hfill\square$

We now justify the existence of the $\Gamma(\ell,\ell^{'})$ and $\sum_{\ell,\ell'=1}^\infty\|{\Lambda}_{\ell}(\theta_0){\Lambda'}_{\ell'}(\theta_0)\Gamma(\ell,\ell^{'})\|$ in the following result.
\begin{lemme}\label{lemGamma}
Under the assumptions {\bf (A0)} and {\bf (A3')} with $\tau=4$, %there exits a finite constant $K$ such that
we have for $(\ell,\ell')\neq(0,0)$
\begin{align}\label{Gam_exist}
\Gamma(\ell,\ell^{'})=\sum_{h=-\infty}^{\infty}\mathbb{E}\left( \epsilon_t\epsilon_{t-\ell}\epsilon_{t-h}\epsilon_{t-h-\ell^{'}}\right)&<\infty
%\quad\text{ and }
\\ \label{sum_Gam_exist}
\text{and}\quad\sum_{\ell,\ell'=1}^\infty\left\|{\Lambda}_{\ell}\left(\theta_0\right){\Lambda'}_{\ell'}\left(\theta_0\right)\Gamma(\ell,\ell^{'})\right\|&
%=\sum_{\ell,\ell'=1}^\infty\left\|\sum_{h=-\infty}^{\infty}{\Lambda}_{\ell}\left(\theta_0\right){\Lambda'}_{\ell'}\mathbb{E}\left( \epsilon_t\epsilon_{t-\ell}\epsilon_{t-h}\epsilon_{t-h-\ell^{'}}\right)\right\|
<\infty.
\end{align}
\end{lemme}
\begin{proof} Note that, for all $h\in\mathbb{Z}$ and all $(\ell,\ell')\neq(0,0)$ we have
\begin{align*}
\left|\mathbb{E}\left[ \epsilon_t\epsilon_{t-\ell}\epsilon_{t-h}\epsilon_{t-h-\ell^{'}}\right]\right|&\leq\left|\mathrm{cum}\left( \epsilon_t,\epsilon_{t-\ell},\epsilon_{t-h},\epsilon_{t-h-\ell^{'}}\right)\right|+\left|\mathbb{E}\left[ \epsilon_t\epsilon_{t-\ell}\right]\right|\left|\mathbb{E}\left[ \epsilon_{t-h}\epsilon_{t-h-\ell^{'}}\right]\right|\\
&\quad+\left|\mathbb{E}\left[ \epsilon_t\epsilon_{t-h}\right]\right|\left|\mathbb{E}\left[ \epsilon_{t-\ell}\epsilon_{t-h-\ell^{'}}\right]\right|+\left|\mathbb{E}\left[ \epsilon_t\epsilon_{t-h-\ell^{'}}\right]\right|\left|\mathbb{E}\left[ \epsilon_{t-\ell}\epsilon_{t-h}\right]\right|.
\end{align*}
Then, using the stationarity of $(\epsilon_t)_{t\in\mathbb{Z}}$, and under the assumptions {\bf (A0)} and {\bf (A3')} with $\tau=4$ it follows that
\begin{align*}
\Gamma(\ell,\ell^{'})&\leq \left[\mathbb{E}\left( \epsilon_t^2\right)\right]^2+\sum_{h=-\infty}^{\infty}\left|\mathrm{cum}\left( \epsilon_0,\epsilon_{-\ell},\epsilon_{-h},\epsilon_{-h-\ell^{'}}\right)\right|\leq K
\end{align*}
which proves \eqref{Gam_exist}. Similarly, we obtain
\begin{align*}
\sum_{\ell,\ell'=1}^\infty\left\|{\Lambda}_{\ell}\left(\theta_0\right){\Lambda'}_{\ell'}\left(\theta_0\right)\Gamma(\ell,\ell^{'})\right\|&\leq \sum_{h=-\infty}^{\infty}
\sum_{\ell,\ell'=1}^\infty\left|\mathrm{cum}\left( \epsilon_0,\epsilon_{-\ell},\epsilon_{-h},\epsilon_{-h-\ell^{'}}\right)\right|\\&\qquad+\left[\mathbb{E}\left( \epsilon_t^2\right)\right]^2\sum_{\ell=1}^\infty\left\|{\Lambda}_{\ell}\left(\theta_0\right)\right\|^2\\&\leq K
\end{align*}
where we have used Lemma \ref{miss}. The conclusion follows.
\end{proof}

\subsection{Proof of Theorem~\ref{loi_res_rho}}\label{proof_loi_res}
The proof is divided in two steps.   % We follow the ideas presented in subsection \ref{sn}.
\subsubsection{Step 1: Taylor's expansion of $\sqrt{n}\hat\gamma_m$ and $\sqrt{n}\hat{\rho}_m$}
The aim of this step is to prove \eqref{hat_gamma} and \eqref{hat_rho}. First we prove  that for $h=1,\dots,m$
\begin{align}\label{Tayl_gamh}
\sqrt{n}\hat{\gamma}(h)&=\sqrt{n}\gamma(h)+\left( \mathbb{E}\left[ \epsilon_{t-h}\frac{\partial}{\partial\theta^{'}}\epsilon_t(\theta_0)\right] \right)\sqrt{n}\left(\hat{\theta}_n-\theta_0\right)+\mathrm{o}_{\mathbb{P}}(1).
\end{align}
A Taylor expansion of $(1/\sqrt{n})\sum_{t=1+h}^n\tilde{\epsilon}_t(\cdot)\tilde{\epsilon}_{t-h}(\cdot)$ around $\theta_0$ gives
\begin{align*}
\sqrt{n}\hat{\gamma}(h)&=\frac{1}{\sqrt{n}}\sum_{t=1+h}^n\tilde{\epsilon}_t(\theta_0)\tilde{\epsilon}_{t-h}(\theta_0)+\left( \frac{1}{n}\sum_{t=1+h}^n\tilde{D}_t(\theta^{*}_n) \right)\sqrt{n}\left(\hat{\theta}_n-\theta_0\right)
%\\
%&=\sqrt{n}\gamma(h)+\left( \frac{1}{n}\sum_{t=1+h}^n\tilde{D}_t(\theta^{*}_n)\right)\sqrt{n}\left(\hat{\theta}_n-\theta_0\right)+R_{n,h,1}
\\&=\sqrt{n}\gamma(h)+\left( \mathbb{E}\left[ D_t(\theta_0)\right] \right) \sqrt{n}\left(\hat{\theta}_n-\theta_0\right)+R_{n,h,1}+R_{n,h,2}+R_{n,h,3},
\end{align*}
where
\begin{align*}
\tilde{D}_t(\theta)&= \frac{\partial\tilde{\epsilon}_t(\theta)}{\partial\theta^{'}}\tilde{\epsilon}_{t-h}(\theta)+\tilde{\epsilon}_{t}(\theta)\frac{\partial\tilde{\epsilon}_{t-h}(\theta)}{\partial\theta^{'}},
\\
D_t(\theta_0)&= \frac{\partial\epsilon_t(\theta_0)}{\partial\theta^{'}}\epsilon_{t-h}+\epsilon_{t}\frac{\partial\epsilon_{t-h}(\theta_0)}{\partial\theta^{'}},
\\
R_{n,h,1}&=\frac{1}{\sqrt{n}}\sum_{t=1+h}^n\left\lbrace \tilde{\epsilon}_t(\theta_0)\tilde{\epsilon}_{t-h}(\theta_0)-\epsilon_t(\theta_0)\epsilon_{t-h}(\theta_0)\right\rbrace,
\\
R_{n,h,2}&=\left( \frac{1}{n}\sum_{t=1+h}^n\left( \tilde{D}_t(\theta^{*}_n)-D_t(\theta_0)\right) \right)\sqrt{n}\left(\hat{\theta}_n-\theta_0\right),
\\
R_{n,h,3}&=\left(\frac{1}{n}\sum_{t=1+h}^nD_t(\theta_0) -\mathbb{E}\left[ D_t(\theta_0)\right]\right)\sqrt{n}\left(\hat{\theta}_n-\theta_0\right),
\end{align*}
and where $\theta^{*}_n$ is between $\hat{\theta}_n$ and $\theta_0$.
Using the orthogonality between $\epsilon_t$ and any linear combination of the past values of $\epsilon_t$ (in particular $\partial\epsilon_{t-h}/\partial\theta$), we have
%Finally, since $h=1,\dots,m$, the random variable $\epsilon_t$ is uncorrelated with $\partial\epsilon_{t-h}/\partial\theta$, then
\begin{align}\label{Tayl_R1_R3}
\sqrt{n}\hat{\gamma}(h)&=\sqrt{n}\gamma(h)+\left( \mathbb{E}\left[ \epsilon_{t-h}\frac{\partial}{\partial\theta^{'}}\epsilon_t(\theta_0)\right] \right)\sqrt{n}\left(\hat{\theta}_n-\theta_0\right)+R_{n,h,1}+R_{n,h,2}+R_{n,h,3}.
\end{align}
Thus, to obtain \eqref{Tayl_gamh}, we just need to prove that in \eqref{Tayl_R1_R3} the sequences of random variables $(R_{n,h,1})_{n\geq 1}$, $(R_{n,h,2})_{n\geq 1}$ and $(R_{n,h,3})_{n\geq 1}$ converge in probability to 0.

One of the three above term is easy to handle. Indeed, by the ergodic theorem, we have
${n}^{-1}\sum_{t=1+h}^nD_t(\theta_0) -\mathbb{E}\left[ D_t(\theta_0)\right]\to 0$
almost-surely as $n\to\infty$. Thus using the tightness of the sequence $(\sqrt{n}(\hat{\theta}_n-\theta_0))_n$, we deduce that  $R_{n,h,3}=\mathrm{o}_{\mathbb{P}}(1)$.

The proof of \eqref{Tayl_gamh} will thus follow from Lemmas~\ref{lemR1} and~\ref{lemR2R3} in which the two others terms $R_{n,h,1}$ and $R_{n,h,2}$ are discussed. These lemmas are stated and proved hereafter (see subsections \ref{l4} and \ref{l5}).

%Now, we come back to the vector $\hat\gamma_m=(\hat\gamma(1),\dots,\hat\gamma(m))'$.
We now remark that in Equation \eqref{Tayl_gamh},
$\mathbb{E}[\epsilon_{t-h}({\partial\epsilon_t(\theta_0)}/{\partial\theta^{'}})]$ is the line $h$ of
the matrix $\Psi_m\in\mathbb R^{m\times(p+q+1)}$ defined by \eqref{Psi_m}. So as $h=1,\dots,m$, Equation \eqref{Tayl_gamh} becomes
\begin{align*}
\sqrt{n}\hat{\gamma}_m =\left( \sqrt{n}\hat{\gamma}(1),\dots,\sqrt{n}\hat{\gamma}(m)\right)^{'}=
\sqrt{n}\gamma_m+\Psi_m\sqrt{n}\left(\hat{\theta}_n-\theta_0\right)+\mathrm{o}_{\mathbb{P}}(1).
\end{align*}
Therefore the Taylor expansion \eqref{hat_gamma} of $\hat\gamma_m$ is proved.

Now, it is clear that the asymptotic distribution of the residual autocovariances $\sqrt{n}\hat{\gamma}_m$ is related to
the asymptotic behavior of $\sqrt{n}( \hat{\theta}_n'-\theta_0^{'},\gamma_m^{'})^{'}$ obtained in Subsection~\ref{proof_loijointe}.
We come back to the vector $\hat\rho_m=(\hat\rho(1),\dots,\hat\rho(m))'$. Note that from \eqref{Tayl_gamh}, we have
$\sqrt{n}(\hat{\gamma}(0)-\gamma(0))=\mathrm{o}_{\mathbb{P}}(1)$. Applying the CLT for mixing
processes (see \cite{herr}) to the process $(\epsilon_t^2)_{t\in\mathbb{Z}}$, we obtain
\begin{align*}
\sqrt{n}\left(\hat\sigma_\epsilon^2-\sigma_\epsilon^2\right)&=\frac{1}{\sqrt{n}}\sum_{t=1}^n\left( \epsilon_{t}^2-\mathbb{E}[\epsilon_t^2] \right)+\mathrm{o}_{\mathbb{P}}(1)
\xrightarrow[n\rightarrow \infty]{\text{in law}}\mathcal{N}\left(0,\sum_{h=-\infty}^{\infty}\mathrm{Cov}\left(\epsilon_t^2,
\epsilon_{t-h}^2\right)\right).
\end{align*}
So we have $\sqrt{n}(\hat\sigma_\epsilon^2-\sigma_\epsilon^2)=\mathrm{O}_{\mathbb{P}}(1)$ and
 $\sqrt{n}({\gamma}(0)-\sigma_\epsilon^2)=\mathrm{O}_{\mathbb{P}}(1)$. Now,
 using \eqref{loi_hat_gamma} and the ergodic theorem, we have
$$n\left(\frac{\hat{\gamma}(h)}{\hat{\gamma}(0)}- \frac{\hat{\gamma}(h)}{\sigma_{\epsilon}^2}\right)=\sqrt{n}\hat{\gamma}(h)\frac{\sqrt{n}\left(\sigma_{\epsilon}^2-\hat{\gamma}(0)\right)}{\sigma_{\epsilon}^2\hat{\gamma}(0)}=\mathrm{O}_{\mathbb{P}}(1),$$
which means $\sqrt{n}\hat{\rho}(h)={\sqrt{n}\hat{\gamma}(h)}/{\sigma_{\epsilon}^2}+\mathrm{O}_{\mathbb{P}}(n^{-1/2}).$
Since $h=1,\dots,m$, it follows that
$$\sqrt{n}\hat{\rho}_m=\frac{\sqrt{n}\hat{\gamma}_m}{\sigma_{\epsilon}^2}+\mathrm{o}_{\mathbb{P}}(1),$$
and the Taylor expansion \eqref{hat_rho} of $\hat\rho_m$ is proved.
This ends our first step.

The next step deals with the asymptotic distributions of $\sqrt{n}\hat{\gamma}_m$ and $\sqrt{n}\hat{\rho}_m$.
\subsubsection{Step 2: asymptotic distributions of $\sqrt{n}\hat{\gamma}_m$ and $\sqrt{n}\hat{\rho}_m$}
%Our purpose is to prove the asymptotic distribution of $\sqrt{n}\hat{\gamma}_m$ and $\sqrt{n}\hat{\rho}_m$.
The  joint asymptotic distribution of $\sqrt{n}\gamma_m$ and $\sqrt{n}(\hat{\theta}_n-\theta_0)$ shows that $\sqrt{n}\hat{\gamma}_m$ has a limiting normal distribution with mean zero and covariance matrix
\begin{align*}
\lim_{n\rightarrow\infty}\mathrm{Var}\left(\sqrt{n}\hat{\gamma}_m\right)&=\lim_{n\rightarrow\infty}\mathrm{Var}\left(\sqrt{n}\gamma_m\right)+\Psi_m\lim_{n\rightarrow\infty}\mathrm{Var}\left(\sqrt{n}(\hat{\theta}_n-\theta_0)\right)\Psi_m^{'}\\
&+\Psi_m\lim_{n\rightarrow\infty}\mathrm{Cov}\left(\sqrt{n}(\hat{\theta}_n-\theta_0),\sqrt{n}\gamma_m\right)+\lim_{n\rightarrow\infty}\mathrm{Cov}\left(\sqrt{n}\gamma_m,\sqrt{n}(\hat{\theta}_n-\theta_0)\right)\Psi_m^{'}\\
&=\Gamma_{m,m}+\Psi_m\Sigma_{\hat{\theta}}\Psi_m^{'}+\Psi_m\Sigma_{\hat{\theta},\gamma_m}+\Sigma_{\hat{\theta},\gamma_m}^{'}\Psi_m^{'}.
\end{align*}
Consequently, we have
$$\lim_{n\rightarrow\infty}\mathrm{Var}\left(\sqrt{n}\hat{\rho}_m\right)=
\lim_{n\rightarrow\infty}\mathrm{Var}\left(\sqrt{n}\frac{\hat{\gamma}_m}{\sigma_{\epsilon}^2}\right)=\frac{1}{\sigma_{\epsilon}^4}\Sigma_{\hat{\gamma}_m}.$$
This ends our second step and the proof is completed.$\hfill\square$

In the following, we  justify the convergence of $R_{n,h,1}$ and $R_{n,h,2}$.
\subsubsection{Step 3: convergence of $R_{n,h,1}$}\label{l4}

\begin{lemme}\label{lemR1} Under the assumptions of Theorem~\ref{loi_res_rho},  the sequence of random variables
\begin{align}\label{R1_o_p}
R_{n,h,1}=\frac{1}{\sqrt{n}}\sum_{t=1+h}^n\left\lbrace \tilde{\epsilon}_t(\theta_0)\tilde{\epsilon}_{t-h}(\theta_0)-\epsilon_t(\theta_0)\epsilon_{t-h}(\theta_0)\right\rbrace
\end{align}
converges in probability to zero as $n\rightarrow\infty$.
\end{lemme}
%{\bf Proof of Lemma~\ref{lemR1R2R3}.}\marge{a faire plus tard}
\begin{proof}
Throughout this proof, $\theta=(\theta_1,\dots,\theta_{p+q},d)'\in\Theta_\delta$ is such that  $\max(d_0,0)<d\leq d_2$
 where $d_2$ is the upper bound of the support of the long-range parameter $d_0$. Let
\begin{align}\label{R1_o_p1}
R_{n,h,1}^1&=\frac{1}{\sqrt{n}}\sum_{t=1+h}^n\left\lbrace \tilde{\epsilon}_t(\theta_0)-\epsilon_t(\theta_0)\right\rbrace \tilde{\epsilon}_{t-h}(\theta_0)%\qquad\text{and}
\\
\text{and}\quad R_{n,h,1}^2&=\frac{1}{\sqrt{n}}\sum_{t=1+h}^n\epsilon_t(\theta_0)\left\lbrace \tilde{\epsilon}_{t-h}(\theta_0)-\epsilon_{t-h}(\theta_0)\right\rbrace.\label{R1_o_p2}
\end{align}
The lemma will be proved as soon as we show that $R_{n,h,1}^1$ and $R_{n,h,1}^2$ tend to zero in probability when $n\to\infty$.
\paragraph{Proof of the  convergence in probability of $R_{n,h,1}^1$}\ \\
The arguments follow the one of Lemma 4 in \cite{BMES2019} in a simpler context.
The proof is quite long so we divide it in four steps.
\subparagraph{$\diamond$ Step 1: preliminaries.}%\ \\
We have
\begin{align*}
R_{n,h,1}^1&=\frac{1}{\sqrt{n}}\sum_{t=1+h}^n \left\lbrace \tilde{\epsilon}_t(\theta_0)-\tilde{\epsilon}_t(\theta)\right\rbrace \tilde{\epsilon}_{t-h}(\theta_0)\\
&\qquad +\frac{1}{\sqrt{n}}\sum_{t=1+h}^n \left\lbrace \tilde{\epsilon}_t(\theta)-\epsilon_t(\theta)\right\rbrace \tilde{\epsilon}_{t-h}(\theta_0)
\\&\qquad +\frac{1}{\sqrt{n}}\sum_{t=1+h}^n \left\lbrace \epsilon_t(\theta)-\epsilon_t(\theta_0)\right\rbrace \tilde{\epsilon}_{t-h}(\theta_0)\\
&= \omega_{n,h,1}(\theta)+\omega_{n,h,2}(\theta)+\omega_{n,h,3}(\theta),
\end{align*}
where
\begin{align*}
\omega_{n,h,1}(\theta)&=\frac{1}{\sqrt{n}}\sum_{t=1+h}^n \left\lbrace \tilde{\epsilon}_t(\theta_0)-\tilde{\epsilon}_t(\theta)\right\rbrace \tilde{\epsilon}_{t-h}(\theta_0),\\
\omega_{n,h,2}(\theta)&=\frac{1}{\sqrt{n}}\sum_{t=1+h}^n \left\lbrace \tilde{\epsilon}_t(\theta)-\epsilon_t(\theta)\right\rbrace \tilde{\epsilon}_{t-h}(\theta_0)
%\qquad\text{and}
\\
\text{and}\quad\omega_{n,h,3}(\theta)&=\frac{1}{\sqrt{n}}\sum_{t=1+h}^n \left\lbrace \epsilon_t(\theta)-\epsilon_t(\theta_0)\right\rbrace \tilde{\epsilon}_{t-h}(\theta_0).
\end{align*}
Therefore, if we prove that the two sequences of random variables $(\omega_{n,h,2}(\theta))_{n\geq 1}$ and  $( \omega_{n,h,1}(\theta)+\omega_{n,h,3}(\theta))_{n\geq 1}$  converge in probability to $0$, then the convergence in probability of $R_{n,h,1}^1$  to zero will be true.
\subparagraph{$\diamond$ Step 2: convergence in probability of $(\omega_{n,h,2}(\theta))_{n\geq 1}$ to $0$} \ \\
For all $\beta>0$, we have
\begin{align*}
\mathbb{P}\left(\left|\omega_{n,h,2} \right|\geq \beta\right)&\leq\frac{1}{\sqrt{n}\beta}\sum_{t=1+h}^n\mathbb{E}\left[ \left| \tilde{\epsilon}_t(\theta)-\epsilon_t(\theta)\right|\left|\tilde{\epsilon}_{t-h}(\theta_0)\right|\right] \\
&\leq \frac{1}{\sqrt{n}\beta}\sum_{t=1+h}^n\left\| \tilde{\epsilon}_t(\theta)-\epsilon_t(\theta)\right\|_{\mathbb{L}^2}\left\|\tilde{\epsilon}_{t-h}(\theta_0)\right\|_{\mathbb{L}^2}.
\end{align*}
First, from \eqref{epsiltilde-th} and using Lemma \ref{miss},  we have
\begin{align}\nonumber
\left\|\tilde{\epsilon}_{t-h}(\theta_0)\right\|^2_{\mathbb{L}^2}& =\mathbb{E}\left[ \left( \sum_{i=0}^{\infty}{\lambda}_{i}^{t-h}\left(\theta_0\right) \epsilon_{t-i-h}\right)^2 \right] \\ \nonumber
&=\sum_{i=1}^{\infty}\sum_{j=1}^{\infty}{\lambda}_{i}^{t-h}\left(\theta_0\right){\lambda}_{j}^{t-h}\left(\theta_0\right) \mathbb{E}\left[\epsilon_{t-i-h}\epsilon_{t-j-h} \right] +\sigma_{\epsilon}^2\left\lbrace {\lambda}_{0}^{t-h}\left(\theta_0\right)\right\rbrace ^2\\
& =\sigma_{\epsilon}^2\sum_{i=1}^{\infty}\left\lbrace {\lambda}_{i}^{t-h}\left(\theta_0\right)\right\rbrace ^2+\sigma_{\epsilon}^2 \nonumber \\
&\le K . \label{eps-tildeL2}
\end{align}
In view of \eqref{epsil-th}, \eqref{epsiltilde-th}  and \eqref{eps-tildeL2}, we may write
\begin{align*}
\mathbb{P}\left( \left| \omega_{n,h,2}(\theta)\right|\geq\beta\right)&\leq \frac{K}{\beta\sqrt{n}}\sum_{t=1+h}^n\left (\mathbb{E}\left[\left( \tilde{\epsilon}_{t}(\theta)-\epsilon_{t}(\theta)\right)^2\right] \right )^{1/2}\\
&\leq \frac{K}{\beta\sqrt{n}}\sum_{t=1}^{n}\left (\sum_{i\geq 0}\sum_{j\geq 0}\left(\lambda^t_i(\theta)-\lambda_i(\theta)\right)\left(\lambda^t_j(\theta)-\lambda_j(\theta)\right)\mathbb{E}\left[\epsilon_{t-i}\epsilon_{t-j}\right] \right )^{1/2}\\
&\leq \frac{\sigma_{\epsilon}K}{\beta\sqrt{n}}\sum_{t=1}^n\left (\sum_{i\geq 0}\left(\lambda^t_i(\theta)-\lambda_i(\theta)\right)^2\right )^{1/2} \\
&\leq \frac{\sigma_{\epsilon}K}{\beta\sqrt{n}}\sum_{t=1}^n\left\|\lambda(\theta)-\lambda^t(\theta)\right\|_{\ell^2}. %\\
%&\leq \frac{2\sigma_{\epsilon}\tilde{K}^{**}}{\beta\sqrt{n}}\sum_{t=1}^n\frac{1}{t^{1/2+(d-\max(d_0,0))}}\rightarrow 0
\end{align*}
We use Lemma \ref{lemme_sur_les_ecarts_des_coef}, the fact that $d>\max(d_0,0)$ and the fractional version of Ces\`{a}ro's Lemma\footnote{Recall that the fractional version of Ces\`{a}ro's Lemma states that for $(h_t)_t$ a sequence of positive real numbers, $\kappa>0$ and $c\ge 0$ we have
$$\lim_{t\rightarrow\infty}h_tt^{1-\kappa}=\left|\kappa\right|c\Rightarrow\lim_{n\rightarrow\infty}\frac{1}{n^{\kappa}}\sum_{t=0}^n h_t=c.$$
} to obtain
\begin{align*}
\mathbb{P}\left( \left| \omega_{n,h,2}(\theta)\right|\geq\beta\right)\leq \frac{\sigma_{\epsilon}K}{\beta}\, \frac{1}{\sqrt{n}}\sum_{t=1}^n\frac{1}{t^{1/2+(d-\max(d_0,0))}}\xrightarrow[n\to\infty]{} 0.
\end{align*}
This proves the expected convergence in probability.
\subparagraph{$\diamond$ Step 3: convergence in probability of $(\omega_{n,h,1}(\theta)+\omega_{n,h,3}(\theta))_{n\geq 1}$} \ \\
Note now that, for all $n\ge1$, we have
\begin{align*}
\omega_{n,h,1}(\theta)+\omega_{n,h,3}(\theta)
&=\frac{1}{\sqrt{n}}\sum_{t=1+h}^n\Big \{ \left( \epsilon_t(\theta)-\tilde{\epsilon}_t(\theta)\right)-\left( \epsilon_t(\theta_0)-\tilde{\epsilon}_t(\theta_0)\right)\Big \} \tilde{\epsilon}_{t-h}(\theta_0).
\end{align*}
A Taylor expansion of the function $(\epsilon_t-\tilde{\epsilon}_t)(\cdot)$ around $\theta_0$ gives
\begin{align}\label{popo} \Big |  (\epsilon_t(\theta)-\tilde{\epsilon}_t(\theta))-( \epsilon_t(\theta_0)-\tilde{\epsilon}_t(\theta_0))\Big | & \le \left \| \frac{\partial ( \epsilon_t-\tilde{\epsilon}_t)}{\partial\theta}(\theta^\star)\right \|_{\mathbb R^{p+q+1}}  \| \theta-\theta_0\|_{\mathbb R^{p+q+1}}
\ 
%  \nonumber \\
%& \le \sup_{0\le c\le 1} \left \| \frac{\partial ( \epsilon_t-\tilde{\epsilon}_t)}{\partial\theta}((1-c)\theta +c \theta_0)\right \|_{\mathbb R^{p+q+1}}  \| \theta-\theta_0\|_{\mathbb R^{p+q+1}} \ .
\end{align}
where $\theta^\star$ is between $\theta_0$ and $\theta$. Following the same method as in the previous step we obtain
\begin{align*}
\mathbb E \Big |\left( \epsilon_t(\theta)-\tilde{\epsilon}_t(\theta)\right)-\left( \epsilon_t(\theta_0)-\tilde{\epsilon}_t(\theta_0)\right)\Big |^2
& \le K\|\theta-\theta_0\|^2_{\mathbb R^{p+q+1}}\sum_{k=1}^{p+q+1}\mathbb E \left [\left | \frac{\partial ( \epsilon_t-\tilde{\epsilon}_t)}{\partial\theta_k}(\theta^\star)  \right |^2\right ] \nonumber\\
%&\leq K\|\theta-\theta_0\|^2_{\mathbb R^{p+q+1}}\sum_{k=1}^{p+q+1}\mathbb E \left [\left |\frac{\partial ( \epsilon_t-\tilde{\epsilon}_t)}{\partial\theta_k}(\theta^\star)  \right |^2\right ]  \nonumber \\
&\le K\|\theta-\theta_0\|^2_{\mathbb R^{p+q+1}}\sum_{k=1}^{p+q+1}\sigma_{\epsilon}^2\left\|(\overset{\textbf{.}}{\lambda_k}-\overset{\textbf{.}}{\lambda_k}^t)(\theta^\star)\right\|_{\ell^2}^2\nonumber.
\end{align*}
As in \cite{hallin}, it can be shown using Stirling's approximation and the fact that $d^\star >d_0$ that 
\begin{align*}
\left\|(\overset{\textbf{.}}{\lambda_k}-\overset{\textbf{.}}{\lambda_k}^t)(\theta^\star)\right\|_{\ell^2}\leq K\frac{1}{t^{1/2+(d^\star-d_0)-\zeta}}
\end{align*}
for any small enough $\zeta>0$. We then deduce that 
\begin{align}\label{ineq:diff}
\Big \|\left( \epsilon_t(\theta)-\tilde{\epsilon}_t(\theta)\right)-\left( \epsilon_t(\theta_0)-\tilde{\epsilon}_t(\theta_0)\right)\Big \|_{\mathbb{L}^2}&\le K\|\theta-\theta_0\|_{\mathbb R^{p+q+1}}\frac{1}{t^{1/2+(d^\star-d_0)-\zeta}}.
%&\leq K\, \|\theta-\theta_0\|^2_{\mathbb R^{p+q+1}}\sup_{d{;}d_0\le d\le d_2}\left ( \frac{1}{t^{{1/2}+(d-d_0)}} \right )^2\nonumber \\
%&\leq K\, \|\theta-\theta_0\|^2_{\mathbb R^{p+q+1}}\frac{1}{t},
\end{align}
The expected convergence in probability follows from \eqref{eps-tildeL2}, \eqref{ineq:diff} and the fractional version of Ces\`{a}ro's Lemma.
\paragraph{Proof of the  convergence in probability of $R_{n,h,1}^2$}\ \\
Under Assumption {\bf (A3)} with $\tau=2$ it follows that $\epsilon_t(\theta_0)$ belongs to $\mathbb{L}^2$.
Thus the proof of the  convergence in probability of $R_{n,h,1}^2$  to zero is shown in the same way as the proof of the  convergence in probability of $R_{n,h,1}^1$   to $0$.
\paragraph{Conclusion : convergence in probability of $R_{n,h,1}$}\ \\
The conclusion is a consequence of the above convergences.
\end{proof}
\subsubsection{Step 4: convergence of $R_{n,h,2}$}\label{l5}

\begin{lemme}\label{lemR2R3} Under the assumptions of Theorem~\ref{loi_res_rho},  the sequence of random variables
\begin{align}\label{R2_o_p}
R_{n,h,2}&=\left( \frac{1}{n}\sum_{t=1+h}^n\left( \tilde{D}_t(\theta^{*}_n)-D_t(\theta_0)\right) \right)\sqrt{n}\left(\hat{\theta}_n-\theta_0\right)
%R_{n,h,3}&=\left( \frac{1}{n}\sum_{t=1+h}^nD_t(\theta_0) -\mathbb{E}\left[ D_t(\theta_0)\right]\right)\sqrt{n}\left(\hat{\theta}_n-\theta_0\right),\label{R3_o_p}
\end{align}
tends to zero in probability as $n\to\infty$ and where $\theta^{*}_n$ is between $\hat{\theta}_n$ and $\theta_0$.
\end{lemme}
%{\bf Proof of Lemma~\ref{lemR1R2R3}.}\marge{a faire plus tard}
\begin{proof}
Since $(\sqrt{n}(\hat{\theta}_n-\theta_0))_n$ is a tight sequence, we have $\sqrt{n}(\hat{\theta}_n-\theta_0)=\mathrm{O}_{\mathbb{P}}(1)$. Hence, to prove the convergence in probability of $(R_{n,h,2})_n$ to 0, it suffices to show that
\begin{align}\label{ecart_D}\frac{1}{n}\sum_{t=1+h}^n\left( \tilde{D}_t(\theta^{*}_n)-D_t(\theta_0)\right) =\mathrm{o}_{\mathbb{P}}(1).
\end{align}
This will be proved using Remark \ref{rmq:important} and Ces\`{a}ro's Lemma. Nevertheless, the proof is quite long so we divide it in four steps.
\paragraph{$\diamond$ Step 1: preliminaries.}%\ \\
We have
%This is because
%\begin{align*}
%\tilde{D}_t(\theta^{*}_n)&= \frac{\partial\tilde{\epsilon}_t(\theta^{*}_n)}{\partial\theta^{'}}\left( \tilde{\epsilon}_{t-h}(\theta^{*}_n)-\epsilon_{t-h}(\theta^{*}_n)\right) +\left( \tilde{\epsilon}_{t}(\theta^{*}_n)-\epsilon_{t}(\theta^{*}_n)\right) \frac{\partial\tilde{\epsilon}_{t-h}(\theta^{*}_n)}{\partial\theta^{'}}+\frac{\partial\tilde{\epsilon}_t(\theta^{*}_n)}{\partial\theta^{'}}\epsilon_{t-h}(\theta^{*}_n)+\epsilon_{t}(\theta^{*}_n) \frac{\partial\tilde{\epsilon}_{t-h}(\theta^{*}_n)}{\partial\theta^{'}}\\
%&=\frac{\partial\tilde{\epsilon}_t(\theta^{*}_n)}{\partial\theta^{'}}\left( \tilde{\epsilon}_{t-h}(\theta^{*}_n)-\epsilon_{t-h}(\theta^{*}_n)\right) +\left( \tilde{\epsilon}_{t}(\theta^{*}_n)-\epsilon_{t}(\theta^{*}_n)\right) \frac{\partial\tilde{\epsilon}_{t-h}(\theta^{*}_n)}{\partial\theta^{'}}\\
%&+\left( \frac{\partial\tilde{\epsilon}_t(\theta^{*}_n)}{\partial\theta^{'}}-\frac{\partial\epsilon_t(\theta^{*}_n)}{\partial\theta^{'}}\right) \epsilon_{t-h}(\theta^{*}_n)+\epsilon_{t}(\theta^{*}_n) \left( \frac{\partial\tilde{\epsilon}_{t-h}(\theta^{*}_n)}{\partial\theta^{'}}-\frac{\partial\epsilon_{t-h}(\theta^{*}_n)}{\partial\theta^{'}}\right) +D_t(\theta^{*}_n).\\
%\end{align*}
%\noindent Which implies
\begin{align*}
\frac{1}{n}\sum_{t=1+h}^n\left( \tilde{D}_t(\theta^{*}_n)-D_t(\theta_0)\right)&=T_{n,h,1}(\theta^{*}_n)+T_{n,h,2}(\theta^{*}_n)+T_{n,h,3}(\theta^{*}_n)+T_{n,h,4}(\theta^{*}_n)+T_{n,h,5}(\theta^{*}_n),
\end{align*}
where
\begin{align*}
T_{n,h,1}(\theta)&=\frac{1}{n}\sum_{t=1+h}^n\frac{\partial\tilde{\epsilon}_t(\theta)}{\partial\theta^{'}}\left( \tilde{\epsilon}_{t-h}(\theta)-\epsilon_{t-h}(\theta)\right), \\ T_{n,h,2}(\theta)&=\frac{1}{n}\sum_{t=1+h}^n\left( \tilde{\epsilon}_{t}(\theta)-\epsilon_{t}(\theta)\right) \frac{\partial\tilde{\epsilon}_{t-h}(\theta)}{\partial\theta^{'}},\\
T_{n,h,3}(\theta)&=\frac{1}{n}\sum_{t=1+h}^n\left( \frac{\partial\tilde{\epsilon}_t(\theta)}{\partial\theta^{'}}-\frac{\partial\epsilon_t(\theta)}{\partial\theta^{'}}\right) \epsilon_{t-h}(\theta), \\ T_{n,h,4}(\theta)&=\frac{1}{n}\sum_{t=1+h}^n\epsilon_{t}(\theta) \left( \frac{\partial\tilde{\epsilon}_{t-h}(\theta)}{\partial\theta^{'}}-\frac{\partial\epsilon_{t-h}(\theta)}{\partial\theta^{'}}\right)
%\quad{and}
\\
\text{and}\quad
T_{n,h,5}(\theta)&=\frac{1}{n}\sum_{t=1+h}^n\left(D_t(\theta)-D_t(\theta_0) \right) .
\end{align*}
Therefore, if we prove that the five sequences of random variables $(T_{n,h,i}(\theta_n^*))_n$ (for $i=1,\dots,5$)  converge in probability to $0$, then  \eqref{ecart_D} will be true.
\paragraph{$\diamond$ Step 2: convergence in probability of $(T_{n,h,1}(\theta_n^*))_n$ to $0$} \ \\
For all $\beta>0$, we have
\begin{align*}
\mathbb{P}\left(\left\|T_{n,h,1}(\theta_n^*)\right\|\geq \beta\right)&\leq \frac{1}{n\beta}\sum_{t=1+h}^n\mathbb{E}\left[ \left\|\frac{\partial\tilde{\epsilon}_t(\theta_n^*)}{\partial\theta^{'}}\right\|\left| \tilde{\epsilon}_{t-h}(\theta_n^*)-\epsilon_{t-h}(\theta_n^*)\right|\right] \\
&\leq \frac{1}{n\beta}\sum_{t=1+h}^n\left\|\tilde{\epsilon}_{t-h}(\theta_n^*)-\epsilon_{t-h}(\theta_n^*)\right\|_{\mathbb{L}^2}\left\|\left\|\frac{\partial\tilde{\epsilon}_t(\theta_n^*)}{\partial\theta^{'}}\right\|_{\mathbb{R}^{p+q+1}}\right\|_{\mathbb{L}^2}.
\end{align*}
First, from \eqref{epsiltilde-th} and using Lemma \ref{miss} we have
\begin{align}\nonumber
\left\|\left\|\frac{\partial}{\partial\theta^{'}}\tilde{\epsilon}_t(\theta_n^*)\right\|_{\mathbb{R}^{p+q+1}}\right\|^2_{\mathbb{L}^2}& \leq K\sum_{k=1}^{p+q+1}\mathbb{E}\left[ \left( \sum_{i=1}^{\infty}\overset{\textbf{.}}{\lambda}_{i,k}^t(\theta_n^*) \epsilon_{t-i}\right)^2 \right] \\ \nonumber
&\leq K\sum_{k=1}^{p+q+1}\sup_{\theta\in\Theta_\delta}\mathbb{E}\left[ \left( \sum_{i=1}^{\infty}\overset{\textbf{.}}{\lambda}_{i,k}^t(\theta) \epsilon_{t-i}\right)^2 \right]\\
&\leq K\sigma_\epsilon^2\sum_{k=1}^{p+q+1}\sup_{\theta\in\Theta_\delta}\sum_{i=1}^{\infty}\left(\overset{\textbf{.}}{\lambda}_{i,k}^t(\theta)\right)^2 \\
%& =\sigma_{\epsilon}^2\sum_{i=1}^{\infty}\left\lbrace \overset{\textbf{.}}{\lambda}_{i,k}^t(\theta)\right\rbrace ^2 \nonumber \\
&\le K , \label{dev-eps-tildeL2}
\end{align}
where we have used the fact that the function $$\theta\mapsto \mathbb E \left [\left | \frac{\partial \tilde{\epsilon}_t)}{\partial\theta_k}(\theta)  \right |^2\right ]$$ is bounded and continuous. In view of \eqref{epsil-th}, \eqref{epsiltilde-th}, \eqref{dev-eps-tildeL2} and following the same way as the step 2 of Lemma \ref{lemR1} we have
\begin{align*}
\mathbb{P}\left( \left|T_{n,h,1}(\theta_n^*)\right|\geq\beta\right)&\leq \frac{K}{\beta n}\sum_{t=1+h}^n\left (\mathbb{E}\left[\left( \tilde{\epsilon}_{t-h}(\theta_n^*)-\epsilon_{t-h}(\theta_n^*)\right)^2\right] \right )^{1/2}\\
&\leq \frac{K}{\beta n}\sum_{t=1+h}^n\sup_{\theta\in\Theta_\delta}\left (\mathbb{E}\left[\left( \tilde{\epsilon}_{t-h}(\theta)-\epsilon_{t-h}(\theta)\right)^2\right] \right )^{1/2}\\
&\leq \frac{K}{\beta n}\sum_{t=1}^{n-h}\sup_{\theta\in\Theta_\delta}\left (\sum_{i\geq 0}\sum_{j\geq 0}\left(\lambda^t_i(\theta)-\lambda_i(\theta)\right)\left(\lambda^t_j(\theta)-\lambda_j(\theta)\right)\mathbb{E}\left[\epsilon_{t-i}\epsilon_{t-j}\right] \right )^{1/2}\\
&\leq \frac{\sigma_{\epsilon}K}{\beta n}\sum_{t=1}^n\sup_{\theta\in\Theta_\delta}\left (\sum_{i\geq 0}\left(\lambda^t_i(\theta)-\lambda_i(\theta)\right)^2\right )^{1/2} \\
&\leq \frac{\sigma_{\epsilon}K}{\beta n}\sum_{t=1}^n\sup_{\theta\in\Theta_\delta}\left\|\lambda(\theta)-\lambda^t(\theta)\right\|_{\ell^2}. %\\
%&\leq \frac{2\sigma_{\epsilon}\tilde{K}^{**}}{\beta\sqrt{n}}\sum_{t=1}^n\frac{1}{t^{1/2+(d-d_0)}}\rightarrow 0
\end{align*}
We use Remark \ref{rmq:important}, the fact that $d_1-d_0>-1/2$ and Ces\`{a}ro's Lemma to obtain
\begin{align*}
\mathbb{P}\left( \left|T_{n,h,1}(\theta_n^*)\right|\geq\beta\right)\leq \frac{\sigma_{\epsilon}K}{\beta}\, \frac{1}{n}\sum_{t=1}^n\frac{1}{t^{1/2+(d_1-d_0)}}\xrightarrow[n\to\infty]{} 0.
\end{align*}
This proves the expected convergence in probability of $T_{n,h,1}(\theta_n^*)$.

The same calculations holds for the sequences of random variables $(T_{n,h,2}(\theta_n^*))_n$, $(T_{n,h,3}(\theta_n^*))_n$ and $(T_{n,h,4}(\theta_n^*))_n$.

\paragraph{$\diamond$ Step 3: convergence in probability of $(T_{n,h,5}(\theta_n^*))_n$ to $0$} \ \\
%$$\lim_{n\rightarrow\infty}\left\|J_n^{*}(\theta^{*}_n)-J_n^{*}(\theta_0)\right\|=0 \ \ \ \text{a.s.}$$

For $1\le i,j\le p+q+1$, $t,s\in\mathbb{Z}$ and $\theta_n^{**}$ between $\theta_n^*$ and $\theta_0$, one has in view of \eqref{deriveesecepsil} and Remark \ref{rmq:important} 
\begin{align}\label{der1_croise}
\mathbb{E}&\left[ \left|\frac{\partial}{\partial\theta_i}\epsilon_t(\theta_n^{**})\frac{\partial}{\partial\theta_j}\epsilon_{s}(\theta_n^{**})\right|\right]\nonumber\\
&\leq\sup_{\theta\in\Theta_\delta}\left(\mathbb{E}\left[ \left| \sum_{k\geq 1}\overset{\textbf{.}}{\lambda}_{k,i}\left(\theta\right)\epsilon_{t-k}\right|^2 \right]\right)^{1/2}\sup_{\theta\in\Theta_\delta}\left(\mathbb{E}\left[ \left| \sum_{k\geq 1}\overset{\textbf{.}}{\lambda}_{k,i}\left(\theta\right)\epsilon_{s-k}\right|^2 \right]\right)^{1/2} \nonumber\\
&\leq K\, \sigma_{\epsilon}^2\left(\sup_{\theta\in\Theta_\delta}\left\|\overset{\textbf{.}}{\lambda}_{k}\left(\theta\right)\right\|_{\ell^2}\right)^2 \nonumber\\
& \leq K. 
\end{align}
%For $1\le i,j\le p+q+1$ and in view of  \eqref{AR-Inf}, \eqref{Coef-Gamma}, we have
%\begin{align*}
%\sup_{\theta\in\Theta_{\delta}}\left|\frac{\partial}{\partial\theta_i}\epsilon_t(\theta)\frac{\partial}{\partial\theta_j}\epsilon_t(\theta) \right|&=\sup_{\theta\in\Theta_{\delta}}\left|\sum_{k_1,k_2\geq 1}\frac{\partial}{\partial\theta_i}\gamma_{k_1}(\theta)\frac{\partial}{\partial\theta_j}\gamma_{k_2}(\theta)X_{t-k_1}X_{t-k_2}\right|\\
%&\leq \sum_{k_1,k_2\geq 1}\sup_{\theta\in\Theta_{\delta}}\left|\frac{\partial}{\partial\theta_i}\gamma_{k_1}(\theta)\right|\sup_{\theta\in\Theta_{\delta}}\left|\frac{\partial}{\partial\theta_j}\gamma_{k_2}(\theta)\right|\left|X_{t-k_1}\right|\left|X_{t-k_2}\right|\\
%&\leq K\sum_{k_1,k_2\geq 1}\log(k_1)k_1^{-1-d_1}\log(k_2)k_2^{-1-d_1} \left|X_{t-k_1}\right|\left|X_{t-k_2}\right| .
%\end{align*}
%Consequently we obtain
%\begin{align}\nonumber
%\mathbb{E}_{\theta_0}\left[ \sup_{\theta\in\Theta_{\delta}}\left|\frac{\partial}{\partial\theta_i}\epsilon_t(\theta)\frac{\partial}{\partial\theta_j}\epsilon_t(\theta)\right|\right]&\leq K\sum_{k_1,k_2\geq 1}\log(k_1)k_1^{-1-d_1}\log(k_2)k_2^{-1-d_1} \sup_{t\in\mathbb{Z}}\mathbb{E}_{\theta_0}\left|X_t\right|^2\\
%&\le K. \label{der1_croise}
%\end{align}
Similar calculation can be done to obtain
\begin{align}\label{der2}
\mathbb{E}\left[\left|\epsilon_t(\theta_n^{**})\frac{\partial^2}{\partial\theta_i\partial\theta_j}\epsilon_{s}(\theta_n^{**}) \right|\right]<\infty.
\end{align}
A Taylor expansion of $D_t(\cdot)$ around $\theta_0$ implies that
\begin{align*}
\left\|T_{n,h,5}(\theta_n^{*})\right\|&\leq \frac{1}{n}\sum_{t=1}^n\left\|\frac{\partial}{\partial\theta}D_t(\theta_n^{**})\right\|\left\|\theta^{*}_n-\theta_0\right\|,
\end{align*}
for some $\theta_n^{**}$ between $\theta_n^*$ and $\theta_0$.
From \eqref{der1_croise} and \eqref{der2}, it follows that
\begin{align}\nonumber
\mathbb{E}\left[\left\|\frac{\partial}{\partial\theta}D_t(\theta_n^{**})\right\|\right]&=\mathbb{E}\left[ \left\|\epsilon_{t-h}(\theta_n^{**})\frac{\partial^2}{\partial\theta\partial\theta^{'}}\epsilon_t(\theta_n^{**})+\frac{\partial}{\partial\theta}\epsilon_{t-h}(\theta_n^{**})\frac{\partial}{\partial\theta^{'}}\epsilon_{t}(\theta_n^{**})\right.\right.\\
&\left.\left. \hspace{4cm}+\frac{\partial}{\partial\theta}\epsilon_{t}(\theta_n^{**})\frac{\partial}{\partial\theta^{'}}\epsilon_{t-h}(\theta_n^{**})+\epsilon_{t}(\theta_n^{**})
\frac{\partial^2}{\partial\theta\partial\theta^{'}}\epsilon_{t-h}(\theta_n^{**})\right\|\right]\nonumber\\
&\le K.\label{devD}
\end{align}
We use Equation \eqref{devD}, the ergodic theorem and the convergence in probability of $(\hat{\theta}_n-\theta_0)_n$ to 0 to deduce that
$T_{n,h,5}(\theta)$ converges in probability to $0$.
\paragraph{$\diamond$ Step 4: end of the proof of  the  convergence in probability of $R_{n,h,2}$ to zero.} \ \\
By Step 2 and 3 we deduce that
$$R_{n,h,2}=\mathrm{o}_{\mathbb{P}}(1)$$
and the convergence in probability is proved.

The proof of the lemma is completed.
\end{proof}
\subsection{Proof of Proposition~\ref{inversibleCm}}\label{proof_invCm}
The following proofs are quite technical and are adaptations of the arguments used in \cite{BMS2018}.

To prove  the invertibility of the normalized matrix $C_m$, we need to introduce the following notation.

Let $S_t(i)$ be the $i$-th component of the vector $S_t=\sum_{j=1}^{t}\left(\Lambda U_j-{\gamma}_m\right)\in\mathbb{R}^{m}$. We remark that
\begin{align}\label{recurrence_St}
S_{t-1}(i)
%&=\left(\sum_{j=1}^{t-1}\left(\Lambda U_j\right)_i\right)-(t-1)\hat{\gamma}(i)\nonumber\\
%&=S_t(i)-\left(\Lambda U_t\right)_i+\hat{\gamma}(i)\nonumber\\
&=S_t(i)-\sum_{k=1}^{p+q+1}\delta_{i,k}\epsilon_t\frac{\partial}{\partial\theta_k}\epsilon_t(\theta_0)-\epsilon_t\epsilon_{t-i}+{\gamma}(i),
\end{align}
where $\delta_{i,k}$ is the $(i,k)-$th entry of the $m\times (p+q+1)$ matrix $\Delta:=-2\Psi_mJ^{-1}$.

If the matrix $C_m$ is not invertible, there exists some real constants $c_1,\dots,c_{m}$ not all equal to zero, such that we have
\begin{align*}
\sum_{i=1}^{m}\sum_{j=1}^{m}c_jC_m(j,i)c_i=
\frac{1}{n^2}\sum_{t=1}^n\sum_{i=1}^{m}\sum_{j=1}^{m}c_jS_t(j)S_t(i)c_i=
\frac{1}{n^2}\sum_{t=1}^n\left(\sum_{i=1}^{m}c_iS_t(i)\right)^2=0,
\end{align*}
which implies that $\sum_{i=1}^{m}c_iS_t(i)=0$ for all $t\ge1$.

Then by \eqref{recurrence_St}, it would imply that
\begin{align}\label{conv-ps-principale}
\sum_{i=1}^m\sum_{k=1}^{p+q+1}c_i\delta_{i,k}\epsilon_t\frac{\partial}{\partial\theta_k}\epsilon_t(\theta_0)+\sum_{i=1}^mc_i\epsilon_t\epsilon_{t-i}%&=\sum_{i=1}^mB_i\hat{\gamma}(i)\nonumber\\
&=\sum_{i=1}^mc_i\gamma(i).
\end{align}
%Using the fact that $\mathbb{E}\|\epsilon_{t-h}({\partial\epsilon_{t}({\theta}_0)}/{\partial\theta'})\| <\infty$ and that $(\sqrt{n}(\hat{\theta}_n-\theta_0))_n$ is a tight sequence (which implies that $\|\hat{\theta}_n-\theta_0\|=\mathrm{O}_{\mathbb{P}}(1/\sqrt{n})$), from \eqref{Tayl_gamh} it follows that
%\begin{align}\label{ecart-gam}
%|\hat{\gamma}(i)-\gamma(i)|=\mathrm{O}_{\mathbb{P}}\left(\frac{1}{\sqrt{n}}\right).
%\end{align}
By the ergodic Theorem, we also have $\sum_{i=1}^m c_i\gamma(i)\to 0$ almost-surely as $n$ goes to infinity.
%Consequently  reporting this convergence and \eqref{ecart-gam} in \eqref{conv-ps-principale} implies that for all $t\ge 1$

Consequently  replacing this convergence in \eqref{conv-ps-principale} implies that for all $t\ge 1$
$$\sum_{i=1}^m\sum_{k=1}^{p+q+1}c_i\delta_{i,k}\epsilon_t\frac{\partial}{\partial\theta_k}\epsilon_t(\theta_0)+\sum_{i=1}^mc_i\epsilon_t\epsilon_{t-i}=0, \ \ \mathrm{a.s.}$$
Using \eqref{epsil-th}, it yields that
$$\epsilon_t\left\lbrace\sum_{\ell\geq 1} \left( \sum_{i=1}^m\sum_{k=1}^{p+q+1}c_i\delta_{i,k}\overset{\textbf{.}}{\lambda}_{\ell,k}\left(\theta_0\right)\right) \epsilon_{t-\ell}+\sum_{\ell=1}^mc_{\ell}\epsilon_{t-\ell}\right\rbrace =0, \ \ \mathrm{a.s.}$$
Or equivalently,
$$\epsilon_t\left\lbrace\sum_{\ell= 1}^m \left( \sum_{i=1}^mc_i\sum_{k=1}^{p+q+1}\delta_{i,k}\overset{\textbf{.}}{\lambda}_{\ell,k}\left(\theta_0\right)+c_{\ell}\right) \epsilon_{t-\ell}+\sum_{\ell\geq m+1} \left( \sum_{i=1}^mc_i\sum_{k=1}^{p+q+1}\delta_{i,k}\overset{\textbf{.}}{\lambda}_{\ell,k}\left(\theta_0\right)\right) \epsilon_{t-\ell}\right\rbrace =0, \ \ \mathrm{a.s.}$$
Thanks to Assumption {\bf (A4)}, $\epsilon_t$ has a positive density in some neighborhood of zero and then $\epsilon_t\neq 0$ almost-surely.
Hence we obtain
$$\sum_{\ell= 1}^m \left( \sum_{i=1}^mc_i\sum_{k=1}^{p+q+1}\delta_{i,k}\overset{\textbf{.}}{\lambda}_{\ell,k}\left(\theta_0\right)+c_{\ell}\right) \epsilon_{t-\ell}+\sum_{\ell\geq m+1} \left( \sum_{i=1}^mc_i\sum_{k=1}^{p+q+1}\delta_{i,k}\overset{\textbf{.}}{\lambda}_{\ell,k}\left(\theta_0\right)\right) \epsilon_{t-\ell} =0, \ \ \mathrm{a.s.}$$
Since the variance of the linear innovation process in not equal to zero, we deduce that
\[\left\{
  \begin{array}{ll}
    \sum_{i=1}^mc_i\sum_{k=1}^{p+q+1}\delta_{i,k}\overset{\textbf{.}}{\lambda}_{\ell,k}\left(\theta_0\right)+c_{\ell} = 0& \ \ \ \text{ for all } \ell\in\left\lbrace 1,\dots,m\right\rbrace \vspace{0.5cm}\\
    \sum_{i=1}^mc_i\sum_{k=1}^{p+q+1}\delta_{i,k}\overset{\textbf{.}}{\lambda}_{\ell,k}\left(\theta_0\right) = 0& \ \ \  \text{ for all } \ell\in\left\lbrace m+1,\dots\right\rbrace.
  \end{array}
\right.\]
Then we would have $c_1=\dots=c_m=0$ which is impossible. Thus we have a contradiction and the matrix $C_m\in\mathbb{R}^{m\times m}$ is non singular.
$\hfill\square$

\subsection{Proof of Theorem~\ref{sn1}}\label{proof_sn1}
We recall that the Skorokhod space $\mathbb{D}^{\ell}[ 0{,}1]$  is the set of $\mathbb{R}^{\ell}-$valued functions on $[ 0{,}1] $ which are right-continuous and have left limits everywhere. It is endowed with the Skorokhod topology and the weak convergence on $\mathbb{D}^{\ell} [ 0{,}1]$ is mentioned by $\xrightarrow[]{\mathbb{D}^{\ell}}$. The integer part of $x$ will be denoted by $\lfloor x\rfloor$.

The proof is divided in two steps.
%We firstly prove \eqref{hat_gamma}
\subsubsection{Functional central limit theorem for $(\Lambda U_t)_{t\ge 1}$}
In view of \eqref{hat_gamma} and  \eqref{U_t},  we deduce that
\begin{align}\label{sqrt_gamma_chap}
\sqrt{n}\hat{\gamma}_m&=\sqrt{n}\gamma_m+\sqrt{n}\Psi_m\left(\hat{\theta}_n-\theta_0\right)+\mathrm{o}_{\mathbb{P}}(1)\nonumber\\
&=\frac{1}{\sqrt{n}}\sum_{t=1}^nU_{2t}+\Psi_m\left( \frac{1}{\sqrt{n}}\sum_{t=1}^nU_{1t}+\mathrm{o}_{\mathbb{P}}(1)\right) +\mathrm{o}_{\mathbb{P}}(1)\nonumber\\
&=\frac{1}{\sqrt{n}}\sum_{t=1}^n\Lambda U_t+\mathrm{o}_{\mathbb{P}}(1).
\end{align}
Now, it is clear that the asymptotic behaviour of $\hat{\gamma}_m$ is related to the limit distribution of $U_t=(U_{1t}^{'},U_{2t}^{'})^{'}$.  Our first goal is to show that there exists a lower triangular matrix $\Pi$ with nonnegative diagonal entries such that
\begin{equation}\label{conv-Skorokhod-1}
\frac{1}{\sqrt{n}}\sum_{t=1}^{\lfloor nr\rfloor}\Lambda U_t\overset{\mathbb{D}^{m}}{\underset{n\rightarrow\infty}{\mathbf{\longrightarrow}}}\left(\Pi\Pi^{'}\right)^{1/2} B_m(r),
\end{equation}
where $(B_m(r))_{r\geq 0}$ is a $m-$dimensional standard Brownian motion.
Using \eqref{epsil-th}, $U_t$ can be rewritten  as
$$U_t=\left( -2\left\lbrace \sum_{i=1}^{\infty}\overset{\textbf{.}}{\lambda}_{i,1}\left(\theta_0\right) \epsilon_t\epsilon_{t-i},\dots,\sum_{i=1}^{\infty}\overset{\textbf{.}}{\lambda}_{i,p+q+1}\left(\theta_0\right) \epsilon_t\epsilon_{t-i}\right\rbrace J^{-1 '},\epsilon_t\epsilon_{t-1},\dots,\epsilon_t\epsilon_{t-m} \right) ^{'}.$$
The non-correlation between $\epsilon_t$'s implies that the process $(U_t)_{t\in\mathbb{Z}}$ of $\mathbb{R}^{p+q+1+m}$ is centered. In
order to apply the functional central limit theorem for strongly mixing process (see \cite{herr}), we need to identify the asymptotic covariance matrix in the classical central limit theorem for the sequence $(U_t)_{t\in\mathbb{Z}}$. It is proved in Proposition~\ref{loijointe} that
$$\frac{1}{\sqrt{n}}\sum_{t=1}^nU_t\xrightarrow[n\to\infty]{\text{in law}}  \ \mathcal{N}\left(0,\Xi:=2\pi f_U(0)\right),$$
where $f_U(0)$ is the spectral density of the stationary process $(U_t)_{t\in\mathbb{Z}}$ evaluated at frequency 0. The existence of the matrix $\Xi$ has already been discussed in Lemma~\ref{lemGamma}.

Since the matrix $\Xi$ is positive definite, it can be factorized as $\Xi=\Upsilon\Upsilon^{'}$, where the $(p+q+1+m)\times(p+q+1+m)$ lower triangular matrix $\Upsilon$ has nonnegative diagonal entries. Therefore, we have
$$\frac{1}{\sqrt{n}}\sum_{t=1}^n\Lambda U_t\xrightarrow[n\to\infty]{\text{in law}}  \ \mathcal{N}\left(0,\Lambda\Xi\Lambda^{'}\right),$$
and the new variance matrix can also been factorized as $\Lambda\Xi\Lambda^{'}=(\Lambda\Upsilon)(\Lambda\Upsilon)^{'}:=\Pi\Pi^{'}$, where $\Pi\in\mathbb{R}^{m\times (p+q+1)}$. Thus $$n^{-1/2}\sum_{t=1}^n(\Pi\Pi^{'})^{-1/2}\Lambda U_t\overset{\text{in law}}{\underset{n\rightarrow\infty}{\mathbf{\longrightarrow}}}\mathcal{N}(0,I_m),$$ where $(\Pi\Pi^{'})^{-1/2}$  is the Moore-Penrose  inverse (see \cite{MagnusN1988}, p. 36) of  $(\Pi\Pi^{'})^{1/2}$. %Note also that $\Pi\in\mathbb{R}^{m\times (p+q+1)}$

Using the same arguments as in the proof of Theorem 2 in \cite{BMES2019}, the asymptotic distribution of $n^{-1/2}\sum_{t=1}^nU_t$ when $n$ tends to infinity is obtained by introducing the random vector $U_t^k$ defined for any positive integer $k$ by
$$U_t^k=\left( -2\left\lbrace \sum_{i=1}^{k}\overset{\textbf{.}}{\lambda}_{i,1}\left(\theta_0\right) \epsilon_t\epsilon_{t-i},\dots,\sum_{i=1}^{k}\overset{\textbf{.}}{\lambda}_{i,p+q+1}\left(\theta_0\right) \epsilon_t\epsilon_{t-i}\right\rbrace J^{-1 '},\epsilon_t\epsilon_{t-1},\dots,\epsilon_t\epsilon_{t-m} \right) ^{'}.$$
Since $U_t^k$ depends on a finite number of values of the noise-process $(\epsilon_t)_{t\in\mathbb{Z}}$, it also satisfies a mixing property (see Theorem 14.1 in \cite{davidson1994}, p. 210).
Then applying the central limit theorem for strongly mixing process  of \cite{herr} shows  that its asymptotic distribution is normal with zero mean and variance matrix $\Xi_k$ that converges when $k$ tends to infinity to $\Xi$. More precisely we have
$$\frac{1}{\sqrt{n}}\sum_{t=1}^nU_t^k\xrightarrow[n\to\infty]{\text{in law}}  \ \mathcal{N}\left(0,\Xi_k\right).$$
The above arguments also apply to matrix $\Xi_k$ with some matrix $\Pi_k$ which is defined analogously as $\Pi$. Consequently we obtain
$$\frac{1}{\sqrt{n}}\sum_{t=1}^n\Lambda U_t^k\xrightarrow[n\to\infty]{\text{in law}}  \ \mathcal{N}\left(0,\Lambda\Xi_k\Lambda^{'}\right)$$
and we also have $n^{-1/2}\sum_{t=1}^n(\Pi_k\Pi_k^{'})^{-1/2}\Lambda U_t^k\xrightarrow[n\to\infty]{\text{in law}}  \ \mathcal{N}(0,I_m)$.

Now we are able to apply the functional central limit theorem (see \cite{herr}) and we obtain that
$$\frac{1}{\sqrt{n}}\sum_{t=1}^{\lfloor nr\rfloor}(\Pi_k\Pi_k^{'})^{-1/2}\Lambda U_t^k\overset{\mathbb{D}^m}{\underset{n\rightarrow\infty}{\mathbf{\longrightarrow}}}B_m(r).$$
Since for all $t\in\left\lbrace 1,\dots,\lfloor nr\rfloor \right\rbrace $ we write
$$(\Pi\Pi^{'})^{-1/2}\Lambda U_t^k=\left((\Pi\Pi^{'})^{-1/2}-(\Pi_k\Pi_k^{'})^{-1/2}\right)\Lambda U_t^k+(\Pi_k\Pi_k^{'})^{-1/2}\Lambda U_t^k,$$
we obtain the following weak convergence on $\mathbb{D}^m\left[ 0,1\right]$:
$$\frac{1}{\sqrt{n}}\sum_{t=1}^{\lfloor nr\rfloor}(\Pi\Pi^{'})^{-1/2}\Lambda U_t^k\overset{\mathbb{D}^m}{\underset{n\rightarrow\infty}{\mathbf{\longrightarrow}}}B_m(r).$$
%\noindent In fact, using the same approach as that followed in the proof of Lemma \ref{normalite_score}, we show that $n^{-1/2}\sum_{t=1}^{n}((\Pi\Pi^{'})^{-1/2}-(\Pi_k\Pi_k^{'})^{-1/2})\Lambda U_t^k$ converge in distribution to 0.\\

In order to conclude that \eqref{conv-Skorokhod-1} is true, it remains to observe that uniformly with respect to $n$
\begin{equation}\label{conv-skorokh-2}
Y_n^k(r):=\frac{1}{\sqrt{n}}\sum_{t=1}^{\lfloor nr\rfloor}(\Pi\Pi^{'})^{-1/2}\Lambda Z_t^k \ \overset{\mathbb{D}^m}{\underset{k\rightarrow\infty}{\mathbf{\longrightarrow}}}0,
\end{equation}
where
$$Z_t^k=\left( -2\left\lbrace \sum_{i=k+1}^{\infty}\overset{\textbf{.}}{\lambda}_{i,1}\left(\theta_0\right) \epsilon_t\epsilon_{t-i},\dots,\sum_{i=k+1}^{\infty}\overset{\textbf{.}}{\lambda}_{i,p+q+1}\left(\theta_0\right) \epsilon_t\epsilon_{t-i}\right\rbrace J^{-1 '},\epsilon_t\epsilon_{t-1},\dots,\epsilon_t\epsilon_{t-m} \right) ^{'}.$$
Using the same arguments as those used in the proof  of Theorem 2 in \cite{BMES2019}, we have
$$\sup_{n}\mathrm{Var}\left(\frac{1}{\sqrt{n}}\sum_{t=1}^{n}Z_t^k\right)\overset{}{\underset{k\rightarrow\infty}{\mathbf{\longrightarrow}}}0$$
and since $\lfloor nr\rfloor\leq n$,
$$\sup_{0\leq r\leq 1}\sup_n\left\lbrace \left\|Y_n^k(r)\right\|\right\rbrace \overset{}{\underset{k\rightarrow\infty}{\mathbf{\longrightarrow}}}0.$$
Thus \eqref{conv-skorokh-2} is true and the proof of \eqref{conv-Skorokhod-1} is achieved.
\subsubsection{Limit theorem}
To conclude the prove of Theorem~\ref{sn1}, we follow the arguments developed in \cite{BMS2018}.
Note that the previous step ensures us that  Assumption 1 in \cite{lobato} is satisfied for the sequence $(\Lambda U_t)_{t\ge 1}$.
Firstly from \eqref{conv-Skorokhod-1} we deduce that
\begin{align}\nonumber
\frac{1}{\sqrt{n}}S_{\lfloor nr\rfloor}&=\frac{1}{\sqrt{n}}\sum_{t=1}^{\lfloor nr\rfloor}\Lambda U_t-\frac{\lfloor nr\rfloor}{n}\left( \frac{1}{\sqrt{n}}\sum_{t=1}^{n}\Lambda U_t\right)
\\&
\overset{\mathbb{D}^m}{\underset{n\rightarrow\infty}{\mathbf{\longrightarrow}}}\ (\Pi\Pi^{'})^{1/2}B_m(r)-r(\Pi\Pi^{'})^{1/2}B_m(1).\label{conv-S}
\end{align}
%
%\begin{equation}\label{conv-S}
%\frac{1}{\sqrt{n}}S_{\lfloor nr\rfloor} \overset{\mathbb{D}^m}{\underset{n\rightarrow\infty}{\mathbf{\longrightarrow}}} \ (\Pi\Pi^{'})^{1/2}\left(B_m(r)-rB_m(1)\right).
%\end{equation}
%In fact, from \eqref{sqrt_gamma_chap} we have
%\begin{align*}
%\frac{1}{\sqrt{n}}S_{\lfloor nr\rfloor}&=\frac{1}{\sqrt{n}}\sum_{t=1}^{\lfloor nr\rfloor}\Lambda U_t-\frac{\lfloor nr\rfloor}{n}\left( \frac{1}{\sqrt{n}}\sum_{t=1}^{n}\Lambda U_t\right).
%\end{align*}
Observe now that  the continuous mapping theorem implies
\begin{align*}
C_m&=\frac{1}{n}\sum_{t=1}^n\left(\frac{1}{\sqrt{n}}S_t\right)\left(\frac{1}{\sqrt{n}}S_t\right)^{'}\\
&\overset{\mathbb{D}^m}{\underset{n\rightarrow\infty}{\mathbf{\longrightarrow}}} \ (\Pi\Pi^{'})^{1/2}\left[ \int_0^1\left\lbrace  B_m(r)-rB_m(1)\right\rbrace\left\lbrace B_m(r)-rB_m(1)\right\rbrace^{'}\mathrm{dr}\right] (\Pi\Pi^{'})^{1/2}
=(\Pi\Pi^{'})^{1/2}V_m(\Pi\Pi^{'})^{1/2}.
\end{align*}
%\noindent Since $\sqrt{n}\hat{\gamma}_m=n^{-1/2}\sum_{t=1}^n\Lambda U_t+\mathrm{o}_{\mathbb{P}}(1)$, we obtain
Using \eqref{sqrt_gamma_chap}, \eqref{conv-S} and again the continuous mapping theorem on the Skorokhod space, one finally obtains
\begin{align*}
n\hat{\gamma}_m^{'}C_m^{-1}\hat{\gamma}_m
&\overset{\mathbb{D}^m}{\underset{n\rightarrow\infty}{\mathbf{\longrightarrow}}} \left\lbrace (\Pi\Pi^{'})^{1/2}B_m(1)\right\rbrace ^{'}\left\lbrace (\Pi\Pi^{'})^{1/2}V_m(\Pi\Pi^{'})^{1/2}\right\rbrace^{-1} \left\lbrace (\Pi\Pi^{'})^{1/2}B_m(1)\right\rbrace
\\
&\hspace{2cm}
=B_m^{'}(1)V_m^{-1}B_m(1):=\mathcal{U}_m.
\end{align*}
Consequently, from \eqref{hat_rho} it follows that
$$n\sigma_{\epsilon}^4\hat{\rho}_m^{'}C_m^{-1}\hat{\rho}_m\overset{\mathbb{D}^m}{\underset{n\rightarrow \infty}{\mathbf{\longrightarrow}}}\mathcal{U}_m,$$
which completes the proof of Theorem~\ref{sn1}.$\hfill\square$
\subsection{Proof of Theorem~\ref{sn2}}\label{proof_sn2}
The proof follows the same line than in the proof of Theorem 2 in \cite{BMS2018} (see also the proof of in \cite{BMES2019}).

\section{Example of explicit calculation of $\Sigma_{\hat{\rho}_m}$ and $C_m$}\label{exple}
The results of the previous subsections \ref{diagnostic1} and \ref{diagnostic2} are particularized in the FARIMA$(1,d_0,0)$ and FARIMA$(0,d_0,1)$ cases. %with martingale difference innovations.
First we consider the case of a FARIMA$(1,d_0,0)$ model of the form
\begin{equation}\label{ex_FARIMA}
(1-L)^{d_0}\left(X_t-aX_{t-1}\right)=\epsilon_t,
\end{equation}
where the  unknown parameter is $\theta_0=(a,d_0)$. We assume that in \eqref{ex_FARIMA} the innovation process $(\epsilon_t)_{t\in\mathbb{Z}}$
is a GARCH$(1,1)$ process given by \eqref{garch}. 
%the model
%\begin{equation} \label{garch}
%\left\{\begin{array}{l}\epsilon_{t}=\sigma_t\eta_{t}\\
%\sigma_t^2=\omega+\alpha_1\epsilon_{t-1}^2+\beta_1\sigma_{t-1}^2,
%\end{array}\right.
%\end{equation}
%with $\omega>0$, $\alpha_1\ge0$ and where $(\eta_t)_{t\in\mathbb{Z}}$ is a sequence of iid centered Gaussian random variables with variance 1.
We also assume that in \eqref{garch}: $\alpha_1^2\kappa+\beta_1^2+2\alpha_1\beta_1<1$,\footnote{This  is a necessary and  sufficient condition for the existence of a nonanticipative stationary solution
 process $(\epsilon_t)_{t\in\mathbb{Z}}$ with fourth-order moments (see \cite[Example 2.3]{FZ2010}).} where $\kappa:=\mathbb{E}\eta_1^4$ and we assume that $\kappa>1$.

For the sake of simplicity we assume that the variables $(\eta_t)_{t\in\mathbb{Z}}$ involved in \eqref{garch} have a symmetric distribution. More precisely, we have the following symmetry
 assumption
\begin{equation}
\label{symetrie}
\mathbb{E}[\epsilon_{t_1}\epsilon_{t_2}\epsilon_{t_3}\epsilon_{t_4}]=0\qquad\text{
when }\qquad t_1\neq t_2,\; t_1\neq t_3\mbox{ and } t_1\neq t_4,
\end{equation}
made in \cite{FZ2009_jtsa,BMCF2012}.  For this particular GARCH$(1,1)$ model with
fourth-order moments and symmetric innovations satisfying  (\ref{symetrie}), it can be shown that
\begin{equation}
\label{symetrie2}
\mathbb{E}\left[ \epsilon_t\epsilon_{t-\ell}\epsilon_{t-h}\epsilon_{t-h-\ell'}\right]=\left\{\begin{array}{ll}
\mathbb{E}\left[ \epsilon_t^2\epsilon_{t-\ell}^2\right]&\qquad\text{if }h=0 \text{ and }\ell=\ell'\\ \vspace{0.1cm}\\
0& \ \ \ \ \ \text{ otherwise.}
\end{array}\right.
\end{equation}
Now we need to compute the autocovariance structure of $(\epsilon_t^2)_{t\in\mathbb{Z}}$. We will use the fact that the GARCH process $(\epsilon_t)_{t\in\mathbb{Z}}$  is fourth-order stationary, then $(\epsilon_t^2)_{t\in\mathbb{Z}}$ is a solution of the following ARMA$(1,1)$ model
\begin{equation}
\label{eqARMA}
\epsilon_t^{2}=\omega+(\alpha_1+\beta_1)\epsilon_
{t-1}^{2}+\nu_t-\beta_1\nu_{t-1},
 \quad t\in\mathbb{Z}
\end{equation}
where $\nu_t=\epsilon_t^{2}-\sigma_t^2$ is the innovation of  $(\epsilon_t^2)_{t\in\mathbb{Z}}$. From \eqref{eqARMA} the autocovariances of $(\epsilon_t^2)_{t\in\mathbb{Z}}$
take the form
\begin{equation}
\label{autocov}
\gamma_{\epsilon^2}(\ell):=\mathrm{Cov}(\epsilon_t^2,\epsilon_{t-\ell}^2)=\gamma_{\epsilon^2}(1)(\alpha_1+\beta_1)^{\ell-1},\qquad \ell\ge1,
\end{equation}
where
\begin{align*}
\gamma_{\epsilon^2}(1)&=\frac{(\kappa-1)(\alpha_1-\alpha_1\beta_1^2-\alpha_1^2\beta_1)}{1-\beta_1^2-2\alpha_1\beta_1-\alpha_1^2\kappa}\sigma_\epsilon^4,\\
\gamma_{\epsilon^2}(0)&:=\mathrm{Var}(\epsilon_t^2)=\frac{(\kappa-1)(1-\beta_1^2-2\alpha_1\beta_1)}{1-\beta_1^2-2\alpha_1\beta_1-\alpha_1^2\kappa}\sigma_\epsilon^4,\\
\text{and }\sigma_\epsilon^2&:=\frac{\omega}{1-\alpha_1-\beta_1}.
\end{align*}
From \eqref{symetrie2} and \eqref{autocov} we deduce that for any $\ell\ge1$
\begin{align}\nonumber
\Gamma(\ell,\ell)&=\mathbb{E}\left[ \epsilon_t^2\epsilon_{t-\ell}^2\right]=\mathrm{Cov}(\epsilon_t^2,\epsilon_{t-\ell}^2)+\mathbb{E}\left[ \epsilon_t^2\right]\mathbb{E}\left[ \epsilon_{t-\ell}^2\right]\\
&=\left\{1+\frac{1}{\sigma_\epsilon^4}\gamma_{\epsilon^2}(1)(\alpha_1+\beta_1)^{\ell-1}\right\}\sigma_\epsilon^4 .\label{expliGam}
\end{align}
\subsection{Examples of analytic and numerical computations of $\Sigma_{\hat{\rho}_m}$}
As mentioned before, the subject of this subsection is to give an explicit expression  of the asymptotic variance of residual autocorrelations $\Sigma_{\hat\rho_m}$ defined in \eqref{loi_hat_rho} in the particular case of model \eqref{ex_FARIMA}.
For that sake,  we need the following additional expressions.  It is classical that the noise derivatives $(\partial\epsilon_t(\theta_0)/\partial a ,\partial\epsilon_t(\theta_0)/\partial d)^{'}$ in \eqref{ex_FARIMA} can be represented as
\begin{equation}\label{noise_deriv}
\begin{pmatrix} \frac{\partial\epsilon_{t}(\theta_0)}{\partial a} \vspace{0.2cm}\\ \frac{\partial\epsilon_{t}(\theta_0)}{\partial d} \end{pmatrix}=-\sum_{j\geq 1}\begin{pmatrix} a^{j-1}\vspace{0.2cm}\\ \frac{1}{j}\end{pmatrix}\epsilon_{t-j}.
\end{equation}
We  compute the information matrices $J(\theta_0)$ and $I(\theta_0)$ by using \eqref{noise_deriv}. Then we have
\begin{align}\label{J_2}
J(\theta_0)&=
%2\mathbb{E}\left[\begin{pmatrix} \frac{\partial\epsilon_{t}(\theta_0)}{\partial a} \vspace{0.2cm}\\ \frac{\partial\epsilon_{t}(\theta_0)}{\partial d} \end{pmatrix}\left(  \frac{\partial\epsilon_{t}(\theta_0)}{\partial a} \ \ \ \frac{\partial\epsilon_{t}(\theta_0)}{\partial d}\right) \right]\nonumber\\
%&=2\sum_{i,j\geq 1}\begin{pmatrix}  a^{j-1}a^{i-1}\mathbb{E}\left[\epsilon_{t-j}\epsilon_{t-i}\right]  &  a^{j-1}\frac{1}{i}\mathbb{E}\left[\epsilon_{t-j}\epsilon_{t-i}\right]  \vspace{0.2cm}\\  \frac{1}{j}a^{i-1}\mathbb{E}\left[\epsilon_{t-j}\epsilon_{t-i}\right] & \frac{1}{ji}\mathbb{E}\left[\epsilon_{t-j}\epsilon_{t-i}\right]\end{pmatrix}\nonumber\\
%&=2\sigma_{\epsilon}^2\sum_{j\geq 1}\begin{pmatrix}  a^{2j-2} &  \frac{a^{j-1}}{j} \vspace{0.2cm}\\  \frac{a^{j-1}}{j} & \frac{1}{j^2}\end{pmatrix}\nonumber\\
%&=
2\sigma_{\epsilon}^2\begin{pmatrix}  \frac{1}{1-a^2} &  -\frac{\ln(1-a)}{a} \vspace{0.2cm}\\  -\frac{\ln(1-a)}{a} & \frac{\pi^2}{6}\end{pmatrix}.
\end{align}
A simple calculation implies that
\begin{equation}\label{invJ_2}
J^{-1}(\theta_0)=\frac{1}{2\sigma_{\epsilon}^2c(a)}\begin{pmatrix}  \frac{\pi^2}{6} &  \frac{\ln(1-a)}{a} \vspace{0.2cm}\\  \frac{\ln(1-a)}{a} & \frac{1}{1-a^2}\end{pmatrix},
\end{equation}
where
\begin{equation}\label{c(a)}
c(a)=\frac{\pi^2}{6(1-a^2)}-\left(\frac{\ln(1-a)}{a}\right)^2.
\end{equation}
We now investigate a similar tractable expression for $I(\theta_0)$. Using \eqref{noise_deriv} and \eqref{symetrie} we have
\begin{align}\label{I_2}
I(\theta_0)&=
%4\sum_{i\geq 1}\Gamma_{2,2}(i,i)\begin{pmatrix} a^{2i-2}   & \frac{a^{i-1}}{i}\vspace{0.2cm}\\ \frac{a^{i-1}}{i} & \frac{1}{i^2} \end{pmatrix}\\
%&=4\sum_{i\geq 1}\sigma_{\epsilon}^4\left\lbrace 1+\frac{\kappa-1}{1-\kappa\alpha^2}\alpha^i\right\rbrace\begin{pmatrix} a^{2i-2}   & \frac{a^{i-1}}{i}\vspace{0.2cm}\\ \frac{a^{i-1}}{i} & \frac{1}{i^2} \end{pmatrix}\\
%&=
2\sigma_{\epsilon}^2J(\theta_0)+4\sigma_{\epsilon}^4\frac{(\kappa-1)(\alpha_1-\alpha_1\beta_1^2-\alpha_1^2\beta_1)}{1-\beta_1^2-2\alpha_1\beta_1-\alpha_1^2\kappa}\begin{pmatrix} \frac{1}{1-a^2(\alpha_1+\beta_1)}   & -\frac{\ln[1-a(\alpha_1+\beta_1)]}{a(\alpha_1+\beta_1)}\vspace{0.2cm}\\ -\frac{\ln[1-a(\alpha_1+\beta_1) ]}{a(\alpha_1+\beta_1)} & \frac{\mathrm{Li}_2(\alpha_1+\beta_1)}{\alpha_1+\beta_1} \end{pmatrix},
\end{align}
where $\mathrm{Li}_2$ is the Spence function defined by $\mathrm{Li}_2(z)=\sum_{k=1}^\infty {z^k}{k^{-2}}$.  Note that we retrieve the well know result: $I(\theta_0)=2\sigma_{\epsilon}^2J(\theta_0)$ in the strong FARIMA case ({\em i.e.} when $\alpha_1=\beta_1=0$ in \eqref{I_2}).

The matrix defined in \eqref{Psi_m} can be rewritten as
\begin{align}\label{psi_2}
\Psi_m&=
%\mathbb{E}\left[ \begin{pmatrix} \epsilon_{t-1} \\ \epsilon_{t-2} \\ \epsilon_{t-3} \end{pmatrix}\left(  \frac{\partial\epsilon_{t}(\theta_0)}{\partial a} \ \ \ \frac{\partial\epsilon_{t}(\theta_0)}{\partial d}\right) \right] \nonumber\\
%&=-\sum_{j\geq 1}\begin{pmatrix}  a^{j-1}\mathbb{E}\left[\epsilon_{t-j}\epsilon_{t-1}\right]  &  \frac{1}{j}\mathbb{E}\left[\epsilon_{t-j}\epsilon_{t-1}\right]  \vspace{0.2cm}\\  a^{j-1}\mathbb{E}\left[\epsilon_{t-j}\epsilon_{t-2}\right] & \frac{1}{j}\mathbb{E}\left[\epsilon_{t-j}\epsilon_{t-2}\right] \vspace{0.2cm}\\
%a^{j-1}\mathbb{E}\left[\epsilon_{t-j}\epsilon_{t-3}\right]  &  \frac{1}{j}\mathbb{E}\left[\epsilon_{t-j}\epsilon_{t-3}\right]\end{pmatrix}\nonumber\\
%&=
-\sigma_{\epsilon}^2\begin{pmatrix} 1  &   a &\dots& a^{m-1} \\1  &  \frac{1}{2} &\dots& \frac{1}{m}
 \end{pmatrix}'.
\end{align}
Using \eqref{expliGam} and under the symmetry assumption \eqref{symetrie},  the matrix $\Gamma_{m,m}$  takes the simple following diagonal form
\begin{align}\label{explicit_Gam}
\Gamma_{m,m}&= \sigma_{\epsilon}^4I_m+ \sigma_{\epsilon}^4\frac{(\kappa-1)(\alpha_1-\alpha_1\beta_1^2-\alpha_1^2\beta_1)}{1-\beta_1^2-2\alpha_1\beta_1-\alpha_1^2\kappa}
\mathrm{diag}(1,(\alpha_1+\beta_1),\dots,(\alpha_1+\beta_1)^{m-1}).
\end{align}
Using \eqref{symetrie}, \eqref{noise_deriv} and \eqref{invJ_2}, the matrix $\Sigma'_{\hat{\theta},\gamma_m}$ is given  by

{\small
\begin{align}\label{Sig_gam_theta}
\Sigma'_{\hat{\theta},\gamma_m}=\frac{1}{\sigma_{\epsilon}^2c(a)} \begin{pmatrix}
 \left\lbrace \frac{\pi^2}{6}+\frac{\ln(1-a)}{a}\right\rbrace \Gamma_{m,m}(1,1)  &  \left\lbrace \frac{1}{1-a^2}+\frac{\ln(1-a)}{a}\right\rbrace \Gamma_{m,m}(1,1)  \\
\left\lbrace a\frac{\pi^2}{6}+\frac{\ln(1-a)}{2a}\right\rbrace\Gamma_{m,m}(2,2)& \left\lbrace \frac{1}{2(1-a^2)}+\ln(1-a)\right\rbrace \Gamma_{m,m}(2,2) \\
\vdots&\vdots\\ \left\lbrace a^{m-1}\frac{\pi^2}{6}+\frac{\ln(1-a)}{ma}\right\rbrace\Gamma_{m,m}(m,m)&
\left\lbrace \frac{1}{m(1-a^2)}+a^{m-2}\ln(1-a)\right\rbrace \Gamma_{m,m}(m,m)   \end{pmatrix},
\end{align}
}
where for any $1\leq i, j\leq m$, $\Gamma_{m,m}(i,j)$ is given by \eqref{explicit_Gam}.

 From Remark \ref{casfort}, in the strong FARIMA case the asymptotic variance of residual autocorrelations takes a simpler form
\begin{align*}
\Sigma_{\hat{\rho}_m}^{\textsc{s}}&=I_m-\frac{1}{c(a)}\left[\frac{\pi^2}{6}\left(a^{i+j-2}\right)+\frac{1}{1-a^2}\left(\frac{1}{ij}\right)+
\frac{\ln(1-a)}{a}\left(\frac{a^{j-1}}{i}+\frac{a^{i-1}}{j}\right)\right]_{1\le i,j\le m}
\end{align*}
where $c(a)$ is the constant given in \eqref{c(a)}.
% and where
%\begin{align*}
%M_1&=\begin{pmatrix} 1& 1/2 & \dots&1/m\\ 1/2 & a/2 & \dots&a/m \\ \vdots & \vdots & \dots&\vdots \\ 1/a & a/m & \dots&a^{m-1}/m\end{pmatrix},\\
%\text{and }\quad M_2&=\begin{pmatrix} 1& a & \dots&a^{m-1}\\ a & a/2 & \dots&a^{m-1} \\ \vdots & \vdots & \dots&\vdots \\ a^{m-1} & a^{m-1}/2 & \dots&a^{m-1}/m\end{pmatrix}.
%\end{align*}

From the above explicit expressions we deduce that the asymptotic variance of residual autocorrelations for this model is in the form
\begin{align*}
\Sigma_{\hat{\rho}_m}&=\Sigma_{\hat{\rho}_m}^{\textsc{s}}+
\frac{(\kappa-1)(\alpha_1-\alpha_1\beta_1^2-\alpha_1^2\beta_1)}{1-\beta_1^2-2\alpha_1\beta_1-\alpha_1^2\kappa}
\left[(\alpha_1+\beta_1)^{i-1}\1_{\{i=j\}}+\frac{1}{c(a)}M(i,j)\right.\\&- \left.\left\{(\alpha_1+\beta_1)^{i-1}+(\alpha_1+\beta_1)^{j-1}\right\}
\frac{1}{c(a)}\left\{\frac{\pi^2}{6}\left(a^{i+j-2}\right)+\frac{1}{1-a^2}\left(\frac{1}{ij}\right)+
\frac{\ln(1-a)}{a}\left(\frac{a^{j-1}}{i}+\frac{a^{i-1}}{j}\right)\right\}\right]_{1\le i,j\le m},
\end{align*}
where
\begin{align*}
M(i,j)&=\left[\frac{\ln(1-a)}{a}\frac{1}{1-a^2(\alpha_1+\beta_1)}-\frac{1}{1-a^2}\frac{\ln(1-a(\alpha_1+\beta_1))}{a(\alpha_1+\beta_1)}\right]\left[\frac{\pi^2}{6}\frac{a^{j-1}}{i}+\frac{1}{ij}\frac{\ln(1-a)}{a} \right]
\\&\quad+\left[\frac{\mathrm{Li}_2(\alpha_1+\beta_1)}{\alpha_1+\beta_1}\frac{1}{1-a^2}-\frac{\ln(1-a)}{a}\frac{\ln(1-a(\alpha_1+\beta_1))}{a(\alpha_1+\beta_1)}\right]\left[\frac{\ln(1-a)}{a}\frac{a^{j-1}}{i}+\frac{1}{ij}\frac{1}{1-a^2} \right]\\&\quad+
\left[\frac{\pi^2}{6}\frac{1}{1-a^2(\alpha_1+\beta_1)}-\frac{\ln(1-a)}{a}\frac{\ln(1-a(\alpha_1+\beta_1))}{a(\alpha_1+\beta_1)}\right]\left[\frac{\pi^2}{6}a^{i+j-2}+\frac{a^{i-1}}{j}\frac{\ln(1-a)}{a} \right]
\\&\quad+
\left[\frac{\mathrm{Li}_2(\alpha_1+\beta_1)}{\alpha_1+\beta_1}\frac{\ln(1-a)}{a}-\frac{\pi^2}{6}\frac{\ln(1-a(\alpha_1+\beta_1))}{a(\alpha_1+\beta_1)}\right]\left[\frac{\ln(1-a)}{a}a^{i+j-2}+\frac{a^{i-1}}{j}\frac{1}{1-a^2} \right].
\end{align*}
For simplicity, we take in the sequel $\beta_1=0$ to consider the case of an ARCH$(1)$ model.
For instance when $m=3$, $\kappa=3$, $\omega=1$ and  $a =-0.55$ we have
%\begin{table}[h]
\begin{center}
\scriptsize
\begin{tabular}{|c|c|c|c|}
\hline \hline
& $\Sigma_{\hat\rho_3}$ & Eigenvalues $\xi_{3}=(\xi_{1,3},\xi_{2,3},\xi_{3,3})$ & $ Z_3(\xi_{3})$\\
\cline{2-4}
$\alpha_1=0$ & $\begin{pmatrix} 0.1383 &0.0859& -0.2720 \\ 0.0859 &0.2490 & 0.0053 \\-0.2720 &0.0053 & 0.9135\end{pmatrix}$ & $(1.0000,0.2791,0.0217)$ & $\chi^2_1$+ 0.2791$\chi^2_1$+ 0.0217$\chi^2_1$\\%\rule[7pt]{0pt}{20pt}\\
\hline
$\alpha_1=0.55$ & $\begin{pmatrix} 0.6989 & 0.3825 &-1.6041 \\ 0.3825  &0.9351 &-0.2342 \\-1.6041 &-0.2342 & 4.7979\end{pmatrix}$ & $(5.3780 ,1.0025 ,0.0513)$ & 5.3780$\chi^2_1$+ 1.0025$\chi^2_1$+ 0.0513$\chi^2_1$\\%\rule[7pt]{0pt}{20pt}\\
\hline
\end{tabular}
\end{center}
%\end{table}
\normalsize

It is clear that for $\alpha_1=0.55$, the \cite{LiMcL1986} approximation by a $\chi^2_1$ distribution will be disastrous. The eigenvalues $\xi_{m}$ can be very different from those of strong FARIMA models which are close to 1 or 0 when the lag $m$  is large enough (see Remark \ref{remKhi2}). More precisely, for instance for $\alpha_1=0$ and $m=12$ we obtain
\begin{align*}
\xi_{12}&=(1.0000, 1.0000, 1.0000, 1.0000, 1.0000, 1.0000, 1.0000, 1.0000, 1.0000, 1.0000,0.0665, 0.0000)',
\end{align*}
In this weak FARIMA$(1,d,0)$ with  $\alpha_1=0.55$ and $m=12$ we also obtain
\begin{align*}
\xi_{12}&=(5.4628, 3.7524, 2.3222, 1.7930, 1.4152, 1.2405, 1.1295, 1.0723, 1.0387, 1.0207, 0.0827, 0.0000)'.
\end{align*}
The same result holds for FARIMA$(0,d,1)$ model with $a$ replaced by $b$ in $\theta_0$.

\subsection{Explicit form of the matrix $C_m$}
The following example gives an explicit form of the normalization matrix $C_m$ for the model given in \eqref{ex_FARIMA}. For reading convenience, we restrict ourselves to the case $m=3$.
Using the expression of $J^{-1}(\theta_0)$ given in \eqref{invJ_2} and Equation \eqref{noise_deriv}, we obtain that for all $1\le j\le n$
\begin{align*}
-2J^{-1}(\theta_0)\epsilon_j\begin{pmatrix} \frac{\partial\epsilon_{j}(\theta_0)}{\partial a} \vspace{0.2cm}\\ \frac{\partial\epsilon_{j}(\theta_0)}{\partial d} \end{pmatrix}&=\begin{pmatrix} v^{(1)}_j(a) \vspace{0.2cm}\\ v^{(2)}_j(a) \end{pmatrix},
\end{align*}
where
$$v^{(1)}_j(a)=\frac{1}{\sigma_{\epsilon}^2c(a)}\sum_{k\geq 1}\left\lbrace \frac{\pi^2}{6}a^{k-1}+\frac{\ln(1-a)}{a}\frac{1}{k}\right\rbrace\epsilon_j\epsilon_{j-k}$$
and
$$v^{(2)}_j(a)=\frac{1}{\sigma_{\epsilon}^2c(a)}\sum_{k\geq 1}\left\lbrace \frac{\ln(1-a)}{a}a^{k-1}+\frac{1}{1-a^2}\frac{1}{k}\right\rbrace\epsilon_j\epsilon_{j-k}.$$
Thus, the vector $\Lambda U_j$ is given by
\begin{align*}
\Lambda U_j&=
%\begin{pmatrix} -\sigma_{\epsilon}^2 & -\sigma_{\epsilon}^2 & 1 & 0 & 0 \vspace{0.2cm}\\ -\sigma_{\epsilon}^2a & -\sigma_{\epsilon}^2/2 & 0 & 1 & 0\vspace{0.2cm}\\ -\sigma_{\epsilon}^2a^2 & -\sigma_{\epsilon}^2/3 & 0 & 0 & 1 \end{pmatrix}\begin{pmatrix} v^{(1)}_j(a) \vspace{0.2cm}\\ v^{(2)}_j(a) \vspace{0.2cm}\\ \epsilon_j\epsilon_{j-1} \vspace{0.2cm}\\ \epsilon_j\epsilon_{j-2} \vspace{0.2cm}\\ \epsilon_j\epsilon_{j-3}\end{pmatrix}\\
%&=
\begin{pmatrix} -\sigma_{\epsilon}^2v^{(1)}_j(a)-\sigma_{\epsilon}^2v^{(2)}_j(a)+\epsilon_j\epsilon_{j-1}\vspace{0.2cm}\\ -\sigma_{\epsilon}^2av^{(1)}_j(a)-\sigma_{\epsilon}^2v^{(2)}_j(a)/2+\epsilon_j\epsilon_{j-2}\vspace{0.2cm}\\ -\sigma_{\epsilon}^2a^2v^{(1)}_j(a)-\sigma_{\epsilon}^2v^{(2)}_j(a)/3+\epsilon_j\epsilon_{j-3} \end{pmatrix}.
\end{align*}
A simple calculation shows that, for any $1\le j_1,j_2\le n$,
\begin{align*}
\left(\Lambda U_{j_1}\right)\left( \Lambda U_{j_2}\right)^{'}&=\begin{pmatrix} K^{(1)}_{j_1}(a)K^{(1)}_{j_2}(a) \hspace{0.2cm}& K^{(1)}_{j_1}(a)K^{(2)}_{j_2}(a) \hspace{0.2cm}& K^{(1)}_{j_1}(a)K^{(3)}_{j_2}(a) \vspace{0.4cm}\\ K^{(2)}_{j_1}(a)K^{(1)}_{j_2}(a) & K^{(2)}_{j_1}(a)K^{(2)}_{j_2}(a) & K^{(2)}_{j_1}(a)K^{(3)}_{j_2}(a) \vspace{0.4cm}\\ K^{(3)}_{j_1}(a)K^{(1)}_{j_2}(a) & K^{(3)}_{j_1}(a)K^{(2)}_{j_2}(a) & K^{(3)}_{j_1}(a)K^{(3)}_{j_2}(a) \end{pmatrix},
\end{align*}
where
\begin{align*}
K^{(1)}_{j}(a)&=-\sigma_{\epsilon}^2v^{(1)}_j(a)-\sigma_{\epsilon}^2v^{(2)}_j(a)+\epsilon_j\epsilon_{j-1},\\
K^{(2)}_{j}(a)&=-\sigma_{\epsilon}^2av^{(1)}_j(a)-\sigma_{\epsilon}^2v^{(2)}_j(a)/2+\epsilon_j\epsilon_{j-2}\\
\text{and }\quad
K^{(3)}_{j}(a)&=-\sigma_{\epsilon}^2a^2v^{(1)}_j(a)-\sigma_{\epsilon}^2v^{(2)}_j(a)/3+\epsilon_j\epsilon_{j-3}.
\end{align*}
Therefore we deduce that for all positive integer $t$

\begin{align*}
S_t&=\sum_{j=1}^t\left(\Lambda U_j-\gamma_3\right)=\sum_{j=1}^t\begin{pmatrix} -\sigma_{\epsilon}^2v^{(1)}_j(a)-\sigma_{\epsilon}^2v^{(2)}_j(a)+\epsilon_j\epsilon_{j-1}\vspace{0.2cm}\\ -\sigma_{\epsilon}^2av^{(1)}_j(a)-\sigma_{\epsilon}^2v^{(2)}_j(a)/2+\epsilon_j\epsilon_{j-2}\vspace{0.2cm}\\ -\sigma_{\epsilon}^2a^2v^{(1)}_j(a)-\sigma_{\epsilon}^2v^{(2)}_j(a)/3+\epsilon_j\epsilon_{j-3} \end{pmatrix}-\frac{t}{n}\begin{pmatrix} \sum_{j=2}^n\epsilon_j\epsilon_{j-1}\vspace{0.2cm}\\ \sum_{j=3}^n\epsilon_j\epsilon_{j-2}\vspace{0.2cm}\\ \sum_{j=4}^n\epsilon_j\epsilon_{j-3} \end{pmatrix}.
\end{align*}
The same result holds for FARIMA$(0,d_0,1)$ model with $a$ replaced by $b$ in $\theta_0$.

\clearpage
\setcounter{page}{1}
\begin{center}
{\Large
Diagnostic checking in FARIMA models with uncorrelated but non-independent error terms: {\bf
Complementary results that are not submitted for publication}}
\end{center}

\section{Supplemental material: Additional Monte Carlo experiments}
For the nominal level $\alpha=5\%$, the empirical  size over the $N$ independent replications should
vary between the significant limits 3.6\% and 6.4\% with probability
95\%.
When the relative rejection frequencies  are outside the 95\%
significant limits, they are displayed in bold type in Tables. 
% \ref{tab1FARIMA}, \ref{tab2FARIMA}, \ref{tab3FARIMA}, \ref{tabFARIMA00f}, \ref{tabFARIMA00garch} and \ref{tabFARIMA00pt}.

\subsection{FARIMA models with $a\neq0$ and $b\neq0$}
Table \ref{tab1FARIMA} displays the relative rejection frequencies of the null hypothesis {\bf (H0)} that the  DGP follows a
strong FARIMA model \eqref{process-sim}, over the $N$ independent replications.
When $p=q=1$  for all tests, the percentages of rejection belong to the confident interval with probabilities 95\%, except for $\mathrm{{LB}}_\textsc{s}$ and ${\mathrm{BP}}_\textsc{s}$ (see Table  \ref{tab1FARIMA}).
%when $m\le 6$ (resp. $m\le 2$).
 Consequently all these tests well control the error of first kind.

We draw the conclusion that in these strong FARIMA cases the proposed modified version may be clearly preferable to the standard ones.

Now, we repeat the same experiments on two weak FARIMA models.
As expected Tables \ref{tab2FARIMA} and  \ref{tab3FARIMA} show that the standard  $\mathrm{{LB}}_\textsc{s}$ or $\mathrm{{BP}}_\textsc{s}$ test poorly performs in assessing the adequacy of these particular weak FARIMA models.
Indeed, we observe that
\begin{itemize}
\item the observed relative rejection frequencies of $\mathrm{{LB}}_\textsc{s}$ and $\mathrm{{BP}}_\textsc{s}$ are definitely outside the significant limits,
\item the errors of the first kind are only globally well controlled by the proposed tests  when $n$ is large.
\end{itemize}
We also investigate the case where the GARCH model \eqref{garch} have infinite fourth moments. As showing in Figures~\ref{fig1},\dots,\ref{fig6} the results are qualitatively similar to what we observe here
in Tables \ref{tab2FARIMA} and \ref{tab3FARIMA}.

\begin{table}[h]
 \caption{\small{Empirical size (in \%) of the modified and standard versions
 of the LB and BP tests in the case of a strong FARIMA$(1,d_0,1)$  defined by \eqref{process-sim}
 with $\theta_0=(0.9,0.2,d_0)$.
The nominal asymptotic level of the tests is $\alpha=5\%$.
The number of replications is $N=1,000$. }}
\begin{center}
%\begin{tabular}{lll rrr rrr}
{\scriptsize
\begin{tabular}{c c ccc ccc c}
\hline\hline \\
$d_0$& Length $n$ & Lag $m$ & $\mathrm{{LB}}_{\textsc{sn}}$&$\mathrm{BP}_{\textsc{sn}}$&$\mathrm{{LB}}_{\textsc{w}}$&$\mathrm{BP}_{\textsc{w}}$&$\mathrm{{LB}}_{\textsc{s}}$&$\mathrm{BP}_{\textsc{s}}$
\vspace*{0.2cm}\\\hline
&& $1$&5.8& 5.7 &\textbf{7.4} &\textbf{7.3}&n.a. &n.a.\\
&& $2$&5.0 &5.0 &\textbf{7.4}&\textbf{ 7.3}&n.a. &n.a.\\
0.05 &$n=1,000$& $3$&4.3   &4.3   &5.8  & 5.8&n.a. &n.a.\\
 &&$6$&4.1 & 4.1 & 5.6 & 5.5 &\textbf{10.9}& \textbf{10.9}\\
 &&$12$&5.1 &4.6 &4.7& 4.5 &\textbf{6.9} &\textbf{6.6}\\
 &&$15$&5.0 &4.7 &5.0 &4.8 &\textbf{6.9} &5.9\\
  \cline{2-9}
&& $1$&6.0 &6.0 &\textbf{7.4} &\textbf{7.4}&n.a. &n.a.\\
&& $2$&\textbf{6.5} &\textbf{6.5} &\textbf{7.9 }&\textbf{7.9}&n.a. &n.a.\\
0.05 &$n=5,000$& $3$&4.7  & 4.7  & \textbf{6.7}  & \textbf{6.7}&n.a. &n.a.\\
 &&$6$&\textbf{3.5}  &\textbf{3.5} & 5.2 & 5.1& \textbf{11.0} &\textbf{10.9}\\
 &&$12$&5.3& 5.3 &5.8 &5.8& \textbf{7.9}& \textbf{7.6}\\
 &&$15$&4.5& 4.5& 5.8 &5.5 &\textbf{7.0} &\textbf{6.9}\\
  \cline{2-9}
&& $1$&4.2 &4.2 &6.1 &6.1&n.a. &n.a.\\
&& $2$&4.2 &4.2 &6.3& 6.4&n.a. &n.a.\\
0.05 &$n=10,000$& $3$&3.8  & 3.8  & 5.9  & 5.9&n.a. &n.a.\\
 &&$6$&\textbf{3.5} & \textbf{3.5} & 4.7 & 4.7 &\textbf{10.4} &\textbf{10.4}\\
 &&$12$&4.2 &4.2 &6.1 &6.1 &\textbf{7.6}& \textbf{7.6}\\
 &&$15$&4.0 &3.8 &5.7 &5.7& \textbf{7.4} &\textbf{7.4}\\

\hline
&& $1$&5.8& 5.8 &\textbf{9.2}& \textbf{9.1} &n.a. &n.a.\\
&& $2$&4.9 &4.9 &\textbf{7.5}& \textbf{7.5}&n.a. &n.a.\\
0.20 &$n=1,000$& $3$&4.6  & 4.5 &  5.9  & 5.9&n.a. &n.a.\\
 &&$6$&4.2 & 4.1  &5.6 & 5.4 &\textbf{10.3}& \textbf{10.2}\\
 &&$12$&5.4 &4.9 &4.7 &4.4& 6.4 &5.9\\
 &&$15$&5.5& 4.9 &5.1& 4.4 &6.8 &6.2\\
  \cline{2-9}
&& $1$&6.4 &6.4 &6.1& 6.2&n.a. &n.a.\\
&& $2$&\textbf{6.8} &\textbf{6.8} &\textbf{6.9} &\textbf{6.9}&n.a. &n.a.\\
0.20 &$n=5,000$& $3$&4.3   &4.3  & 5.9  & 5.8&n.a. &n.a.\\
 &&$6$&3.8  &3.8  &4.6  &4.6 &\textbf{10.0} &\textbf{10.0}\\
 &&$12$&5.2 &5.2& 5.7 &5.6 &\textbf{7.6} &\textbf{7.5}\\
 &&$15$&4.5 &4.5& 5.6 &5.3& \textbf{6.8}& \textbf{6.7}\\
  \cline{2-9}
&& $1$&4.5 &4.5 &5.5 &5.5&n.a. &n.a.\\
&& $2$&4.1 &4.1 &5.8& 5.8&n.a. &n.a.\\
0.20 &$n=10,000$& $3$&\textbf{3.1}  & \textbf{3.1}  & 5.3  & 5.3&n.a. &n.a.\\
 &&$6$&3.7 & 3.6 & 4.3 & 4.3 &\textbf{10.1}& \textbf{10.1}\\
 &&$12$&3.8 &3.8 &6.1 &6.1&\textbf{ 7.5} &\textbf{7.5}\\
 &&$15$&3.7 &3.7 &5.8& 5.7 &\textbf{7.0}& \textbf{6.9}\\

\hline
&& $1$&4.3 &4.3&\textbf{ 8.7} &\textbf{8.7} &n.a. &n.a.\\
&& $2$&\textbf{3.0} &\textbf{3.0}& 5.9& 5.9&n.a. &n.a.\\
0.45 &$n=1,000$& $3$&3.7  & 3.7  & 4.4  & 4.4&n.a. &n.a.\\
 &&$6$&3.8 &3.8& 4.7 &4.5 &\textbf{8.1}& \textbf{7.8}\\
 &&$12$&5.1 &4.6& 4.3& 4.2& 5.1 &4.9\\
 &&$15$&4.6 &4.5 &4.7 &4.3& 5.0& 4.7\\
  \cline{2-9}
&& $1$& 5.6 &5.5 &6.0& 6.0&n.a. &n.a.\\
&& $2$&5.2 &5.2 &6.4 &6.4&n.a. &n.a.\\
0.45 &$n=5,000$& $3$&4.0   &4.0  & 5.9  & 5.9&n.a. &n.a.\\
 &&$6$&3.8 & 3.8 & 4.6 & 4.6 &\textbf{10.1} & \textbf{9.9}\\
 &&$12$&5.2& 5.2 &5.4& 5.4& \textbf{7.2}& \textbf{7.1}\\
 &&$15$&4.6 &4.6& 5.0 &4.9& \textbf{6.7} &\textbf{6.6}\\
  \cline{2-9}
&& $1$&4.3 &4.3 &5.3 &5.3&n.a. &n.a.\\
&& $2$&3.2& 3.2 &5.7& 5.7&n.a. &n.a.\\
0.45 &$n=10,000$& $3$&\textbf{3.1}   &\textbf{3.0}  & 5.4   &5.4&n.a. &n.a.\\
 &&$6$&3.7 &3.7 &4.3 &4.3 &\textbf{9.8} &\textbf{9.8}\\
 &&$12$&4.3 &4.3& 5.8 &5.8 &\textbf{7.2}& \textbf{7.0}\\
 &&$15$&3.6 &3.3 &5.7& 5.7 &\textbf{6.8} &\textbf{6.8}\\
\hline\hline
\\

\end{tabular}
}
\end{center}
\label{tab1FARIMA}
\end{table}

\begin{table}[h]
 \caption{\small{Empirical size (in \%) of the modified and standard versions
 of the LB and BP tests in the case of a weak FARIMA$(1,d_0,1)$  defined by \eqref{process-sim}
 with $\theta_0=(0.9,0.2,d_0)$ and where $\omega=0.4$, $\alpha_1=0.3$
and $\beta_1=0.3$ in \eqref{garch}.
The nominal asymptotic level of the tests is $\alpha=5\%$.
The number of replications is $N=1,000$. }}
\begin{center}
%\begin{tabular}{lll rrr rrr}
{\scriptsize
\begin{tabular}{c c ccc ccc c}
\hline\hline \\
$d_0$& Length $n$ & Lag $m$ & $\mathrm{{LB}}_{\textsc{sn}}$&$\mathrm{BP}_{\textsc{sn}}$&$\mathrm{{LB}}_{\textsc{w}}$&$\mathrm{BP}_{\textsc{w}}$&$\mathrm{{LB}}_{\textsc{s}}$&$\mathrm{BP}_{\textsc{s}}$
\vspace*{0.2cm}\\\hline
&& $1$&4.9 &4.9 &\textbf{6.7} &\textbf{6.7}&n.a. &n.a.\\
&& $2$&3.8& 3.8& 6.3 &6.3&n.a. &n.a.\\
0.05 &$n=1,000$& $3$&\textbf{3.2}  & \textbf{3.2}  & 5.2  & 5.2&n.a. &n.a.\\
 &&$6$&3.9 & 3.8 & 4.9  &4.8& \textbf{18.5}& \textbf{18.3}\\
 &&$12$&\textbf{2.3}&  \textbf{2.3} & 4.1  &4.0& \textbf{10.2} & \textbf{9.7}\\
 &&$15$&\textbf{2.7}& \textbf{2.3} &4.4 &4.2 &\textbf{9.7} &\textbf{9.3}\\
  \cline{2-9}
&& $1$&5.1 &5.1& 5.6 &5.6&n.a. &n.a.\\
&& $2$&4.9 &4.9 &5.4 &5.4&n.a. &n.a.\\
0.05 &$n=5,000$& $3$&\textbf{2.6}  & \textbf{2.6}  & 5.0  & 5.0&n.a. &n.a.\\
 &&$6$&\textbf{3.5}  &\textbf{3.5}  &4.4  &4.4 &\textbf{19.6}& \textbf{19.6}\\
 &&$12$&\textbf{2.7} & \textbf{2.7}  &\textbf{3.3}  &\textbf{3.2}& \textbf{11.4} &\textbf{11.4}\\
 &&$15$&\textbf{3.4} & \textbf{3.4} & 4.2  &4.1 &\textbf{10.8} &\textbf{10.7}\\
  \cline{2-9}
&& $1$&4.8 &4.8 &\textbf{6.9} &\textbf{6.9}&n.a. &n.a.\\
&& $2$&4.8 &4.8 &\textbf{6.7}& \textbf{6.7}&n.a. &n.a.\\
0.05 &$n=10,000$& $3$&4.7  & 4.7  & 5.5 &  5.5&n.a. &n.a.\\
 &&$6$&\textbf{3.3} &\textbf{ 3.3}  &6.4 & 6.4 &\textbf{20.2}& \textbf{20.2}\\
 &&$12$&4.2  &4.2 & 6.3 & 6.3 &\textbf{12.4} &\textbf{12.3}\\
 &&$15$&3.6 & 3.6 & 5.5 & 5.5& \textbf{11.6}& \textbf{11.6}\\

\hline
&& $1$&5.3 &5.3 &\textbf{7.8} &\textbf{7.7}  &n.a. &n.a.\\
&& $2$&3.6 &\textbf{3.4}& 5.7& 5.7&n.a. &n.a.\\
0.20 &$n=1,000$& $3$&\textbf{3.1 } & \textbf{3.1} &  4.9  & 4.8&n.a. &n.a.\\
 &&$6$&\textbf{3.3} & \textbf{3.2} & 4.5  &4.5& \textbf{17.6} &\textbf{17.4}\\
 &&$12$&\textbf{2.3}& \textbf{2.0}& 4.1 &4.1& \textbf{9.4} &\textbf{8.9}\\
 &&$15$&\textbf{2.4} &\textbf{2.1}& 4.4& 4.2 &\textbf{9.0} &\textbf{8.1}\\
  \cline{2-9}
&& $1$&4.6 &4.6 &4.3& 4.3&n.a. &n.a.\\
&& $2$&4.3& 4.3 &4.4 &4.4&n.a. &n.a.\\
0.20 &$n=5,000$& $3$&\textbf{3.1}  & \textbf{3.1} &  4.4  & 4.3&n.a. &n.a.\\
 &&$6$&4.1 & 4.1  &3.9 & 3.9 &\textbf{19.0} &\textbf{19.0}\\
 &&$12$&\textbf{2.6}&  \textbf{2.6} & \textbf{2.9 }& \textbf{2.9}& \textbf{10.9} &\textbf{10.6}\\
 &&$15$&\textbf{3.4} &\textbf{ 3.3}  &4.0  &4.0 &\textbf{10.0} & \textbf{9.9}
\\
  \cline{2-9}
&& $1$&4.8 &4.8 &5.1 &5.1&n.a. &n.a.\\
&& $2$&4.7 &4.7& 5.0& 5.0&n.a. &n.a.\\
0.20 &$n=10,000$& $3$&4.5   &4.5  & 4.8 &  4.8&n.a. &n.a.\\
 &&$6$&\textbf{3.5}  &\textbf{3.5} & 5.6  &5.6 &\textbf{19.1} &\textbf{19.1}\\
 &&$12$&4.1  &4.1  &5.9 & 5.9 &\textbf{12.1} &\textbf{12.1}\\
 &&$15$&3.7  &3.7 & 5.3 & 5.3 &\textbf{11.3} &\textbf{11.3}\\

\hline
&& $1$&4.4  &4.4 &\textbf{11.1} &\textbf{11.0 }&n.a. &n.a.\\
&& $2$&\textbf{3.4 }&\textbf{3.4 }&5.4& 5.3&n.a. &n.a.\\
0.45 &$n=1,000$& $3$&\textbf{3.1 } & \textbf{3.1}  & 4.9  & 4.9&n.a. &n.a.\\
 &&$6$&\textbf{3.1} & \textbf{2.9} & 4.5 & 4.4 &\textbf{15.3} &\textbf{15.1}\\
 &&$12$&\textbf{2.2 }&\textbf{2.1} &4.0 &4.0 &\textbf{7.9 }&\textbf{7.5}\\
 &&$15$&\textbf{2.1} &\textbf{2.0}& 4.4 &4.3&\textbf{ 7.0} &\textbf{6.5}\\
  \cline{2-9}
&& $1$& 3.9& 3.9& 4.2 &4.2&n.a. &n.a.\\
&& $2$&\textbf{3.4}& \textbf{3.4} &4.2 &4.2&n.a. &n.a.\\
0.45 &$n=5,000$& $3$&\textbf{2.9 }  &\textbf{2.9}  & 4.4  & 4.4&n.a. &n.a.\\
 &&$6$&\textbf{3.5} & \textbf{3.5}  &3.9 & 3.9 &\textbf{18.4}& \textbf{18.4}\\
 &&$12$&\textbf{2.4}& \textbf{2.4} &\textbf{2.8} &\textbf{2.7} &\textbf{9.9} &\textbf{9.8}\\
 &&$15$&\textbf{3.2} &\textbf{3.2} &3.9 &3.8 &\textbf{9.2}& \textbf{9.2}\\
  \cline{2-9}
&& $1$&4.6 &4.6 &5.3 &5.3&n.a. &n.a.\\
&& $2$&4.3 &4.3 &5.1& 5.0&n.a. &n.a.\\
0.45 &$n=10,000$& $3$&\textbf{3.5}   &\textbf{3.5}  & 5.0  & 5.0&n.a. &n.a.\\
 &&$6$&\textbf{2.8 }& \textbf{2.8} & 5.3  &5.3 &\textbf{19.3}& \textbf{19.3}\\
 &&$12$&4.2 & 4.2 & 5.5 & 5.5 &\textbf{12.2}& \textbf{12.2}\\
 &&$15$&3.6 & \textbf{3.5}  &5.5  &5.5& \textbf{11.4}& \textbf{11.4}\\
\hline\hline
\\

\end{tabular}
}
\end{center}
\label{tab2FARIMA}
\end{table}

\begin{table}[h]
 \caption{\small{Empirical size (in \%) of the modified and standard versions
 of the LB and BP tests in the case of a weak FARIMA$(1,d_0,1)$  defined by \eqref{process-sim}--\eqref{noise-sim}  with $\theta_0=(0.9,0.2,d_0)$.
The nominal asymptotic level of the tests is $\alpha=5\%$.
The number of replications is $N=1,000$. }}
\begin{center}
%\begin{tabular}{lll rrr rrr}
{\scriptsize
\begin{tabular}{c c ccc ccc c}
\hline\hline \\
$d_0$& Length $n$ & Lag $m$ & $\mathrm{{LB}}_{\textsc{sn}}$&$\mathrm{BP}_{\textsc{sn}}$&$\mathrm{{LB}}_{\textsc{w}}$&$\mathrm{BP}_{\textsc{w}}$&$\mathrm{{LB}}_{\textsc{s}}$&$\mathrm{BP}_{\textsc{s}}$
\vspace*{0.2cm}\\\hline
&& $1$&5.1& 5.1 &\textbf{7.3}& \textbf{7.3}&n.a. &n.a.\\
&& $2$&3.6 &3.6 &\textbf{6.9 }&\textbf{6.9}&n.a. &n.a.\\
0.05 &$n=1,000$& $3$&\textbf{2.9}  &\textbf{ 2.9}  & 4.3  & 4.1&n.a. &n.a.\\
 &&$6$&\textbf{2.6} & \textbf{2.5} & \textbf{3.1}  &\textbf{3.0} &\textbf{10.3} &\textbf{10.3}\\
 &&$12$&\textbf{0.9} &\textbf{0.9} &\textbf{1.2 }&\textbf{1.1} &\textbf{8.7} &\textbf{8.3}\\
 &&$15$&\textbf{0.4} &\textbf{0.4} &\textbf{1.0} &\textbf{0.8} &\textbf{8.0}& \textbf{7.3}\\
  \cline{2-9}
&& $1$&3.9& 3.9 &5.4& 5.4&n.a. &n.a.\\
&& $2$&3.9 &3.9& 5.9& 5.9&n.a. &n.a.\\
0.05 &$n=5,000$& $3$&3.9  & 3.9 &  5.5  & 5.5&n.a. &n.a.\\
 &&$6$&\textbf{3.2}  &\textbf{3.1} & 3.8 & 3.8& \textbf{10.6} &\textbf{10.6}
\\
 &&$12$&\textbf{2.4 }&\textbf{2.4} &\textbf{3.5} &\textbf{3.4} &\textbf{8.3}& \textbf{8.2}\\
 &&$15$&\textbf{2.7}& \textbf{2.7} &\textbf{3.3} &\textbf{3.3} &\textbf{8.4} &\textbf{8.3}\\
  \cline{2-9}
&& $1$&5.0 &5.0& 5.2& 5.2&n.a. &n.a.\\
&& $2$&4.9 &4.9 &4.5& 4.5&n.a. &n.a.\\
0.05 &$n=10,000$& $3$&3.8  & 3.8   &5.6  & 5.6&n.a. &n.a.\\
 &&$6$&3.6 & 3.6 & 4.5 & 4.5& \textbf{10.4}& \textbf{10.4}\\
 &&$12$&\textbf{3.3} &\textbf{3.3} &4.3 &4.3 &\textbf{8.5}& \textbf{8.4}\\
 &&$15$&4.7& 4.7 &3.8 &3.8 &\textbf{7.7} &\textbf{7.4}\\

\hline
&& $1$&5.7  &5.6 &\textbf{10.1} &\textbf{10.0} &n.a. &n.a.\\
&& $2$&\textbf{3.4} &\textbf{3.4} &5.5 &5.5&n.a. &n.a.\\
0.20 &$n=1,000$& $3$&3.7   &3.7  & 4.0  & 4.0&n.a. &n.a.\\
 &&$6$&\textbf{2.9}  &\textbf{2.8} & \textbf{2.5 } &\textbf{2.4}& \textbf{10.2} & \textbf{9.7}\\
 &&$12$&\textbf{0.9 }&\textbf{0.9} &\textbf{1.1} &\textbf{1.1} &\textbf{7.9} &\textbf{7.2}\\
 &&$15$&\textbf{0.5}& \textbf{0.5} &\textbf{0.8} &\textbf{0.8}& \textbf{7.5}& \textbf{6.9}\\
  \cline{2-9}
&& $1$&\textbf{3.5}& \textbf{3.5}& 4.0 &3.9&n.a. &n.a.\\
&& $2$&3.7 &3.7 &4.3 &4.3&n.a. &n.a.\\
0.20 &$n=5,000$& $3$&4.1  & 4.1  & 5.0 &  5.0&n.a. &n.a.\\
 &&$6$&\textbf{3.1 }&\textbf{ 3.1} & \textbf{3.5}  &\textbf{3.5} &\textbf{10.0} &\textbf{10.0}\\
 &&$12$&\textbf{2.8} &\textbf{2.8} &\textbf{3.3} &\textbf{3.3}& \textbf{8.2} &\textbf{8.2}\\
 &&$15$&\textbf{2.4} &\textbf{2.4} &\textbf{3.1} &\textbf{3.1} &\textbf{7.9}& \textbf{7.8}\\
  \cline{2-9}
&& $1$&5.1 &5.1& 4.8& 4.8&n.a. &n.a.\\
&& $2$&4.7 &4.7& 4.2 &4.2&n.a. &n.a.\\
0.20 &$n=10,000$& $3$&3.8  & 3.8  & 4.7  & 4.7&n.a. &n.a.\\
 &&$6$&3.8 & 3.8 & 4.1&  4.1& \textbf{10.1}& \textbf{10.1}\\
 &&$12$&\textbf{3.4}& \textbf{3.4} &4.0 &4.0& \textbf{8.0}& \textbf{8.0}\\
 &&$15$&4.8& 4.8 &3.6 &3.6 &\textbf{7.5} &\textbf{7.4}\\

\hline
&& $1$&3.8  &3.8 &\textbf{12.1} &\textbf{12.0} &n.a. &n.a.\\
&& $2$&\textbf{2.4} &\textbf{2.4 }&4.4& 4.4 &n.a. &n.a.\\
0.45 &$n=1,000$& $3$& \textbf{2.7 } & \textbf{2.6}  & 3.8  & 3.7&n.a. &n.a.\\
 &&$6$&\textbf{3.2 }&\textbf{3.0} &\textbf{2.3 }&\textbf{2.3} &\textbf{8.3} &\textbf{7.9}
\\
 &&$12$&\textbf{1.1} &\textbf{0.9 }&\textbf{1.0 }&\textbf{0.9}& 6.4 &6.3\\
 &&$15$&\textbf{0.3} &\textbf{0.3} &\textbf{1.4}& \textbf{1.1} &\textbf{6.8} &6.4\\
  \cline{2-9}
&& $1$& \textbf{3.1} &\textbf{3.1 }&4.4 &4.4&n.a. &n.a.\\
&& $2$&\textbf{2.7 }&\textbf{2.7} &4.5 &4.5&n.a. &n.a.\\
0.45 &$n=5,000$& $3$&\textbf{3.2}   &\textbf{3.2}   &4.9  & 4.9&n.a. &n.a.\\
 &&$6$&\textbf{3.2}& \textbf{3.1} &\textbf{3.4} &\textbf{3.4} &\textbf{9.7}& \textbf{9.7}\\
 &&$12$&\textbf{3.3}& \textbf{3.3} &\textbf{3.3} &\textbf{3.3} &\textbf{7.3}& \textbf{7.3}\\
 &&$15$&\textbf{2.4} &\textbf{2.4} &\textbf{3.2} &\textbf{3.1}& \textbf{7.2} &\textbf{7.0}
\\
  \cline{2-9}
&& $1$&5.1 &5.1& 4.8 &4.8&n.a. &n.a.\\
&& $2$&4.9 &4.9 &4.3 &4.3&n.a. &n.a.\\
0.45 &$n=10,000$& $3$&3.6 &  3.6  & 4.9  & 4.9&n.a. &n.a.\\
 &&$6$&\textbf{3.5}  &\textbf{3.5} & 4.3  &4.2 &\textbf{10.2} &\textbf{10.2}
\\
 &&$12$&3.7 &3.7 &3.7 &3.7& \textbf{7.7}& \textbf{7.6}\\
 &&$15$&4.8 &4.8 &3.9 &3.9& \textbf{7.2}& \textbf{7.1}\\
\hline\hline
\\

\end{tabular}
}
\end{center}
\label{tab3FARIMA}
\end{table}

In this section, we repeat the same experiments as in Section \ref{simul} to examine the power of the tests for the null hypothesis of Model \eqref{process-sim} against the following FARIMA alternative defined by
\begin{equation}\label{power-sim}
(1-L)^d\left(X_t-aX_{t-1}\right)=\epsilon_t-b_1\epsilon_{t-1}-b_2\epsilon_{t-2},
\end{equation}
with  $\theta_0=(a,b_1,b_2,d_0)$ and where the innovation process  $(\epsilon_t)_{t\in\mathbb{Z}}$ follows a strong or weak white noise introduced in Section \ref{simul}.

For each of these $N$ replications we fit a  FARIMA$(1,d,1)$ model \eqref{process-sim}
and perform standard and modified tests based on $m=1,2,3$, $6$, $12$ and $15$ residual autocorrelations.

Tables \ref{ptab1FARIMA}, \ref{ptab2FARIMA} and \ref{ptab3FARIMA} compare the empirical powers of Model \eqref{power-sim}
with  $\theta_0=(0.9,1,-0.2,d_0)$ over the $N$ independent replications.
For these particular strong and weak FARIMA models,  we notice that the standard $\mathrm{{BP}}_{\textsc{s}}$ and
$\mathrm{{{LB}}_{\textsc{s}}}$ and our proposed  tests have very similar powers except for $\mathrm{{BP}}_{\textsc{sn}}$ and $\mathrm{{{LB}}_{\textsc{sn}}}$ when  $n=5,000$. % in the weak case.

%%%%%%%%%%%%%%%%%%%%%%%%%%%%%%%%%%%%%%%
\begin{table}[h]
 \caption{\small{Empirical power (in \%) of the modified and standard versions of the LB and BP tests in the case of
 a strong FARIMA$(1,d_0,2)$  defined by \eqref{power-sim}
 with $\theta_0=(0.9,1,-0.2,d_0)$.
The nominal asymptotic level of the tests is $\alpha=5\%$.
The number of replications is $N=1,000$. }}
\begin{center}
%\begin{tabular}{lll rrr rrr}
{\scriptsize
\begin{tabular}{c c ccc ccc c}
\hline\hline \\
$d_0$& Length $n$ & Lag $m$ & $\mathrm{{LB}}_{\textsc{sn}}$&$\mathrm{BP}_{\textsc{sn}}$&$\mathrm{{LB}}_{\textsc{w}}$&$\mathrm{BP}_{\textsc{w}}$&$\mathrm{{LB}}_{\textsc{s}}$&$\mathrm{BP}_{\textsc{s}}$
\vspace*{0.2cm}\\\hline
%&& $1$&7.1  &7.1 &15.0 &15.0&n.a. &n.a.\\
%&& $2$&7.4 & 7.4 &15.4 &15.3&n.a. &n.a.\\
%0.05 &$n=1,000$& $3$&9.7  & 9.5 &  8.8 &  8.8&n.a. &n.a.\\
% &&$6$&10.9 &10.6 & 9.7 & 8.8 &13.0& 12.7\\
% &&$12$&11.2 &10.8  &6.9  &6.4 & 8.1 & 7.9\\
% &&$15$&11.8 &10.6 & 4.4 & 3.9 & 7.3 & 7.1\\
  \cline{2-9}
&& $1$&24.5& 24.5 &37.9& 37.9&n.a. &n.a.\\
&& $2$&28.8 &28.8& 46.1 &46.1 &n.a. &n.a.\\
0.05 &$n=5,000$& $3$&36.7 & 36.7 & 22.2 & 22.1&n.a. &n.a.\\
 &&$6$&55.7& 55.7 &40.6 &40.3 &47.6 &47.6\\
 &&$12$&54.9& 54.7& 27.2 &27.2& 28.3 &28.0\\
 &&$15$&54.0& 53.6& 18.0& 17.8 &27.9& 27.7\\
  \cline{2-9}
&& $1$&44.9 &44.9 &62.8 &62.7&n.a. &n.a.\\
&& $2$&51.4& 51.3 &76.1 &76.0&n.a. &n.a.\\
0.05 &$n=10,000$& $3$&62.8 & 62.8 & 39.9 & 39.9&n.a. &n.a.\\
 &&$6$&86.5 &86.5 &80.9& 80.8 &84.7& 84.7\\
 &&$12$&85.8 &85.8 &64.9 &64.8 &66.4& 66.2\\
 &&$15$&82.0 &82.0 &43.2 &43.2 &60.8& 60.8\\

\hline
%&& $1$&4.7  &4.7 &26.2 &26.2 &n.a. &n.a.\\
%&& $2$&7.4 & 7.3& 30.3 &30.2&n.a. &n.a.\\
%0.20 &$n=1,000$& $3$&9.2   &9.1  &15.7  &15.7&n.a. &n.a.\\
% &&$6$&9.3 & 9.2 &23.2 &22.4 &27.2 &26.8\\
% &&$12$&11.5 &11.0 &13.0 &12.4 &18.2 &17.9\\
% &&$15$&10.8 &10.5 & 3.5 & 3.3 &16.0 &15.7\\
  \cline{2-9}
&& $1$&14.0 &14.0 &58.0& 57.9&n.a. &n.a.\\
&& $2$&22.2 &22.2 &71.1 &71.1&n.a. &n.a.\\
0.20 &$n=5,000$& $3$&24.1 & 23.8  &40.7 & 40.7&n.a. &n.a.\\
 &&$6$&32.1 &32.0& 74.4 &74.4 &78.5& 78.5\\
 &&$12$&52.3 &52.2 &62.4& 62.2 &67.7 &67.6\\
 &&$15$&51.6& 51.3 &14.1 &14.0 &62.1 &61.7\\
  \cline{2-9}
&& $1$&21.4 &21.4& 84.9 &85.0&n.a. &n.a.\\
&& $2$&30.6 &30.6 &93.1 &93.1&n.a. &n.a.\\
0.20 &$n=10,000$& $3$& 35.6&  35.6 & 65.9 & 65.7&n.a. &n.a.\\
 &&$6$&44.1 &44.1 &96.9 &96.9 &97.8 &97.8\\
 &&$12$&76.3 &76.2 &93.2 &93.2 &94.3 &94.3\\
 &&$15$& 73.7 &73.7 &43.9 &43.9 &91.6& 91.6\\

\hline
%&& $1$&0.2 & 0.2 &67.0 &66.8 &n.a. &n.a.\\
%&& $2$&9.3 & 9.0 &94.9 &94.9&n.a. &n.a.\\
%0.45 &$n=1,000$& $3$&18.4 & 18.3 & 99.7 & 99.7&n.a. &n.a.\\
% &&$6$&21.6 &21.3& 99.6& 99.6 &99.7 &99.7\\
% &&$12$&22.2 &21.2 &99.4 &99.4 &99.4 &99.4\\
% &&$15$&23.4 &21.8 &99.2 &99.2& 99.3 &99.3
%\\
  \cline{2-9}
&& $1$& 0.0&  0.& 100.0 &100.0&n.a. &n.a.\\
&& $2$&49.1  &49.1 &100.0 &100.0&n.a. &n.a.\\
0.45 &$n=5,000$& $3$&69.0&  69.0& 100.0& 100.0&n.a. &n.a.\\
 &&$6$&76.7&  76.6 &100.0 &100.0 &100.0 &100.0\\
 &&$12$&86.8 & 86.7& 100.0 &100.0& 100.0& 100.0\\
 &&$15$&90.9 & 90.7 &100.0 &100.0 &100.0 &100.0\\
  \cline{2-9}
&& $1$&0.0 &0.0 &100.0 &100.0&n.a. &n.a.\\
&& $2$&77.9 & 77.9 &100.0 &100.0&n.a. &n.a.\\
0.45 &$n=10,000$& $3$&90.3 & 90.2& 100.0& 100.0&n.a. &n.a.\\
 &&$6$&94.2 & 94.2 &100.0 &100.0 &100.0 &100.0\\
 &&$12$&98.9 & 98.9 &100.0 &100.0& 100.0& 100.0\\
 &&$15$&99.5 & 99.4& 100.0 &100.0& 100.0 &100.0\\
\hline\hline
\\

\end{tabular}
}
\end{center}
\label{ptab1FARIMA}
\end{table}

\begin{table}[h]
 \caption{\small{Empirical power (in \%) of the modified and standard versions of the LB and BP tests in the case of
 a weak FARIMA$(1,d_0,2)$  defined by \eqref{power-sim}
 with $\theta_0=(0.9,1,-0.2,d_0)$ and where $\omega=0.4$, $\alpha_1=0.3$
and $\beta_1=0.3$ in \eqref{garch}.
The nominal asymptotic level of the tests is $\alpha=5\%$.
The number of replications is $N=1,000$. }}
\begin{center}
%\begin{tabular}{lll rrr rrr}
{\scriptsize
\begin{tabular}{c c ccc ccc c}
\hline\hline \\
$d_0$& Length $n$ & Lag $m$ & $\mathrm{{LB}}_{\textsc{sn}}$&$\mathrm{BP}_{\textsc{sn}}$&$\mathrm{{LB}}_{\textsc{w}}$&$\mathrm{BP}_{\textsc{w}}$&$\mathrm{{LB}}_{\textsc{s}}$&$\mathrm{BP}_{\textsc{s}}$
\vspace*{0.2cm}\\\hline
%&& $1$&6.2 & 6.2 &13.1& 13.3&n.a. &n.a.\\
%&& $2$&6.9  &6.9 &14.3 &14.0&n.a. &n.a.\\
%0.05 &$n=1,000$& $3$&9.0  & 9.0  & 7.7  & 7.6&n.a. &n.a.\\
% &&$6$&10.1 & 9.9 & 7.4  &7.3 &18.3& 18.2\\
% &&$12$&8.3 &8.2& 5.4& 5.1 &9.1& 8.6\\
% &&$15$&6.9 &6.0 &4.6& 4.2 &9.8 &9.1\\
  \cline{2-9}
&& $1$&22.5 &22.5 &32.8 &32.7&n.a. &n.a.\\
&& $2$&27.3 &27.3& 41.7 &41.8&n.a. &n.a.\\
0.05 &$n=5,000$& $3$&32.4 & 32.3  &20.1 & 20.0&n.a. &n.a.\\
 &&$6$&52.1 &52.0 &34.0 &34.0& 55.8 &55.7\\
 &&$12$&54.1& 54.1& 23.5 &23.5& 34.2 &34.1\\
 &&$15$&53.9& 53.4& 17.1 &16.9& 31.9& 31.8\\
  \cline{2-9}
&& $1$&36.1& 36.1 &53.2 &53.2&n.a. &n.a.\\
&& $2$&44.9 &44.9& 64.5 &64.5&n.a. &n.a.\\
0.05 &$n=10,000$& $3$&56.5 & 56.5  &33.1  &33.1&n.a. &n.a.\\
 &&$6$&83.1 &83.1& 71.2 &71.2& 86.4 &86.2\\
 &&$12$&84.0 &83.9 &59.0 &59.0 &70.4 &70.2\\
 &&$15$&80.6 &80.5 &40.1& 40.1 &67.4& 67.2\\

\hline
%&& $1$&4.8 & 4.8 &25.1& 24.6  &n.a. &n.a.\\
%&& $2$&8.1  &7.8 &25.9& 25.8&n.a. &n.a.\\
%0.20 &$n=1,000$& $3$&8.1 &  8.1 & 14.9  &14.5 &n.a. &n.a.\\
% &&$6$&8.9 & 8.7 &19.6 &19.3 &32.8 &32.3\\
% &&$12$&8.5 & 7.9& 11.9& 11.7 &20.4 &19.8\\
% &&$15$&6.8  &5.7  &4.4  &4.2 &17.9 &17.8
%\\
  \cline{2-9}
&& $1$&14.6 &14.5 &51.0 &50.9&n.a. &n.a.\\
&& $2$& 21.8 &21.8& 67.1 &67.1&n.a. &n.a.\\
0.20 &$n=5,000$& $3$&22.4 & 22.3 & 37.7&  37.7&n.a. &n.a.\\
 &&$6$&32.3& 32.3 &68.3 &68.3 &81.9& 81.9
\\
 &&$12$&51.6 &51.5& 55.9 &55.8 &68.7 &68.5\\
 &&$15$&51.7 &51.6 &64.2 &64.1 &64.8 &64.6
\\
  \cline{2-9}
&& $1$&22.8 &22.8 &74.1 &74.0&n.a. &n.a.\\
&& $2$&29.6 &29.6 &86.2 &86.2&n.a. &n.a.\\
0.20 &$n=10,000$& $3$&32.9  &32.9  &56.6  &56.5&n.a. &n.a.\\
 &&$6$&43.1& 43.1& 92.3 &92.3 &97.1 &97.1\\
 &&$12$&72.9 &72.8& 88.3 &88.3 &93.8 &93.8\\
 &&$15$&71.2 &71.1 &89.1 &88.9 &92.0 &92.0\\

\hline
%&& $1$&0.5 & 0.5 &66.1 &66.0&n.a. &n.a.\\
%&& $2$&9.0 & 9.0& 92.4 &92.4&n.a. &n.a.\\
%0.45 &$n=1,000$& $3$&13.6 & 13.5 & 98.1 & 98.1&n.a. &n.a.\\
% &&$6$&18.3& 18.3& 99.2 &99.2 &99.6 &99.6\\
% &&$12$&16.8 &16.2 &97.8& 97.8 &99.3 &99.3
%\\
% &&$15$&15.6 &14.5 &97.6 &97.4& 99.2 &99.2\\
  \cline{2-9}
&& $1$& 30.1 & 30.1 &99.8 &99.8&n.a. &n.a.\\
&& $2$&40.1 & 40.1 &100.0 &100.0&n.a. &n.a.\\
0.45 &$n=5,000$& $3$&57.9&  57.9& 100.0 &100.0&n.a. &n.a.\\
 &&$6$&65.7  &65.7 &100.0 &100.0 &100.0 &100.0\\
 &&$12$&78.8 & 78.5& 100.0 &100.0 &100.0 &100.0\\
 &&$15$&84.7  &84.6 &100.0 &100.0& 100.0 &100.0\\
  \cline{2-9}
&& $1$&62.2  &62.2 &99.9 &99.9&n.a. &n.a.\\
&& $2$&72.2  &72.2 &100.0&  99.9 &n.a. &n.a.\\
0.45 &$n=10,000$& $3$&84.8  &84.8& 100.0 &100.0&n.a. &n.a.\\
 &&$6$&89.8 & 89.7 &100.0 &100.0 &100.0 &100.0\\
 &&$12$&97.7  &97.7 &100.0& 100.0 &100.0 &100.0\\
 &&$15$&99.0& 99.0 &100.0 &100.0 &100.0 &100.0\\
\hline\hline
\\

\end{tabular}
}
\end{center}
\label{ptab2FARIMA}
\end{table}

\begin{table}[h]
 \caption{\small{Empirical power (in \%) of the modified and standard versions of the LB and BP tests in the case of
 a weak FARIMA$(1,d_0,2)$  defined by \eqref{power-sim}--\eqref{noise-sim}
 with $\theta_0=(0.9,1,-0.2,d_0)$.
The nominal asymptotic level of the tests is $\alpha=5\%$.
The number of replications is $N=1,000$. }}
\begin{center}
%\begin{tabular}{lll rrr rrr}
{\scriptsize
\begin{tabular}{c c ccc ccc c}
\hline\hline \\
$d_0$& Length $n$ & Lag $m$ & $\mathrm{{LB}}_{\textsc{sn}}$&$\mathrm{BP}_{\textsc{sn}}$&$\mathrm{{LB}}_{\textsc{w}}$&$\mathrm{BP}_{\textsc{w}}$&$\mathrm{{LB}}_{\textsc{s}}$&$\mathrm{BP}_{\textsc{s}}$
\vspace*{0.2cm}\\\hline
%&& $1$&10.1& 10.1 &18.2 &18.2&n.a. &n.a.\\
%&& $2$&8.1 & 8.1 &16.8 &16.8&n.a. &n.a.\\
%0.05 &$n=1,000$& $3$&8.6  & 8.6 & 10.1  & 9.9&n.a. &n.a.\\
% &&$6$&9.6  &9.2&  8.0  &7.8 &11.4 &11.3\\
% &&$12$&4.7& 4.4 &3.9 &3.7 &9.1 &8.8\\
% &&$15$&2.6& 2.4 &1.7 &1.6& 8.8 &8.5\\
  \cline{2-9}
&& $1$&27.6& 27.6 &42.6& 42.7&n.a. &n.a.\\
&& $2$&32.7 &32.6& 51.4 &51.3&n.a. &n.a.\\
0.05 &$n=5,000$& $3$&36.9 & 36.9 & 23.7 & 23.7&n.a. &n.a.\\
 &&$6$&53.3 &53.0& 39.7 &39.7 &46.0 &45.9
\\
 &&$12$&49.6 &49.3 &23.7 &23.7 &29.3& 29.2\\
 &&$15$&44.4& 44.2 &17.5 &17.4 &28.5 &28.1\\
  \cline{2-9}
&& $1$&48.5 &48.5& 68.3& 68.3&n.a. &n.a.\\
&& $2$&58.7 &58.6 &76.6 &76.5&n.a. &n.a.\\
0.05 &$n=10,000$& $3$&66.8  &66.8 & 42.5 & 42.5&n.a. &n.a.\\
 &&$6$&84.2 &84.0& 77.0 &76.9 &83.2& 83.2\\
 &&$12$&79.9 &79.9 &62.7 &62.6 &66.0& 66.0\\
 &&$15$&75.8 &75.8 &40.5 &40.5 &61.4 &61.3\\

\hline
%&& $1$&5.1 & 5.1& 30.1 &30.3 &n.a. &n.a.\\
%&& $2$&8.0 & 8.0 &33.8 &33.7&n.a. &n.a.\\
%0.20 &$n=1,000$& $3$&7.9  & 7.9  &18.1 & 18.1&n.a. &n.a.\\
% &&$6$&7.4 & 7.2 &23.4 &22.9 &25.4 &25.3\\
% &&$12$&4.7 & 4.4&  9.5 & 9.0 &17.8 &17.3\\
% &&$15$& 2.9  &2.5&  2.5 & 2.3 &16.3& 15.6\\
  \cline{2-9}
&& $1$&15.3& 15.3 &62.4 &62.5&n.a. &n.a.\\
&& $2$&23.5 &23.4& 74.6& 74.6&n.a. &n.a.\\
0.20 &$n=5,000$& $3$&25.9 & 25.9 & 45.3 & 45.2&n.a. &n.a.\\
 &&$6$&34.0 &34.0 &73.1 &72.9& 78.5 &78.4\\
 &&$12$&51.3& 50.8 &56.8 &56.6 &64.5& 64.4\\
 &&$15$&46.3 &45.8 &15.0& 14.9 &60.1 &60.1\\
  \cline{2-9}
&& $1$&23.0& 23.0 &85.2 &85.2&n.a. &n.a.\\
&& $2$&33.8 &33.8 &93.6& 93.6&n.a. &n.a.\\
0.20 &$n=10,000$& $3$&36.5 & 36.5 & 68.3 & 68.3&n.a. &n.a.\\
 &&$6$&46.8& 46.7& 95.4 &95.4 &97.1& 97.1\\
 &&$12$&81.7 &81.7 &90.8 &90.8 &93.7 &93.6
\\
 &&$15$&79.0 &78.7 &44.2 &44.0 &91.7 &91.7
\\

\hline
%&& $1$&0.3  &0.3& 65.2 &65.3 &n.a. &n.a.\\
%&& $2$&9.4 & 9.2 &90.2 &90.2&n.a. &n.a.\\
%0.45 &$n=1,000$& $3$& 15.6  &15.6 & 95.1 & 95.1&n.a. &n.a.\\
% &&$6$&16.5& 16.0 &94.8 &94.8& 96.4& 96.4
%\\
% &&$12$&9.7 & 9.2 &94.7& 94.7 &96.4 &96.4\\
% &&$15$&12.5 &12.0 &93.0 &93.0 &96.0 &96.0\\
  \cline{2-9}
&& $1$& 41.9 & 41.9 &99.9 &99.9&n.a. &n.a.\\
&& $2$&51.9 & 51.9 &100.0& 100.0&n.a. &n.a.\\
0.45 &$n=5,000$& $3$&66.7 & 66.7 &100.0 &100.0&n.a. &n.a.\\
 &&$6$&73.6 & 73.6 &100.0 &100.0 &100.0& 100.0\\
 &&$12$&83.1 & 83.0 &100.0 &100.0 &100.0 &100.0\\
 &&$15$&85.5 & 85.4 &100.0 &100.0 &100.0 &100.0
\\
  \cline{2-9}
&& $1$&69.2&  69.2 &100.0  &99.9&n.a. &n.a.\\
&& $2$&79.2&  79.2 &100.0& 100.0&n.a. &n.a.\\
0.45 &$n=10,000$& $3$&90.8&  90.8 &100.0 &100.0 &n.a. &n.a.\\
 &&$6$&93.6  &93.6 &100.0 &100.0 &100.0 &100.0
\\
 &&$12$&97.8 & 97.8& 100.0& 100.0 &100.0 &100.0\\
 &&$15$&99.1  &99.1 &100.0 &100.0 &100.0 &100.0\\
\hline\hline
\\

\end{tabular}
}
\end{center}
\label{ptab3FARIMA}
\end{table}

\newpage
\clearpage
\subsection{Small sample size}
The following tables deal with the same numerical experiments that in Section \ref{num-ill} when the sample sizes are less than 500.

%%%%%%%%%%%%%%%%%%%%%%%%%%%%%%%%%%%%%%%%%%%%%%%%%%%
\begin{table}[h]
 \caption{\small{Empirical size (in \%) of the modified and standard versions
 of the LB and BP tests in the case of a strong FARIMA$(0,d_0,0)$
  defined by \eqref{process-sim}  with $\theta_0=(0,0,d_0)$.
The nominal asymptotic level of the tests is $\alpha=5\%$.
The number of replications is $N=1,000$. }}
\begin{center}
%\begin{tabular}{lll rrr rrr}
{\scriptsize
\begin{tabular}{c c ccc ccc c}
\hline\hline \\
$d_0$& Length $n$ & Lag $m$ & $\mathrm{{LB}}_{\textsc{sn}}$&$\mathrm{BP}_{\textsc{sn}}$&$\mathrm{{LB}}_{\textsc{w}}$&$\mathrm{BP}_{\textsc{w}}$&$\mathrm{{LB}}_{\textsc{s}}$&$\mathrm{BP}_{\textsc{s}}$
\vspace*{0.2cm}\\\hline
&& $1$&3.9 & 3.6 &\textbf{10.1} & \textbf{9.6}&n.a. &n.a.\\
&& $2$&\textbf{3.3} &\textbf{3.2} &\textbf{8.1}& \textbf{7.4}& \textbf{7.6}& \textbf{7.1}\\
0.05 &$n=100$& $3$&3.8  & \textbf{3.1} &  5.9 &  5.2&\textbf{8.1}& \textbf{6.8}\\
 &&$6$&\textbf{3.1} & \textbf{2.7}  &5.0 & 3.9 &\textbf{6.9} &5.9\\
 &&$12$&\textbf{2.4}&  \textbf{1.3} & 3.9 & \textbf{2.1}& 5.8 &3.8\\
 &&$15$&\textbf{2.8}  &\textbf{1.0 }& 4.5  &\textbf{2.3} &\textbf{6.9}& 4.3\\
  \cline{2-9}
&& $1$&5.3 &5.2 &\textbf{7.6}& \textbf{7.3}&n.a. &n.a.\\
&& $2$&5.0& 4.7 &5.4 &5.3&6.1 &6.0\\
0.05 &$n=250$& $3$&4.7  & 4.5  & 5.6 &  5.5&5.8 &5.6\\
 &&$6$&5.2 & 4.8  &6.4  &6.1 &6.7 &6.3\\
 &&$12$&5.0 & 3.8 & 4.4 & 3.7& 6.2& 5.3\\
 &&$15$&4.6 & \textbf{3.2}  &4.4 & \textbf{3.5}& 6.0& 4.9\\
  \cline{2-9}
&& $1$&5.0& 5.0 &5.6& 5.6&n.a. &n.a.\\
&& $2$&5.5& 5.5 &5.7& 5.6&6.0 &5.8\\
0.05 &$n=500$& $3$&5.9  & 5.7  & 5.9  & 5.7&\textbf{6.6} &\textbf{6.5}\\
 &&$6$&5.3 & 5.1 & 5.6&  5.2 &6.0& 5.9\\
 &&$12$&5.1  &4.3 & 5.0  &4.7& 5.9 &5.0\\
 &&$15$&5.4  &4.5 & 4.6  &4.2 &6.0 &5.2\\

\hline
&& $1$&4.5& 4.0 &5.9 &5.3 &n.a. &n.a.\\
&& $2$&4.1& 3.7 &6.5& 6.0&\textbf{6.5} &5.8\\
0.20 &$n=100$& $3$&4.1 &  \textbf{3.5} &  5.3 &  4.9&6.4& 6.1\\
 &&$6$&3.3 &\textbf{2.9} &4.6 &3.7& 6.1 &4.9\\
 &&$12$&3.6 &\textbf{1.5} &4.1& \textbf{2.0}& 5.5 &\textbf{3.4}\\
 &&$15$&\textbf{2.9} &\textbf{0.9}& 4.4&\textbf{ 2.0} &6.5& \textbf{3.5}\\
  \cline{2-9}
&& $1$&5.8  & 5.7  & 5.8  & 5.7&n.a. &n.a.\\
&& $2$&5.2& 5.1& 5.2& 4.8 &5.8 &5.6\\
0.20 &$n=250$& $3$&5.1& 5.0 &5.5& 5.4 &5.4 &5.1\\
 &&$6$&5.7& 5.4 &5.9 &5.3& 6.3 &5.7\\
 &&$12$&5.6 &4.0 &4.2& 3.8& 5.8 &5.1\\
 &&$15$&4.8 &3.6 &4.5& 3.6 &6.2 &4.7\\
  \cline{2-9}
&& $1$& 5.7 &  5.5  & 5.0  & 5.0&n.a. &n.a.\\
&& $2$&5.4 &5.4& 5.4& 5.3 &5.5 &5.3\\
0.20 &$n=500$& $3$&6.2 &6.1 &5.7 &5.6 &6.3 &6.2\\
 &&$6$&5.4& 5.0& 5.5 &5.0 &5.6 &5.6\\
 &&$12$&5.1 &4.4 &5.0 &4.7 &6.0& 5.0\\
 &&$15$&5.2 &4.3 &4.4 &4.2 &5.9 &5.1\\

\hline
&& $1$&4.3  & 4.1  & \textbf{9.4}  & \textbf{8.9} &n.a. &n.a.\\
&& $2$&3.9 &3.4 &\textbf{8.3}& \textbf{7.5}& \textbf{7.7} &\textbf{7.3}\\
0.45 &$n=100$& $3$&4.0& \textbf{3.3} &\textbf{6.5}& 5.7 &\textbf{7.0} &\textbf{6.5}\\
 &&$6$&\textbf{3.3} &\textbf{2.4}& 4.7 &\textbf{3.5} &\textbf{6.5}& 5.3\\
 &&$12$&\textbf{3.5}& \textbf{1.7}& \textbf{3.9}& \textbf{2.3} &5.5 &\textbf{3.2}\\
 &&$15$&3.9 &\textbf{1.4} &4.2& \textbf{2.2} &6.1 &3.7\\
  \cline{2-9}
&& $1$&5.4  & 5.4   &\textbf{8.2}  &\textbf{ 7.9}&n.a. &n.a.\\
&& $2$&5.0 &4.9 &5.3 &5.1 &5.5 &5.3\\
0.45 &$n=250$& $3$&5.1 &5.0 &5.8& 5.3 &5.3 &5.0\\
 &&$6$&5.6& 5.2 &6.0 &5.2& 6.2 &5.4\\
 &&$12$&5.4 &3.9& 4.6 &3.9& 5.8& 5.2\\
 &&$15$&5.1 &4.0 &4.7& 3.7 &6.2 &5.0\\
  \cline{2-9}
&& $1$&5.4 &  5.2  & 5.6 &  5.6&n.a. &n.a.\\
&& $2$&5.2 &5.2 &5.4 &5.3 &5.9& 5.8\\
0.45 &$n=500$& $3$&5.9& 5.8 &6.3& 6.1& 6.4& 6.4\\
 &&$6$&6.0 &5.6 &5.6 &5.0 &5.6& 5.5\\
 &&$12$&4.9& 3.9 &5.6& 4.8 &5.7& 5.1\\
 &&$15$&5.2 &4.3 &4.6 &4.2 &6.1& 4.9\\
\hline\hline
\\

\end{tabular}
}
\end{center}
\label{tab1dn-petit0}
\end{table}

\begin{table}[h]
 \caption{\small{Empirical size (in \%) of the modified and standard versions
 of the LB and BP tests in the case of a weak FARIMA$(0,d_0,0)$
  defined by \eqref{process-sim}  with $\theta_0=(0,0,d_0)$ with $\omega=0.4$, $\alpha_1=0.3$
  and $\beta_1=0.3$ in \eqref{noise-sim}.
The nominal asymptotic level of the tests is $\alpha=5\%$.
The number of replications is $N=1,000$. }}
\begin{center}
%\begin{tabular}{lll rrr rrr}
{\scriptsize
\begin{tabular}{c c ccc ccc c}
\hline\hline \\
$d_0$& Length $n$ & Lag $m$ & $\mathrm{{LB}}_{\textsc{sn}}$&$\mathrm{BP}_{\textsc{sn}}$&$\mathrm{{LB}}_{\textsc{w}}$&$\mathrm{BP}_{\textsc{w}}$&$\mathrm{{LB}}_{\textsc{s}}$&$\mathrm{BP}_{\textsc{s}}$
\vspace*{0.2cm}\\\hline
&& $1$&\textbf{2.3}  & \textbf{2.3} & \textbf{10.1}  & \textbf{9.6}&n.a. &n.a.\\
&& $2$& \textbf{2.6} & \textbf{2.6}  &5.9 & 5.3 &\textbf{13.1} &\textbf{12.4}\\
0.05 &$n=100$& $3$&\textbf{1.9} &\textbf{ 1.6 }& 4.0  &\textbf{3.1}& \textbf{11.1} & \textbf{9.9}\\
 &&$6$&\textbf{1.4}  &\textbf{1.1}  &\textbf{3.0} & \textbf{2.5} &\textbf{12.8} &\textbf{11.2}\\
 &&$12$&\textbf{1.0} & \textbf{0.3} & \textbf{3.5} & \textbf{2.0} &\textbf{14.5}& \textbf{10.8}\\
 &&$15$&\textbf{0.8} & \textbf{0.1 }& \textbf{2.6}  &\textbf{0.8} &\textbf{16.1} &\textbf{11.0}\\
  \cline{2-9}
&& $1$&\textbf{3.0}   &\textbf{3.0}   &\textbf{8.1}   &\textbf{8.1}&n.a. &n.a.\\
&& $2$&\textbf{2.6}  &\textbf{2.4 } &5.3  &5.2& \textbf{16.4} &\textbf{16.4}\\
0.05 &$n=250$& $3$&\textbf{1.9} & \textbf{1.8}  &4.3 & 3.9 &\textbf{16.2} &\textbf{15.6}\\
 &&$6$&\textbf{0.7} & \textbf{0.4} & 4.3  &4.1& \textbf{20.1} &\textbf{18.8}\\
 &&$12$&\textbf{0.6} & \textbf{0.5}  &3.6 & \textbf{2.6}&\textbf{ 24.6}& \textbf{22.4}\\
 &&$15$&\textbf{0.2}  &\textbf{0.2 }& 4.0 & \textbf{2.9} &\textbf{25.7 }&\textbf{22.4}\\ \cline{2-9}
&& $1$&\textbf{3.4}  & \textbf{3.4}  & \textbf{7.2 } & \textbf{7.0}&n.a. &n.a.\\
&& $2$&\textbf{2.0}  &\textbf{2.0} & 6.3 & 6.3 &\textbf{20.4} &\textbf{20.3}\\
0.05 &$n=500$& $3$&\textbf{1.5}  &\textbf{1.5} & 5.1  &5.0 &\textbf{21.1}& \textbf{20.7}\\
 &&$6$&\textbf{0.9 } &\textbf{0.9}  &4.6  &4.6 &\textbf{28.0} &\textbf{27.6}\\
 &&$12$&\textbf{0.4}  &\textbf{0.4} & 4.0  &\textbf{3.2} &\textbf{34.2} &\textbf{32.8}\\
 &&$15$&\textbf{0.1}  &\textbf{0.0}  &\textbf{3.3}  &\textbf{3.0} &\textbf{36.2}& \textbf{34.7}\\
\hline
&& $1$&\textbf{2.8}  & \textbf{2.7}  & 5.3  & 5.0 &n.a. &n.a.\\
&& $2$&\textbf{3.1 } &\textbf{3.1} & 4.9  &4.2 &\textbf{10.9} &\textbf{10.1}\\
0.20 &$n=100$& $3$&\textbf{1.8} &\textbf{1.6} &3.8& \textbf{2.9}& \textbf{9.9} &\textbf{8.3}\\
 &&$6$&\textbf{1.9}  &\textbf{1.1} & \textbf{2.9} & \textbf{2.0} &\textbf{10.8} &\textbf{ 9.0}
\\
 &&$12$&\textbf{0.8}& \textbf{ 0.3} & \textbf{3.1} & \textbf{1.8}& \textbf{13.1} & \textbf{9.7}
\\
 &&$15$&\textbf{0.7}  &\textbf{0.1} & \textbf{2.3 }&\textbf{ 0.7} &\textbf{14.7} & \textbf{9.6}
\\
  \cline{2-9}
&& $1$&\textbf{3.2}  & \textbf{3.2}  & 5.5  & 5.4&n.a. &n.a.\\
&& $2$&\textbf{3.0}  &\textbf{3.0}  &4.3 & 4.2 &\textbf{14.4}& \textbf{14.3}\\
0.20 &$n=250$& $3$& \textbf{2.4} & \textbf{2.3} & 3.6  &\textbf{3.4}& \textbf{14.9} &\textbf{14.2}\\
 &&$6$&\textbf{0.7}  &\textbf{0.7}  &4.3  &3.8& \textbf{18.3}& \textbf{17.3}\\
 &&$12$&\textbf{0.6} & \textbf{0.4} & \textbf{3.5} & \textbf{2.6} &\textbf{23.6}& \textbf{21.2}\\
 &&$15$&\textbf{0.4} & \textbf{0.1} & 3.8 & \textbf{2.5} &\textbf{23.9}& \textbf{21.0}\\
  \cline{2-9}
&& $1$&3.8   &3.8  & 5.3  & 5.3&n.a. &n.a.\\
&& $2$&\textbf{2.4 } &\textbf{2.3 }& 6.1 & 6.1 &\textbf{18.9} &\textbf{18.9}\\
0.20 &$n=500$& $3$&\textbf{1.8 } &\textbf{1.7} & 4.9  &4.6& \textbf{19.9}& \textbf{19.6}\\
 &&$6$&\textbf{0.9} & \textbf{0.9} & 4.4 & 4.3& \textbf{26.5}& \textbf{26.2}\\
 &&$12$&\textbf{0.4 } &\textbf{0.4 } &3.7 &\textbf{ 3.2 }&\textbf{33.5}& \textbf{31.5}\\
 &&$15$&\textbf{0.1} & \textbf{0.1}  &\textbf{3.3 } &\textbf{3.0}& \textbf{35.4} &\textbf{33.8}\\
\hline
&& $1$&\textbf{2.8}  & \textbf{2.6}  & \textbf{8.9}  &\textbf{ 8.3}&n.a. &n.a.\\
&& $2$&\textbf{2.5 }& \textbf{2.2}  &\textbf{6.9}  &\textbf{6.5 }&\textbf{12.1} &\textbf{11.4}\\
0.45 &$n=100$& $3$&\textbf{1.6}  &\textbf{1.5} & 5.0 & 4.1 &\textbf{11.4} &\textbf{10.0}\\
 &&$6$&\textbf{1.6} &\textbf{ 1.2}  &3.4  &\textbf{2.2}& \textbf{10.9} & \textbf{8.4}\\
 &&$12$&\textbf{0.9} & \textbf{0.5} & \textbf{3.2 }& \textbf{1.9}& \textbf{13.5} &\textbf{10.0}\\
 &&$15$&\textbf{0.9}  &\textbf{0.3}  &\textbf{2.2} & \textbf{0.8}& \textbf{14.3} &\textbf{ 9.0}\\
  \cline{2-9}
&& $1$&\textbf{3.3}  & \textbf{3.1} &  \textbf{8.7 }&  \textbf{8.6} &n.a. &n.a.\\
&& $2$&\textbf{3.3}  &\textbf{3.1} & 6.1  &6.1 &\textbf{16.8 }&\textbf{16.2}\\
0.45 &$n=250$& $3$&\textbf{2.6 } &\textbf{2.5} & 4.3  &4.2 &\textbf{15.5}& \textbf{15.1}\\
 &&$6$&\textbf{1.0} & \textbf{0.9} & 4.5  &4.3 &\textbf{19.0} &\textbf{18.0}\\
 &&$12$&\textbf{0.6}  &\textbf{0.4}  &3.9 & \textbf{2.8} &\textbf{23.7}& \textbf{21.8}\\
 &&$15$&\textbf{0.4} & \textbf{0.3}  &3.6 & \textbf{2.5 }&\textbf{24.5} &\textbf{21.6}\\
  \cline{2-9}
&& $1$&3.6  & \textbf{3.5}  & \textbf{6.7}  & \textbf{6.6}&n.a. &n.a.\\
&& $2$&\textbf{2.4}  &\textbf{2.3}  &\textbf{6.9}  &\textbf{6.8}& \textbf{20.0} &\textbf{20.0}\\
0.45 &$n=500$& $3$&\textbf{1.7} & \textbf{1.7 }& 5.4 & 5.2& \textbf{21.3} &\textbf{21.2}\\
 &&$6$&\textbf{1.0}  &\textbf{0.9} & 4.8 & 4.5 &\textbf{26.9} &\textbf{26.4}\\
 &&$12$&\textbf{0.5 } &\textbf{0.4}  &3.7  &\textbf{3.5} &\textbf{33.2} &\textbf{32.0}\\
 &&$15$&\textbf{0.1} & \textbf{0.1}  &\textbf{3.5} & \textbf{3.1} &\textbf{36.3} &\textbf{34.8}\\
\hline\hline
\\
\end{tabular}
}
\end{center}
\label{tab2dn-petit0}
\end{table}
\begin{table}[h]
 \caption{\small{Empirical size (in \%) of the modified and standard versions
 of the LB and BP tests in the case of a weak FARIMA$(0,d_0,0)$  defined by \eqref{process-sim}--\eqref{noise-sim}  with $\theta_0=(0,0,d_0)$.
The nominal asymptotic level of the tests is $\alpha=5\%$.
The number of replications is $N=1,000$. }}
\begin{center}
%\begin{tabular}{lll rrr rrr}
{\scriptsize
\begin{tabular}{c c ccc ccc c}
\hline\hline \\
$d_0$& Length $n$ & Lag $m$ & $\mathrm{{LB}}_{\textsc{sn}}$&$\mathrm{BP}_{\textsc{sn}}$&$\mathrm{{LB}}_{\textsc{w}}$&$\mathrm{BP}_{\textsc{w}}$&$\mathrm{{LB}}_{\textsc{s}}$&$\mathrm{BP}_{\textsc{s}}$
\vspace*{0.2cm}\\ \hline
&& $1$&\textbf{2.2 } & \textbf{2.1}  &\textbf{20.0}  &\textbf{19.5}&n.a. &n.a.\\
&& $2$& \textbf{1.5 } &\textbf{1.5}& \textbf{15.2}& \textbf{14.7}& \textbf{18.3}& \textbf{17.3}\\
0.05 &$n=100$& $3$&\textbf{1.1 }& \textbf{0.9} &\textbf{10.7} &\textbf{10.1}& \textbf{15.3}& \textbf{14.4}\\
 &&$6$&\textbf{0.4} & \textbf{0.2} & 6.0 & 5.2 &\textbf{10.4} &\textbf{ 9.7}\\
 &&$12$&\textbf{0.0}& \textbf{0.0} &\textbf{3.2 }&\textbf{2.5} &\textbf{8.2}& 5.9\\
 &&$15$&\textbf{0.2}& \textbf{0.0} &\textbf{2.4} &\textbf{1.7}& \textbf{7.7 }&5.0\\
  \cline{2-9}
&& $1$&\textbf{3.2}   &\textbf{2.9}&  \textbf{14.4} & \textbf{14.2}&n.a. &n.a.\\
&& $2$&\textbf{3.1}&  \textbf{2.9} &\textbf{10.7} &\textbf{10.6}& \textbf{18.7}& \textbf{18.3}\\
0.05 &$n=250$& $3$&\textbf{1.9}  &\textbf{1.8} & \textbf{7.8}  &\textbf{7.6} &\textbf{16.3} &\textbf{16.0}\\
 &&$6$&\textbf{0.9}  &\textbf{0.6}&  4.5 & 4.2& \textbf{12.6}& \textbf{12.0}\\
 &&$12$&\textbf{0.4 }&\textbf{ 0.3} & \textbf{2.0}&  \textbf{1.5 }&\textbf{10.6} & \textbf{8.8}\\
 &&$15$&\textbf{0.2 }& \textbf{0.2} & \textbf{1.3}&  \textbf{1.3} &\textbf{10.0} & \textbf{8.2}
 \\
  \cline{2-9}
&& $1$&4.3  & 4.3 & \textbf{11.7} & \textbf{11.6}&n.a. &n.a.\\
&& $2$&3.7 & 3.7  &\textbf{8.7} & \textbf{8.6}& \textbf{18.7} &\textbf{18.6}\\
0.05 &$n=500$& $3$&\textbf{2.9}  &\textbf{2.7} & \textbf{6.5} & 6.4 &\textbf{16.7} &\textbf{16.6}\\
 &&$6$& \textbf{1.8} & \textbf{1.6}  &\textbf{3.4}  &\textbf{3.2} &\textbf{14.4}& \textbf{14.1}\\
 &&$12$&\textbf{0.3 }& \textbf{0.2} & \textbf{2.2} & \textbf{1.7}& \textbf{10.9} &\textbf{10.4}\\
 &&$15$&\textbf{0.2} &\textbf{ 0.2}  &\textbf{1.1 }& \textbf{1.0}& \textbf{10.2} & \textbf{9.7}\\

\hline
&& $1$&3.9  & 3.7&  \textbf{11.9}  &\textbf{11.3} &n.a. &n.a.\\
&& $2$&\textbf{1.5} & \textbf{1.5}  &\textbf{7.4}  &\textbf{6.8} &\textbf{12.3}&\textbf{ 11.4}\\
0.20 &$n=100$& $3$&\textbf{1.4 } &\textbf{1.4} & 5.2 & 4.5& \textbf{10.7} &\textbf{ 9.6}\\
 &&$6$&\textbf{0.3} &\textbf{0.2}&\textbf{ 2.3}& \textbf{1.8} &\textbf{8.4}& \textbf{7.6}
\\
 &&$12$&\textbf{0.1} &\textbf{0.0} &\textbf{1.1}& \textbf{0.8} &\textbf{6.5}& 4.2
\\
 &&$15$&\textbf{0.2} &\textbf{0.0 }&\textbf{0.9} &\textbf{0.4} &5.8& \textbf{3.4}
\\
  \cline{2-9}
&& $1$&3.9  & 3.8 &  \textbf{7.1} &  \textbf{6.9}&n.a. &n.a.\\
&& $2$&3.6&  \textbf{3.4} & 6.1 & 5.7 &\textbf{13.2 }&\textbf{13.1}\\
0.20 &$n=250$& $3$& \textbf{1.9 }& \textbf{1.8} & 3.8 & \textbf{3.4}& \textbf{11.7}& \textbf{11.3}\\
 &&$6$&\textbf{0.9} &\textbf{0.6} &\textbf{2.6}& \textbf{2.3}&\textbf{ 9.8}& \textbf{9.3}\\
 &&$12$&\textbf{0.3} &\textbf{0.3} &\textbf{1.0}& \textbf{0.6 }&\textbf{8.8 }&\textbf{7.6}\\
 &&$15$&\textbf{0.2} &\textbf{0.2}& \textbf{0.5}& \textbf{0.5}& \textbf{8.9}& \textbf{7.2}\\
  \cline{2-9}
&& $1$&5.3  & 5.3  & 6.3 &  6.1&n.a. &n.a.\\
&& $2$&4.0 & 3.9 & 5.4  &5.3 &\textbf{15.8} &\textbf{15.6}\\
0.20 &$n=500$& $3$&\textbf{3.3} & \textbf{3.3}&  3.7 & 3.6& \textbf{12.9}& \textbf{12.9}\\
 &&$6$&\textbf{1.9} & \textbf{1.5} & \textbf{1.4}  &\textbf{1.4} &\textbf{11.9}& \textbf{11.5}\\
 &&$12$&\textbf{0.2} &\textbf{0.1}& \textbf{1.2} &\textbf{0.9}& \textbf{9.8}&\textbf{ 9.2}\\
 &&$15$&\textbf{0.3} &\textbf{0.2}& \textbf{0.5}& \textbf{0.5} &\textbf{9.2 }&\textbf{8.9}\\

\hline
&& $1$&3.9  & 3.8 & \textbf{21.5} &\textbf{ 20.2}&n.a. &n.a.\\
&& $2$&\textbf{1.6} & \textbf{1.5 }&\textbf{13.1} &\textbf{11.9} &\textbf{16.5}& \textbf{16.4}\\
0.45 &$n=100$& $3$&\textbf{1.2}  &\textbf{0.9} & \textbf{7.5 }& \textbf{7.2} &\textbf{13.7}& \textbf{12.7}\\
 &&$6$&\textbf{0.7 }& \textbf{0.7} & \textbf{3.1} & \textbf{2.4} &\textbf{10.6} & \textbf{9.2}\\
 &&$12$&\textbf{0.1 }&\textbf{0.0} &\textbf{1.3}& \textbf{0.8} &\textbf{6.9}& 5.2\\
 &&$15$&\textbf{0.2 }&\textbf{0.0}& \textbf{1.3} &\textbf{0.3}& 6.2 &3.8\\
  \cline{2-9}
&& $1$&5.0  & 5.0  &\textbf{15.7} &\textbf{ 15.5} &n.a. &n.a.\\
&& $2$&\textbf{3.0} & \textbf{3.0} &\textbf{10.4} &\textbf{10.0}& \textbf{18.6}& \textbf{18.2}\\
0.45 &$n=250$& $3$&\textbf{2.3}  &\textbf{2.3 }& \textbf{7.5} & \textbf{7.3} &\textbf{16.1} &\textbf{15.9}\\
 &&$6$&\textbf{0.6}  &\textbf{0.4}  &3.6  &3.6 &\textbf{12.1} &\textbf{11.4}\\
 &&$12$&\textbf{0.4}& \textbf{0.3} &\textbf{1.5} &\textbf{1.1} &\textbf{9.7} &\textbf{8.6}
\\
 &&$15$&\textbf{0.2} & \textbf{0.2} & \textbf{1.1} & \textbf{0.8} &\textbf{10.1} & \textbf{8.8}\\
  \cline{2-9}
&& $1$&4.8   &4.8 & \textbf{12.5}  &\textbf{12.5}&n.a. &n.a.\\
&& $2$&4.2 & 4.0 & \textbf{8.9} & \textbf{8.7}& \textbf{19.6}& \textbf{19.5}\\
0.45 &$n=500$& $3$&3.2&  3.2 & 5.7 & 5.6 &16.6& 16.6\\
 &&$6$&\textbf{2.0} & \textbf{1.8} & \textbf{2.6} & \textbf{2.5} &\textbf{13.7} &\textbf{13.4}\\
 &&$12$&\textbf{0.1} & \textbf{0.1} & \textbf{1.5}  &\textbf{1.1} &\textbf{10.8} &\textbf{10.3}\\
 &&$15$&0.3&  0.2 & 0.6 & 0.6 &10.4 &10.1\\
\hline\hline
\\

\end{tabular}
}
\end{center}
\label{tab3dn-petit0}
\end{table}

\begin{table}[h]
 \caption{\small{Empirical size (in \%) of the modified and standard versions
 of the LB and BP tests in the case of a strong FARIMA$(1,d_0,1)$ defined by \eqref{process-sim}
 with $\theta_0=(0.9,0.2,d_0)$.
The nominal asymptotic level of the tests is $\alpha=5\%$.
The number of replications is $N=1,000$. }}
\begin{center}
%\begin{tabular}{lll rrr rrr}
{\scriptsize
\begin{tabular}{c c ccc ccc c}
\hline\hline \\
$d_0$& Length $n$ & Lag $m$ & $\mathrm{{LB}}_{\textsc{sn}}$&$\mathrm{BP}_{\textsc{sn}}$&$\mathrm{{LB}}_{\textsc{w}}$&$\mathrm{BP}_{\textsc{w}}$&$\mathrm{{LB}}_{\textsc{s}}$&$\mathrm{BP}_{\textsc{s}}$
\vspace*{0.2cm}\\\hline
&& $1$&4.7 & 4.4 &{\bf 23.2} &{\bf 22.9}&n.a. &n.a.\\
&& $2$&3.9 &3.6 &{\bf 8.1}& {\bf 7.5}&n.a. &n.a.\\
0.05 &$n=100$& $3$&4.3  & 4.0  & {\bf 6.9}  & 6.1&n.a. &n.a.\\
 &&$6$&4.7 &3.6 &5.4 &3.7 &{\bf 8.4} &5.9\\
 &&$12$&5.1 &{\bf 2.9}& 4.2 &{\bf 2.3}& 5.0 &{\bf 2.6}\\
 &&$15$&6.2 &3.5 &4.9 &{\bf 2.5 }&5.8& {\bf 2.5}\\
  \cline{2-9}
&& $1$&5.3 & 5.3 &{\bf 10.8} &{\bf 10.7}&n.a. &n.a.\\
&& $2$&3.6 &{\bf 3.3} &{\bf 6.8 }&{\bf 6.8}&n.a. &n.a.\\
0.05 &$n=250$& $3$&4.0 &  3.7 &  5.7  & 5.4&n.a. &n.a.\\
 &&$6$&4.2  &3.7 & 5.4 & 5.1 &{\bf 10.6}  &{\bf 9.6}\\
 &&$12$&{\bf 3.1}& {\bf 2.2} &5.3 &4.2 &{\bf 6.5} &5.7\\
 &&$15$&{\bf 3.3}&{\bf  2.5} &5.6& 4.3 &6.4 &5.2\\
  \cline{2-9}
&& $1$&4.6 &4.6 &{\bf 6.9} &{\bf 6.8}&n.a. &n.a.\\
&& $2$&4.3 &4.2 &5.8 &5.6&n.a. &n.a.\\
0.05 &$n=500$& $3$&4.3  & 4.2   &5.7  & 5.5&n.a. &n.a.\\
 &&$6$&5.0  &4.8  &{\bf 6.7}  &{\bf 6.5}&{\bf  11.0}&{\bf  10.7}\\
 &&$12$&4.9 &4.2 &5.5 &4.6 &{\bf 7.1} &6.2\\
 &&$15$&5.6 &4.3 &5.7 &4.5& {\bf 7.1} &6.2\\

\hline
&& $1$&5.1 & 4.8 &{\bf 27.1} &{\bf 25.9} &n.a. &n.a.\\
&& $2$&4.0 &3.8& {\bf 8.7} &{\bf 8.2}&n.a. &n.a.\\
0.20 &$n=100$& $3$&4.1 &  4.0   &{\bf 7.5}  & {\bf 6.9}&n.a. &n.a.\\
 &&$6$&5.5& 3.9& 5.3 &3.9 &{\bf 7.6} &6.2\\
 &&$12$&4.9 &{\bf 3.0} &4.3 &{\bf 2.6} &4.3 &{\bf 2.9}\\
 &&$15$&{\bf 6.9} &{\bf 2.4} &5.1& {\bf 2.9} &5.2& {\bf 2.7}\\
  \cline{2-9}
&& $1$&5.1  &5.0 &{\bf 14.0}& {\bf 13.9}&n.a. &n.a.\\
&& $2$&3.4 &{\bf 3.1}&{\bf 7.3} &{\bf 7.2}&n.a. &n.a.\\
0.20 &$n=250$& $3$&4.3  & 4.1  & 6.2 &  5.9&n.a. &n.a.\\
 &&$6$&4.7 & 4.3 & 6.0 & 5.5& {\bf 10.3} &{\bf  9.8}\\
 &&$12$&3.8& {\bf 2.6} &5.1 &4.3& 5.7 &5.1\\
 &&$15$&3.9 &{\bf 2.8} &5.9& 4.4& 5.7 &5.0\\
  \cline{2-9}
&& $1$&5.6 & 5.6 &{\bf 12.1} &{\bf 12.1}&n.a. &n.a.\\
&& $2$&4.9 &4.9& {\bf 7.0} &{\bf 6.9}&n.a. &n.a.\\
0.20 &$n=500$& $3$&5.0  & 4.9   &{\bf 6.7}  & 6.4&n.a. &n.a.\\
 &&$6$&5.5&  5.2 & 6.2  &5.7 &{\bf 10.1} &{\bf  9.6}\\
 &&$12$&5.6& 4.8& 5.3 &4.6 &6.3 &5.3\\
 &&$15$&5.7 &4.4 &5.4 &4.5 &5.9 &5.1\\

\hline
&& $1$&{\bf 3.2} & \textbf{3.1} &\textbf{32.0} &\textbf{31.6} &n.a. &n.a.\\
&& $2$&\textbf{3.5} &\textbf{3.4}& \textbf{8.3}& \textbf{7.3}&n.a. &n.a.\\
0.45 &$n=100$& $3$&\textbf{2.9}  & \textbf{2.5 }&  \textbf{6.9}  & 6.4&n.a. &n.a.\\
 &&$6$&3.8& \textbf{2.9}& 3.6 &\textbf{2.8} &4.6 &\textbf{3.5}\\
 &&$12$&3.6 &\textbf{1.3}& \textbf{2.7} &\textbf{1.8} &\textbf{2.1} &\textbf{1.2}\\
 &&$15$&4.1 &\textbf{1.9 }&3.7& \textbf{1.5}& \textbf{2.2} &\textbf{0.9}\\
  \cline{2-9}
&& $1$&\textbf{3.4} & \textbf{3.3} &\textbf{18.3}& \textbf{18.0}&n.a. &n.a.\\
&& $2$&\textbf{3.2}& \textbf{3.2} &6.4 &6.1&n.a. &n.a.\\
0.45 &$n=250$& $3$& 3.6  & \textbf{3.4}  & 5.2   &5.1&n.a. &n.a.\\
 &&$6$&3.8 &\textbf{3.3} &4.8& 4.4 &\textbf{7.9} &\textbf{7.3}\\
 &&$12$&\textbf{3.1} &\textbf{2.3}& 4.0 &\textbf{3.2} &4.4 &3.7\\
 &&$15$&\textbf{3.2}& \textbf{2.3} &4.7 &\textbf{3.3} &4.0& \textbf{3.1}\\
  \cline{2-9}
&& $1$&3.6 & 3.6 &\textbf{14.5}& \textbf{14.4}&n.a. &n.a.\\
&& $2$&\textbf{3.4} &\textbf{3.4} &5.3& 5.3&n.a. &n.a.\\
0.45 &$n=500$& $3$&\textbf{3.4}  & \textbf{3.4}  & 5.5 &  5.5&n.a. &n.a.\\
 &&$6$&5.0& 4.7 &4.9& 4.6& \textbf{7.2}& \textbf{7.0}\\
 &&$12$&5.2& 4.7 &4.4 &3.9& 4.2& 4.0\\
 &&$15$&5.0 &4.3& 4.4 &3.6& 4.2& 3.7\\
\hline\hline
\\

\end{tabular}
}
\end{center}
\label{tab1n-petit}
\end{table}

\begin{table}[h]
 \caption{\small{Empirical size (in \%) of the modified and standard versions
 of the LB and BP tests in the case of a weak FARIMA$(1,d_0,1)$
 defined by \eqref{process-sim} with $\theta_0=(0.9,0.2,d_0)$ and where $\omega=0.4$, $\alpha_1=0.3$
and $\beta_1=0.3$ in \eqref{garch}.
The nominal asymptotic level of the tests is $\alpha=5\%$.
The number of replications is $N=1,000$. }}
\begin{center}
%\begin{tabular}{lll rrr rrr}
{\scriptsize
\begin{tabular}{c c ccc ccc c}
\hline\hline \\
$d_0$& Length $n$ & Lag $m$ & $\mathrm{{LB}}_{\textsc{sn}}$&$\mathrm{BP}_{\textsc{sn}}$&$\mathrm{{LB}}_{\textsc{w}}$&$\mathrm{BP}_{\textsc{w}}$&$\mathrm{{LB}}_{\textsc{s}}$&$\mathrm{BP}_{\textsc{s}}$
\vspace*{0.2cm}\\\hline
&& $1$&\textbf{3.1}  &\textbf{3.1} &\textbf{19.7} &\textbf{18.7}&n.a. &n.a.\\
&& $2$& \textbf{2.0} &\textbf{1.7}& \textbf{7.8} &\textbf{7.3}&n.a. &n.a.\\
0.05 &$n=100$& $3$&\textbf{1.7}   &\textbf{1.6}   &\textbf{6.8 } & 6.2&n.a. &n.a.\\
 &&$6$&\textbf{1.4} & \textbf{0.9} & 6.1  &4.7 &\textbf{15.6} &\textbf{12.5}\\
 &&$12$&\textbf{1.5}  &\textbf{0.9}  &5.1  &3.7 &\textbf{13.5} & \textbf{8.9}\\
 &&$15$&\textbf{2.0} & \textbf{1.2 }& 5.1 & \textbf{2.6}& \textbf{13.1} & \textbf{8.9}\\
  \cline{2-9}
&& $1$&\textbf{2.5}  &\textbf{2.4} &\textbf{10.6} &\textbf{10.0}&n.a. &n.a.\\
&& $2$&\textbf{2.1 }&\textbf{1.7}& \textbf{6.6} &6.4&n.a. &n.a.\\
0.05 &$n=250$& $3$&\textbf{1.2}  & \textbf{1.1}   &5.7  & 5.2&n.a. &n.a.\\
 &&$6$&\textbf{0.8}  &\textbf{0.8} & 5.3  &4.7 &\textbf{25.0} &\textbf{24.2}\\
 &&$12$&\textbf{0.8} & \textbf{0.7} & 3.7 & \textbf{3.3}& \textbf{23.5}& \textbf{21.5}\\
 &&$15$&\textbf{1.1}  &\textbf{1.1} & 3.8  &\textbf{3.0} &\textbf{24.7 }&\textbf{21.8}\\
  \cline{2-9}
&& $1$&\textbf{2.4}& \textbf{2.4} &\textbf{8.1} &\textbf{8.1}&n.a. &n.a.\\
&& $2$&\textbf{1.7 }&\textbf{1.7 }&\textbf{7.1} &\textbf{7.0}&n.a. &n.a.\\
0.05 &$n=500$& $3$&\textbf{0.8} &  \textbf{0.7}  & 6.1  & 6.0&n.a. &n.a.\\
 &&$6$&\textbf{0.7 }&\textbf{ 0.6} & 4.6 & 4.2& \textbf{31.5} &\textbf{31.0}\\
 &&$12$&\textbf{1.1}  &\textbf{1.1}  &3.9  &3.8 &\textbf{33.5}& \textbf{32.3}\\
 &&$15$&\textbf{1.0} & \textbf{0.9}  &4.6 & 4.0 &\textbf{35.0} &\textbf{33.4}\\

\hline
&& $1$&\textbf{2.6}  &\textbf{2.6} &\textbf{24.0} &\textbf{23.4} &n.a. &n.a.\\
&& $2$&\textbf{1.7}& \textbf{1.6} &\textbf{9.0} &\textbf{8.4}&n.a. &n.a.\\
0.20 &$n=100$& $3$&\textbf{2.3}  & \textbf{1.7}  & \textbf{6.7}  & 6.2&n.a. &n.a.\\
 &&$6$&\textbf{1.5} & \textbf{0.8}  &5.5  &4.2 &\textbf{15.2} &\textbf{12.3}
\\
 &&$12$&\textbf{1.4 }&\textbf{ 0.6 } &4.5  &\textbf{3.1} &\textbf{12.0} & \textbf{7.7}
\\
 &&$15$&\textbf{2.0} & \textbf{0.8 }& 4.7 & \textbf{2.8} &\textbf{11.2} & \textbf{7.5}\\
  \cline{2-9}
&& $1$&\textbf{3.5}  &\textbf{3.5} &\textbf{17.1} &\textbf{16.8}&n.a. &n.a.\\
&& $2$&\textbf{1.9 }&\textbf{1.9}& \textbf{8.5} &\textbf{8.0}&n.a. &n.a.\\
0.20 &$n=250$& $3$&\textbf{1.1}  & \textbf{1.0}  & 5.5  & 5.0&n.a. &n.a.\\
 &&$6$&\textbf{0.7} & \textbf{0.7 }& 4.3 & 4.1 &\textbf{24.2} &\textbf{23.4}\\
 &&$12$&\textbf{0.6} & \textbf{0.6} & \textbf{3.3} & \textbf{2.9} &\textbf{22.1}& \textbf{19.7}\\
 &&$15$&\textbf{0.6} & \textbf{0.5}  &3.8 & \textbf{3.1}& \textbf{22.9} &\textbf{20.1}\\
  \cline{2-9}
&& $1$&\textbf{2.5}  &\textbf{2.4} &\textbf{12.0} &\textbf{11.8}&n.a. &n.a.\\
&& $2$&\textbf{2.0}& \textbf{2.0}& \textbf{7.7} &\textbf{7.7}&n.a. &n.a.\\
0.20 &$n=500$& $3$&\textbf{1.4}  & \textbf{1.4}  & 6.1 &  5.6&n.a. &n.a.\\
 &&$6$&\textbf{0.8}  &\textbf{0.8} & 4.3 & 4.0 &\textbf{30.2} &\textbf{29.6}\\
 &&$12$&\textbf{0.8} & \textbf{0.7}&  \textbf{3.4}&  \textbf{3.2 }&\textbf{33.2}& \textbf{31.7}\\
 &&$15$&\textbf{0.7} &\textbf{ 0.6} & 4.3  &3.8& \textbf{34.3} &\textbf{32.7}\\

\hline
&& $1$&\textbf{2.4} & \textbf{2.3} &\textbf{33.2} &\textbf{32.9} &n.a. &n.a.\\
&& $2$&\textbf{1.4} &\textbf{1.3 }&\textbf{8.5} &\textbf{7.8}&n.a. &n.a.\\
0.45 &$n=100$& $3$&\textbf{1.5 } & \textbf{1.2}  & 6.3  & 5.4&n.a. &n.a.\\
 &&$6$&\textbf{1.4} &\textbf{ 0.8} & 4.5 & \textbf{3.5}& \textbf{10.5 }& \textbf{8.3}\\
 &&$12$&\textbf{0.8}& \textbf{0.3 }&4.3& \textbf{2.7} &\textbf{7.0}& 5.0\\
 &&$15$&\textbf{1.5}& \textbf{0.4} &4.1 &\textbf{2.4} &\textbf{7.5} &4.3\\
  \cline{2-9}
&& $1$&\textbf{2.1} & \textbf{2.1} &\textbf{20.1 }&\textbf{20.1} &n.a. &n.a.\\
&& $2$&\textbf{1.7}& \textbf{1.7}& 5.9 &5.8&n.a. &n.a.\\
0.45 &$n=250$& $3$&\textbf{1.1} &  \textbf{0.8}  & 5.2  & 4.9&n.a. &n.a.\\
 &&$6$&\textbf{0.9} & \textbf{0.9 }& 4.1  &3.7 &\textbf{18.8} &\textbf{18.0}\\
 &&$12$&\textbf{0.4}&  \textbf{0.4}  &\textbf{2.6} & \textbf{2.1} &\textbf{17.4} &\textbf{15.4}\\
 &&$15$&\textbf{0.2} & \textbf{0.2} & 4.2  &\textbf{3.0} &\textbf{18.4}& \textbf{15.7}\\
  \cline{2-9}
&& $1$&\textbf{2.1} &\textbf{2.1} &\textbf{13.3}& \textbf{13.2}&n.a. &n.a.\\
&& $2$&\textbf{1.2} &\textbf{1.2} &5.8 &5.7&n.a. &n.a.\\
0.45 &$n=500$& $3$&\textbf{1.1} & \textbf{1.0}  & 4.9 &  4.9&n.a. &n.a.\\
 &&$6$&\textbf{0.6} & \textbf{0.6}  &4.0 & 3.8& \textbf{27.3} &\textbf{26.4}\\
 &&$12$&\textbf{0.2} & \textbf{0.2} & \textbf{3.1} & \textbf{2.8}& \textbf{28.3}& \textbf{27.0}\\
 &&$15$&\textbf{0.2} & \textbf{0.1} & 4.3 & 3.8 &\textbf{28.4} &\textbf{26.8}\\
\hline\hline
\\

\end{tabular}
}
\end{center}
\label{tab2n-petit}
\end{table}

\clearpage
\begin{table}[h]
 \caption{\small{Empirical size (in \%) of the modified and standard versions
 of the LB and BP tests in the case of a weak FARIMA$(1,d_0,1)$  defined by \eqref{process-sim}
 with $\theta_0=(0.9,0.2,d_0)$ and where $\omega=0.04$, $\alpha_1=0.12$
and $\beta_1=0.85$ in \eqref{garch}.
The nominal asymptotic level of the tests is $\alpha=5\%$.
The number of replications is $N=1,000$. }}
\begin{center}
%\begin{tabular}{lll rrr rrr}
{\scriptsize
\begin{tabular}{c c ccc ccc c}
\hline\hline \\
$d_0$& Length $n$ & Lag $m$ & $\mathrm{{LB}}_{\textsc{sn}}$&$\mathrm{BP}_{\textsc{sn}}$&$\mathrm{{LB}}_{\textsc{w}}$&$\mathrm{BP}_{\textsc{w}}$&$\mathrm{{LB}}_{\textsc{s}}$&$\mathrm{BP}_{\textsc{s}}$
\vspace*{0.2cm}\\\hline
&& $1$&\textbf{3.1}  &\textbf{3.1} &\textbf{19.7} &\textbf{18.7}&n.a. &n.a.\\
&& $2$& \textbf{2.0} &\textbf{1.7}& \textbf{7.8} &\textbf{7.3}&n.a. &n.a.\\
0.05 &$n=100$& $3$&\textbf{1.7}   &\textbf{1.6}   &\textbf{6.8}  & 6.2&n.a. &n.a.\\
 &&$6$&\textbf{1.4} & \textbf{0.9} & 6.1  &4.7 &\textbf{15.6} &\textbf{12.5}\\
 &&$12$&\textbf{1.5}  &\textbf{0.9}  &5.1  &3.7 &\textbf{13.5} & \textbf{8.9}\\
 &&$15$&\textbf{2.0} & \textbf{1.2} & 5.1 & \textbf{2.6}& \textbf{13.1} & \textbf{8.9}\\
  \cline{2-9}
&& $1$&\textbf{2.5}  &\textbf{2.4 }&\textbf{10.6} &\textbf{10.0}&n.a. &n.a.\\
&& $2$&\textbf{2.1} &\textbf{1.7}& \textbf{6.6} &6.4&n.a. &n.a.\\
0.05 &$n=250$& $3$&\textbf{1.2}  & \textbf{1.1}   &5.7  & 5.2&n.a. &n.a.\\
 &&$6$&\textbf{0.8}  &\textbf{0.8} & 5.3  &4.7 &\textbf{25.0} &\textbf{24.2}\\
 &&$12$&\textbf{0.8} & \textbf{0.7} & 3.7 & \textbf{3.3}& \textbf{23.5}& \textbf{21.5}\\
 &&$15$&\textbf{1.1}  &\textbf{1.1} & 3.8  &\textbf{3.0} &\textbf{24.7} &\textbf{21.8}\\
  \cline{2-9}
&& $1$&\textbf{2.4}& \textbf{2.4} &\textbf{8.1} &\textbf{8.1}&n.a. &n.a.\\
&& $2$&\textbf{1.7} &\textbf{1.7} &\textbf{7.1} &\textbf{7.0}&n.a. &n.a.\\
0.05 &$n=500$& $3$&\textbf{0.8} &  \textbf{0.7}  & 6.1  & 6.0&n.a. &n.a.\\
 &&$6$&\textbf{0.7} & \textbf{0.6} & 4.6 & 4.2& \textbf{31.5} &\textbf{31.0}\\
 &&$12$&\textbf{1.1}  &\textbf{1.1}  &3.9  &3.8 &\textbf{33.5}&\textbf{ 32.3}\\
 &&$15$&\textbf{1.0} & \textbf{0.9}  &4.6 & 4.0 &\textbf{35.0} &\textbf{33.4}\\

\hline
&& $1$&\textbf{2.6}  &\textbf{2.6} &\textbf{24.0} &\textbf{23.4} &n.a. &n.a.\\
&& $2$&\textbf{1.7}& \textbf{1.6} &\textbf{9.0} &\textbf{8.4}&n.a. &n.a.\\
0.20 &$n=100$& $3$&\textbf{2.3}  &\textbf{ 1.7}  & \textbf{6.7}  & 6.2&n.a. &n.a.\\
 &&$6$&\textbf{1.5 }& \textbf{0.8}  &5.5  &4.2 &\textbf{15.2} &\textbf{12.3}
\\
 &&$12$&\textbf{1.4 }& \textbf{0.6}  &4.5  &\textbf{3.1 }&\textbf{12.0} & \textbf{7.7}
\\
 &&$15$&\textbf{2.0 }& \textbf{0.8} & 4.7 & \textbf{2.8} &\textbf{11.2} & \textbf{7.5}\\
  \cline{2-9}
&& $1$&\textbf{3.5}  &\textbf{3.5} &\textbf{17.1} &\textbf{16.8}&n.a. &n.a.\\
&& $2$&\textbf{1.9} &\textbf{1.9}& \textbf{8.5} &\textbf{8.0}&n.a. &n.a.\\
0.20 &$n=250$& $3$&\textbf{1.1}  & \textbf{1.0}  & 5.5  & 5.0&n.a. &n.a.\\
 &&$6$&\textbf{0.7} & \textbf{0.7} & 4.3 & 4.1 &\textbf{24.2} &\textbf{23.4}\\
 &&$12$&\textbf{0.6} & \textbf{0.6} & \textbf{3.3} & \textbf{2.9} &\textbf{22.1}& \textbf{19.7}\\
 &&$15$&\textbf{0.6} & \textbf{0.5}  &3.8 & \textbf{3.1}& \textbf{22.9} &\textbf{20.1}\\
  \cline{2-9}
&& $1$&\textbf{2.5}  &\textbf{2.4} &\textbf{12.0} &\textbf{11.8}&n.a. &n.a.\\
&& $2$&\textbf{2.0}& \textbf{2.0}& \textbf{7.7} &\textbf{7.7}&n.a. &n.a.\\
0.20 &$n=500$& $3$&\textbf{1.4}  & \textbf{1.4}  & 6.1 &  5.6&n.a. &n.a.\\
 &&$6$&\textbf{0.8}  &\textbf{0.8} & 4.3 & 4.0 &\textbf{30.2} &\textbf{29.6}\\
 &&$12$&\textbf{0.8} & \textbf{0.7}&  \textbf{3.4}&  \textbf{3.2} &\textbf{33.2}& \textbf{31.7}\\
 &&$15$&\textbf{0.7} & \textbf{0.6} & 4.3  &3.8& \textbf{34.3} &\textbf{32.7}\\

\hline
&& $1$&\textbf{2.4} & \textbf{2.3} &\textbf{33.2} &\textbf{32.9} &n.a. &n.a.\\
&& $2$&\textbf{1.4} &\textbf{1.3} &\textbf{8.5} &\textbf{7.8}&n.a. &n.a.\\
0.45 &$n=100$& $3$&\textbf{1.5}  & \textbf{1.2}  & 6.3  & 5.4&n.a. &n.a.\\
 &&$6$&\textbf{1.4} & \textbf{0.8} & 4.5 & \textbf{3.5}& \textbf{10.5} &\textbf{ 8.3}\\
 &&$12$&\textbf{0.8}& \textbf{0.3} &4.3&\textbf{ 2.7} &\textbf{7.0}& 5.0\\
 &&$15$&\textbf{1.5}& \textbf{0.4} &4.1 &\textbf{2.4} &\textbf{7.5} &4.3\\
  \cline{2-9}
&& $1$&\textbf{2.1} & \textbf{2.1 }&\textbf{20.1} &\textbf{20.1} &n.a. &n.a.\\
&& $2$&\textbf{1.7}& \textbf{1.7}& 5.9 &5.8&n.a. &n.a.\\
0.45 &$n=250$& $3$&\textbf{1.1} &  \textbf{0.8}  & 5.2  & 4.9&n.a. &n.a.\\
 &&$6$&\textbf{0.9} & \textbf{0.9} & 4.1  &3.7 &\textbf{18.8} &\textbf{18.0}\\
 &&$12$&\textbf{0.4}&  \textbf{0.4}  &\textbf{2.6} &\textbf{ 2.1} &\textbf{17.4} &\textbf{15.4}\\
 &&$15$&\textbf{0.2} & \textbf{0.2} & 4.2  &\textbf{3.0 }&\textbf{18.4}& \textbf{15.7}\\
  \cline{2-9}
&& $1$&\textbf{2.1}  &\textbf{2.1} &\textbf{13.3}& \textbf{13.2}&n.a. &n.a.\\
&& $2$&\textbf{1.2} &\textbf{1.2} &5.8 &5.7&n.a. &n.a.\\
0.45 &$n=500$& $3$&\textbf{1.1}  & \textbf{1.0}  & 4.9 &  4.9&n.a. &n.a.\\
 &&$6$&\textbf{0.6} & \textbf{0.6}  &4.0 & 3.8& \textbf{27.3} &\textbf{26.4}\\
 &&$12$&\textbf{0.2} & \textbf{0.2} & \textbf{3.1} & \textbf{2.8}& \textbf{28.3}& \textbf{27.0}\\
 &&$15$&\textbf{0.2} & \textbf{0.1} & 4.3 & 3.8 &\textbf{28.4} &\textbf{26.8}\\
\hline\hline
\\

\end{tabular}
}
\end{center}
\label{tab2n-petitbis}
\end{table}

\newpage
\clearpage

%\newpage
%\clearpage
\subsection{GARCH process with infinite moment}
In order the see if the test procedures remain reliable for GARCH process
with infinite moment (for $\alpha_1+\beta_1\ge 1$), we replicate the numerical experiments
made on Model \eqref{process-sim}--\eqref{garch} with $\omega=0.04$, $\alpha_1=0.13$ and $\beta_1=0.88$.

\begin{table}[h]
 \caption{\small{Empirical size (in \%) of the modified and standard versions
 of the LB and BP tests in the case of a weak FARIMA$(1,d_0,1)$  defined by \eqref{process-sim}
 with $\theta_0=(0.9,0.2,d_0)$ and where $\omega=0.04$, $\alpha_1=0.12$
and $\beta_1=0.85$ in \eqref{garch}.
The nominal asymptotic level of the tests is $\alpha=5\%$.
The number of replications is $N=1,000$. }}
\begin{center}
%\begin{tabular}{lll rrr rrr}
{\scriptsize
\begin{tabular}{c c ccc ccc c}
\hline\hline \\
$d_0$& Length $n$ & Lag $m$ & $\mathrm{{LB}}_{\textsc{sn}}$&$\mathrm{BP}_{\textsc{sn}}$&$\mathrm{{LB}}_{\textsc{w}}$&$\mathrm{BP}_{\textsc{w}}$&$\mathrm{{LB}}_{\textsc{s}}$&$\mathrm{BP}_{\textsc{s}}$
\vspace*{0.2cm}\\\hline
&& $1$&\textbf{3.3} &\textbf{3.3} &\textbf{8.9}& \textbf{8.9}&n.a. &n.a.\\
&& $2$&\textbf{2.5}& \textbf{2.5}& \textbf{7.5}& \textbf{7.5}&n.a. &n.a.\\
0.05 &$n=1,000$& $3$&\textbf{2.1}  & \textbf{2.1}  & 5.4  & 5.3&n.a. &n.a.\\
 &&$6$&\textbf{1.1} & \textbf{1.0 }& 4.3 & 4.1& \textbf{38.4} &\textbf{38.1}\\
 &&$12$&\textbf{0.6} & \textbf{0.6}  &3.7 & \textbf{3.2} &\textbf{43.3} &\textbf{42.7}\\
 &&$15$&\textbf{0.2}  &\textbf{0.2} & \textbf{3.5} & \textbf{3.5 }&\textbf{45.7} &\textbf{44.6}\\
%  \cline{2-9}
%&& $1$&4.6& 4.6& 6.0& 6.0&n.a. &n.a.\\
%&& $2$&3.8 &3.8 &5.4 &5.4&n.a. &n.a.\\
%0.05 &$n=5,000$& $3$&3.4  & 3.4  & 5.3 &  5.3&n.a. &n.a.\\
% &&$6$&1.4 & 1.4 & 4.5 & 4.5& 53.0& 53.0\\
% &&$12$&0.8 & 0.7 & 4.1 & 4.1 &62.9& 62.6\\
% &&$15$&0.3 & 0.3 & 4.0 & 3.7& 66.2& 66.0\\
  \cline{2-9}
&& $1$&4.4 &4.4 &5.4 &5.4&n.a. &n.a.\\
&& $2$&\textbf{3.5}& \textbf{3.5}& 5.0 &5.0&n.a. &n.a.\\
0.05 &$n=10,000$& $3$&\textbf{3.3}  & \textbf{3.3}  & 6.3  & 6.3&n.a. &n.a.\\
 &&$6$&\textbf{2.3} & \textbf{2.2} & 4.5 & 4.4 &\textbf{57.7}& \textbf{57.6}\\
 &&$12$&\textbf{1.6 }&\textbf{ 1.6} & 4.0  &3.9 &\textbf{68.9} &\textbf{68.8}\\
 &&$15$&\textbf{1.3} &\textbf{ 1.3} & 4.6  &4.6 &\textbf{73.8} &\textbf{73.8}\\
  \cline{2-9}
&& $1$&4.6& 4.6& 5.2 &5.2&n.a. &n.a.\\
&& $2$&4.3 &4.3 &5.0& 5.0&n.a. &n.a.\\
0.05 &$n=20,000$& $3$&3.9  & 3.9   &4.8  & 4.8&n.a. &n.a.\\
 &&$6$&\textbf{2.5}&  \textbf{2.5} & 5.0 & 4.8 &\textbf{41.7} &\textbf{41.7}\\
 &&$12$&3.8  &3.8  &4.2  &4.2& \textbf{51.1} &\textbf{51.1}
\\
 &&$15$&\textbf{3.3} & \textbf{3.2}  &3.7  &3.7 &\textbf{52.8} &\textbf{52.7}\\

\hline
&& $1$&\textbf{3.5} &\textbf{3.4} &\textbf{9.9} &\textbf{9.7} &n.a. &n.a.\\
&& $2$&\textbf{2.4}& \textbf{2.4} &6.4 &6.4&n.a. &n.a.\\
0.20 &$n=1,000$& $3$&\textbf{2.2}  & \textbf{2.1}  & 4.9  & 4.9&n.a. &n.a.\\
 &&$6$&\textbf{1.1}  &\textbf{0.8} & \textbf{3.5 } &\textbf{3.4} &\textbf{37.4} &\textbf{37.2}\\
 &&$12$&\textbf{0.3 }& \textbf{0.3} & \textbf{3.5}  &\textbf{3.3} &\textbf{42.9}& \textbf{42.4}\\
 &&$15$&\textbf{0.0} & \textbf{0.0 }& 3.6  &\textbf{3.5 }&\textbf{44.4} &\textbf{43.2}\\
%  \cline{2-9}
%&& $1$&4.6 &4.6 &5.1 &5.1&n.a. &n.a.\\
%&& $2$&3.6& 3.6 &4.6& 4.6&n.a. &n.a.\\
%0.20 &$n=5,000$& $3$&3.3   &3.3  & 4.8 &  4.8&n.a. &n.a.\\
% &&$6$&1.3&  1.3 & 4.1 & 4.1 &52.1 &51.9\\
% &&$12$&0.6  &0.6 & 4.0  &4.0 &61.2 &61.0\\
% &&$15$&0.5 & 0.4&  3.9  &3.7 &65.1& 65.0\\
  \cline{2-9}
&& $1$&4.2 &4.2 &4.0& 4.0&n.a. &n.a.\\
&& $2$&\textbf{3.4 }&\textbf{3.4} &4.1 &4.1&n.a. &n.a.\\
0.20 &$n=10,000$& $3$&\textbf{3.3 } & \textbf{3.3} &  5.3  & 5.3&n.a. &n.a.\\
 &&$6$&\textbf{2.2} & \textbf{2.2}  &4.3 & 4.3 &\textbf{55.8} &\textbf{55.8}\\
 &&$12$&\textbf{1.6} & \textbf{1.6} & 3.9 & 3.9 &\textbf{67.7} &\textbf{67.7}\\
 &&$15$&\textbf{1.3} & \textbf{1.3} & 4.1 & 4.1& \textbf{72.9} &\textbf{72.9}\\
  \cline{2-9}
&& $1$&5.0 &5.0 &4.3& 4.3&n.a. &n.a.\\
&& $2$&4.6 &4.6 &4.4 &4.4&n.a. &n.a.\\
0.20 &$n=20,000$& $3$&3.9  & 3.9  & 4.7 &  4.7&n.a. &n.a.\\
 &&$6$&\textbf{2.7} & \textbf{2.7} & 4.7 & 4.7 &\textbf{41.0} &\textbf{41.0}\\
 &&$12$&3.7&  3.7 & 4.0 & 4.0 &\textbf{50.3 }&\textbf{50.3}\\
 &&$15$&\textbf{3.4} & \textbf{3.4} & 3.6  &\textbf{3.5} &\textbf{51.9} &\textbf{51.8}\\

\hline
&& $1$&\textbf{3.0} & \textbf{3.0} &\textbf{12.1}& \textbf{12.2} &n.a. &n.a.\\
&& $2$&\textbf{1.8} &\textbf{1.8}& 5.5& 5.4&n.a. &n.a.\\
0.45 &$n=1,000$& $3$&\textbf{1.7}  & \textbf{1.6}  & 4.4  & 4.4&n.a. &n.a.\\
 &&$6$&\textbf{0.6} & \textbf{0.6} & \textbf{3.4} & \textbf{3.3}& \textbf{34.7}& \textbf{34.4}\\
 &&$12$&\textbf{0.4} & \textbf{0.4 } &\textbf{3.3} & \textbf{3.0}& \textbf{38.6} &\textbf{38.0}\\
 &&$15$&\textbf{0.2}  &\textbf{0.2} & 3.6 & \textbf{3.5}& \textbf{40.0} &\textbf{38.9}\\
%  \cline{2-9}
%&& $1$& 4.2& 4.2 &5.4 &5.4&n.a. &n.a.\\
%&& $2$&3.4 &3.4 &4.5& 4.5&n.a. &n.a.\\
%0.45 &$n=5,000$& $3$&2.3  & 2.3 &  5.0 &  5.0&n.a. &n.a.\\
% &&$6$&1.8 & 1.8 & 4.3 & 4.3 &50.6& 50.6\\
% &&$12$&0.6  &0.6  &4.0&  4.0 &60.7 &60.4\\
% &&$15$&0.5  &0.5 & 3.9 & 3.8 &64.6& 64.4\\
  \cline{2-9}
&& $1$&3.7 &3.6 &3.7 &3.7&n.a. &n.a.\\
&& $2$&\textbf{3.0 }&\textbf{3.0}& 3.7& 3.8&n.a. &n.a.\\
0.45 &$n=10,000$& $3$&\textbf{3.0}  & \textbf{3.0}  & 5.1  & 5.1&n.a. &n.a.\\
 &&$6$&\textbf{2.0 }& \textbf{2.0}&  4.7 & 4.7& \textbf{55.3}& \textbf{55.3}\\
 &&$12$&\textbf{1.7} & \textbf{1.7} & 3.8 & 3.8 &\textbf{67.6}& \textbf{67.4}\\
 &&$15$&\textbf{1.3 } &\textbf{1.3}  &4.1 & 4.1& \textbf{72.0} &\textbf{71.8}\\
  \cline{2-9}
&& $1$& 5.0& 5.0 &4.1 &4.1&n.a. &n.a.\\
&& $2$&4.5 &4.5 &4.1 &4.1&n.a. &n.a.\\
0.45 &$n=20,000$& $3$&3.7  & 3.7  & 4.8  & 4.8&n.a. &n.a.\\
 &&$6$&\textbf{2.9} & \textbf{2.9}  &4.5  &4.4& \textbf{40.5} &\textbf{40.5}\\
 &&$12$&3.7 & 3.7  &3.8  &3.8 &\textbf{49.8} &\textbf{49.7}\\
 &&$15$&\textbf{3.5}  &\textbf{3.5} & 3.6  &3.6 &\textbf{51.3}& \textbf{51.3}\\
\hline\hline
\\

\end{tabular}
}
\end{center}
\label{tab2FARIMAbis}
\end{table}

%We also investigate the case where the GARCH model \eqref{garch} have infinite fourth moments. 
As showing in Figures~\ref{fig1},\dots,\ref{fig6} the results are qualitatively similar to what we observe here
in Tables \ref{tabFARIMA00garch},\ref{tabFARIMA00pt} \ref{tab2FARIMA} and  \ref{tab3FARIMA}.

Figures~\ref{fig1},\dots,\ref{fig5} display the residual autocorrelations  of a realization of size $n=2,000$ for weak FARIMA models  \eqref{process-sim}--\eqref{garch} with $\omega=0.04$, $\alpha_1=0.13$,  $\beta_1=0.88$ and three values of $d_0$, and their 5\% significance limits under the strong  FARIMA
and weak  FARIMA assumptions. These figures confirm clearly the conclusions drawn in Subsection \ref{simul}.
The horizontal dotted lines (blue color) correspond to the 5\% significant limits obtained under the strong FARIMA assumption.
The solid lines (red color) and  dashed lines (green color) correspond also  to the 5\% significant limits under the weak FARIMA assumption.
The full lines  correspond to the asymptotic significance limits for the residual autocorrelations
obtained in  Theorem~\ref{loi_res_rho}. The dashed lines (green color) correspond to the self-normalized asymptotic significance limits for the residual autocorrelations
obtained in Theorem~\ref{sn2}.

\clearpage
\begin{figure}[h]
\includegraphics[width=12cm,height=10cm]{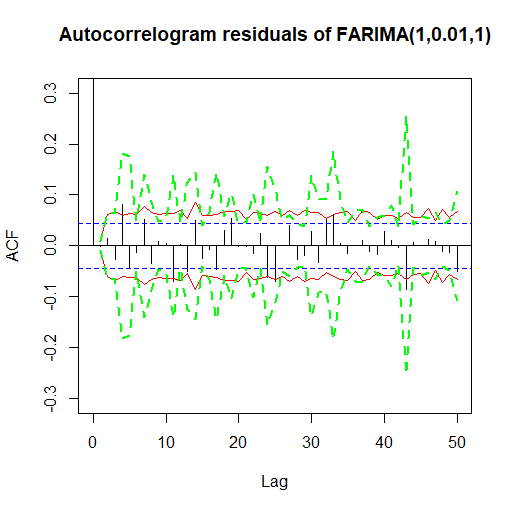}
\caption{Autocorrelation of  a realization of size $n=2,000$ for a weak  FARIMA$(1,0.01,1)$  model \eqref{process-sim}--\eqref{garch} with $\theta_0=(0.9,0.2,0.01)$ and where $\omega=0.04$, $\alpha_1=0.13$ and $\beta_1=0.88$.
The horizontal dotted lines (blue color) correspond to the 5\% significant limits obtained under the strong FARIMA assumption.
The solid lines (red color) and  dashed lines (green color) correspond also  to the 5\% significant limits under the weak FARIMA assumption.
The full lines  correspond to the asymptotic significance limits for the residual autocorrelations
obtained in Theorem~\ref{loi_res_rho}. The dashed lines (green color) correspond to the self-normalized asymptotic significance limits for the residual autocorrelations
obtained in Theorem~\ref{sn2}.}
\label{fig1}
\end{figure}

\begin{figure}[h]
\includegraphics[width=12cm,height=10cm]{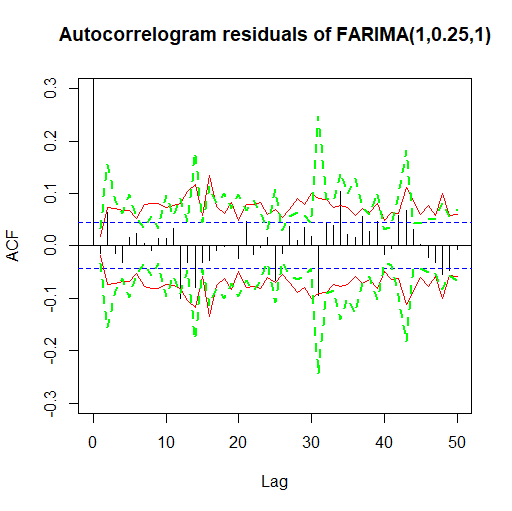}
\caption{Autocorrelation of  a realization of size $n=2,000$ for a weak  FARIMA$(1,0.25,1)$  model \eqref{process-sim}--\eqref{garch} with $\theta_0=(0.9,0.2,0.25)$ and where $\omega=0.04$, $\alpha_1=0.13$ and $\beta_1=0.88$.
The horizontal dotted lines (blue color) correspond to the 5\% significant limits obtained under the strong FARIMA assumption.
The solid lines (red color) and  dashed lines (green color) correspond also  to the 5\% significant limits under the weak FARIMA assumption.
The full lines  correspond to the asymptotic significance limits for the residual autocorrelations
obtained in Theorem~\ref{loi_res_rho}. The dashed lines (green color) correspond to the self-normalized asymptotic significance limits for the residual autocorrelations
obtained in Theorem~\ref{sn2}.}
\label{fig2}
\end{figure}

\begin{figure}[h]
\includegraphics[width=12cm,height=10cm]{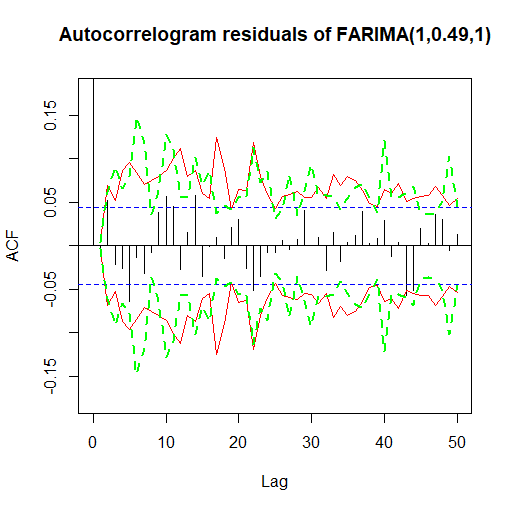}
\caption{Autocorrelation of  a realization of size $n=2,000$ for a weak  FARIMA$(1,0.49,1)$  model \eqref{process-sim}--\eqref{garch} with $\theta_0=(0.9,0.2,0.49)$ and where $\omega=0.04$, $\alpha_1=0.13$ and $\beta_1=0.88$.
The horizontal dotted lines (blue color) correspond to the 5\% significant limits obtained under the strong FARIMA assumption.
The solid lines (red color) and  dashed lines (green color) correspond also  to the 5\% significant limits under the weak FARIMA assumption.
The full lines  correspond to the asymptotic significance limits for the residual autocorrelations
obtained in Theorem~\ref{loi_res_rho}. The dashed lines (green color) correspond to the self-normalized asymptotic significance limits for the residual autocorrelations
obtained in Theorem~\ref{sn2}.}
\label{fig3}
\end{figure}

\begin{figure}[h]
\includegraphics[width=12cm,height=10cm]{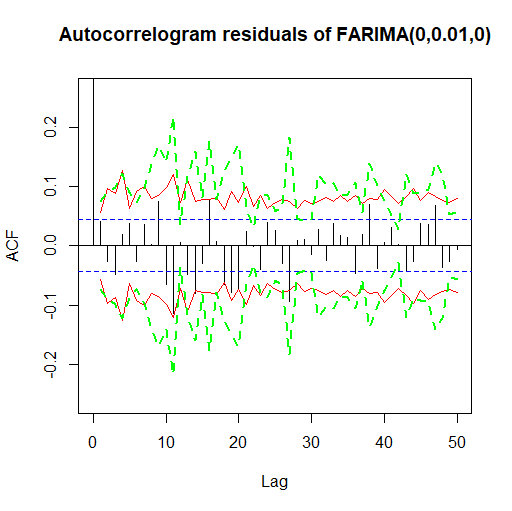}
\caption{Autocorrelation of  a realization of size $n=2,000$ for a weak  FARIMA$(0,0.01,0)$  model \eqref{process-sim}--\eqref{garch} with $\theta_0=(0,0,0.01)$ and where $\omega=0.04$, $\alpha_1=0.13$ and $\beta_1=0.88$.
The horizontal dotted lines (blue color) correspond to the 5\% significant limits obtained under the strong FARIMA assumption.
The solid lines (red color) and  dashed lines (green color) correspond also  to the 5\% significant limits under the weak FARIMA assumption.
The full lines  correspond to the asymptotic significance limits for the residual autocorrelations
obtained in Theorem~\ref{loi_res_rho}. The dashed lines (green color) correspond to the self-normalized asymptotic significance limits for the residual autocorrelations
obtained in Theorem~\ref{sn2}.}
\label{fig4}
\end{figure}

\begin{figure}[h]
\includegraphics[width=12cm,height=10cm]{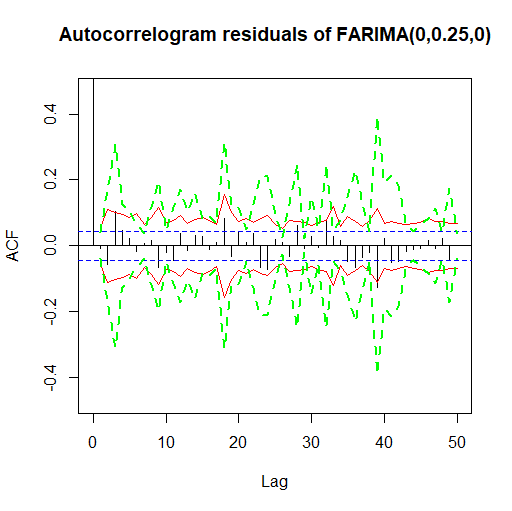}
\caption{Autocorrelation of  a realization of size $n=2,000$ for a weak  FARIMA$(0,0.25,0)$  model \eqref{process-sim}--\eqref{garch} with $\theta_0=(0,0,0.25)$ and where $\omega=0.04$, $\alpha_1=0.13$ and $\beta_1=0.88$.
The horizontal dotted lines (blue color) correspond to the 5\% significant limits obtained under the strong FARIMA assumption.
The solid lines (red color) and  dashed lines (green color) correspond also  to the 5\% significant limits under the weak FARIMA assumption.
The full lines  correspond to the asymptotic significance limits for the residual autocorrelations
obtained in Theorem~\ref{loi_res_rho}. The dashed lines (green color) correspond to the self-normalized asymptotic significance limits for the residual autocorrelations
obtained in Theorem~\ref{sn2}.}
\label{fig5}
\end{figure}

\begin{figure}[h]
\includegraphics[width=12cm,height=10cm]{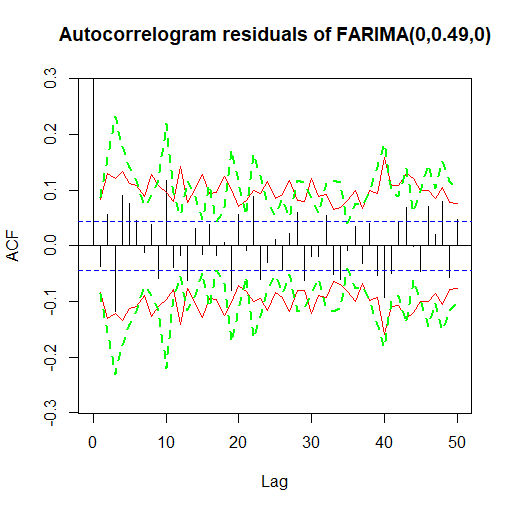}
\caption{Autocorrelation of  a realization of size $n=2,000$ for a weak  FARIMA$(0,0.49,0)$  model \eqref{process-sim}--\eqref{garch} with $\theta_0=(0,0,0.49)$ and where $\omega=0.04$, $\alpha_1=0.13$ and $\beta_1=0.88$.
The horizontal dotted lines (blue color) correspond to the 5\% significant limits obtained under the strong FARIMA assumption.
The solid lines (red color) and  dashed lines (green color) correspond also  to the 5\% significant limits under the weak FARIMA assumption.
The full lines  correspond to the asymptotic significance limits for the residual autocorrelations
obtained in Theorem~\ref{loi_res_rho}. The dashed lines (green color) correspond to the self-normalized asymptotic significance limits for the residual autocorrelations
obtained in Theorem~\ref{sn2}.}
\label{fig6}
\end{figure}
\newpage
\clearpage
\tableofcontents

%\newpage
\end{document}